\documentclass[12pt,letterpaper]{JHEP3}

\usepackage{amscd,amsmath,amssymb,amsfonts,xspace,mathrsfs,amsthm,upgreek}
\usepackage{bbm}
\usepackage[bbgreekl]{mathbbol}

\usepackage{color,mathtools}
\usepackage{url}
\let\normalcolor\relax
\usepackage{latexsym}
\usepackage{qtree}
\usepackage{graphicx}
\usepackage{dsfont}
\usepackage{color}
\usepackage{float}
\usepackage{longtable}

\usepackage{epsf}
\usepackage[latin1]{inputenc}
\usepackage{multirow}
\usepackage{epsfig}

\newtheorem{theorem}{Theorem}
\newtheorem{proposition}{Proposition}
\newtheorem{corollary}{Corollary}

\newtheorem{conj}{Conjecture}

\newtheorem{remark}{Remark}

\numberwithin{equation}{section}
\numberwithin{theorem}{section}
\numberwithin{proposition}{section}
\numberwithin{corollary}{section}
\numberwithin{lemma}{section}
\numberwithin{conj}{section}
\numberwithin{assume}{section}
\numberwithin{remark}{section}

\allowdisplaybreaks

\def\lfig#1#2#3#4#5{
\begin{figure}[t]
 \centerline{\includegraphics[width=#3]{#2}}
 \vspace{#5}
  \caption{#1 \label{#4}}
 \end{figure}
}

\def\bea{\begin{eqnarray}}
\def\eea{\end{eqnarray}}
\def\be{\begin{equation}}
\def\ee{\end{equation}}
\def\ba{\begin{align}}
\def\ea{\end{align}}
\def\bse{\begin{subequations}}
\def\ese{\end{subequations}}

\newcommand{\nn}{\nonumber}

\def\det{\,{\rm det}\, }
\def\diag{{\rm diag}}
\def\sign{{\rm sgn}}

\def\Ch{{\rm Ch}}

\def\Sym{\,{\rm Sym}\, }
\def\Span{\,{\rm Span}\, }

\def\Im{\,{\rm Im}\,}

\newcommand{\sgn}{\mathop{\rm sgn}}

\def\({\left(}
\def\){\right)}
\def\[{\left[}
\def\]{\right]}
\def\<{\left\langle}
\def\>{\right\rangle}
\def\hf{{1\over 2}}

\newcommand{\p}{\partial}

\newcommand{\vth}{\vartheta}

\def\vph{\varphi}

\newcommand{\eps}{\epsilon}
\newcommand{\veps}{\varepsilon}

\renewcommand{\d}{\mathrm{d}}
\newcommand{\de}{\mathrm{d}}

\newcommand{\I}{\mathrm{i}}

\newcommand{\rmR}{\mathrm{R}}

\newcommand{\Asf}{{\sf A}}
\newcommand{\Bsf}{{\sf B}}
\newcommand{\Csf}{{\sf C}}
\newcommand{\asf}{{\sf a}}
\newcommand{\bsf}{{\sf b}}
\newcommand{\csf}{{\sf c}}

\newcommand{\cL}{\mathcal{L}}
\newcommand{\cD}{\mathcal{D}}

\newcommand{\cV}{\mathcal{V}}
\newcommand{\cC}{\mathcal{C}}

\newcommand{\cG}{\mathcal{G}}
\newcommand{\cB}{\mathcal{B}}

\newcommand{\cM}{\mathcal{M}}

\newcommand{\cE}{\mathcal{E}}

\newcommand{\cT}{\mathcal{T}}
\newcommand{\cJ}{\mathcal{J}}

\newcommand{\cI}{\mathcal{I}}
\newcommand{\cO}{\mathcal{O}}

\newcommand{\cA}{\mathcal{A}}

\newcommand{\Zbbm}{\mathbbm{Z}}

\newcommand{\pbbm}{\mathbbm{p}}
\newcommand{\xbbm}{\mathbbm{x}}
\newcommand{\ybbm}{\mathbbm{y}}
\newcommand{\vbbm}{\mathbbm{v}}
\newcommand{\wbbm}{\mathbbm{w}}
\newcommand{\ubbm}{\mathbbm{u}}
\newcommand{\zbbm}{\mathbbm{z}}

\newcommand{\kbbm}{\mathbbm{k}}

\newcommand{\abbm}{\mathbbm{a}}
\newcommand{\bbbm}{\mathbbm{b}}
\newcommand{\ebbm}{\mathbbm{e}}

\def\bbLambda{\mathbb{\Lambda}}

\newcommand{\Fb}{{\mathbb F}}

\newcommand{\IT}{\mathds{T}}
\newcommand{\IR}{\mathds{R}}
\newcommand{\IC}{\mathds{C}}
\newcommand{\IZ}{\mathds{Z}}

\newcommand{\IN}{\mathds{N}}
\newcommand{\IH}{\mathds{H}}

\newcommand{\IP}{\mathds{P}}
\newcommand{\IA}{\mathds{A}}
\newcommand{\ID}{\mathds{D}}

\def\scR{\mathscr{R}}

\def\Cv{\mathscr{C}}

\def\Ev{\mathscr{E}}
\def\Fv{\mathscr{F}}

\def\Zv{\mathscr{Z}}

\def\frm{\mathfrak{m}}
\def\frr{\mathfrak{r}}

\def\frt{\mathfrak{t}}
\def\frz{\mathfrak{z}}

\def\frR{\mathfrak{R}}

\newcommand{\bfg}{{\bf g}}

\newcommand{\bfa}{{\boldsymbol a}}
\newcommand{\bfb}{{\boldsymbol b}}

\newcommand{\bfv}{{\boldsymbol v}}

\newcommand{\bfk}{{\boldsymbol k}}

\newcommand{\bfn}{{\boldsymbol n}}

\newcommand{\bfq}{{\boldsymbol q}}
\newcommand{\bfr}{{\boldsymbol r}}
\newcommand{\bfs}{{\boldsymbol s}}

\newcommand{\bfz}{{\boldsymbol z}}
\newcommand{\bfx}{{\boldsymbol x}}
\newcommand{\bfy}{{\boldsymbol y}}

\newcommand{\bfD}{{\boldsymbol D}}

\newcommand{\bfrho}{{\boldsymbol \rho}}

\newcommand{\bftet}{{\boldsymbol \theta}}
\newcommand{\bfmu}{{\boldsymbol \mu}}
\newcommand{\bfnu}{{\boldsymbol \nu}}
\newcommand{\bftnu}{\tilde{\boldsymbol \nu}}

\newcommand{\bfLam}{{\boldsymbol \Lambda}}

\newcommand{\bflam}{{\boldsymbol \lambda}}
\newcommand{\bfgam}{{\boldsymbol \gamma}}

\newcommand{\tlh}{\tilde h}

\newcommand{\tell}{\tilde\ell}
\newcommand{\tmu}{\tilde\mu}
\newcommand{\tnu}{\tilde \nu}

\newcommand{\bftmu}{\tilde\bfmu}

\newcommand{\bftlam}{\tilde\bflam}

\def\tbsf{\tilde\bsf}

\def\tfrm{\tilde\frm}

\def\ba{\bar a}

\def\by{\bar y}

\def\bw{\bar w}

\def\btau{\bar \tau}

\def\bOm{\bar\Omega}
\def\bqq{\bar\q}

\def\hs{\hat s}

\def\hmu{\hat\mu}
\def\hrho{\hat\rho}

\def\hgam{\hat\gamma}
\def\bfhv{\hat\bfv}
\def\bfhgam{\hat\bfgam}
\def\bfhmu{\hat\bfmu}

\def\bfhtmu{\,\hat{\tilde{\!\bfmu}}}
\def\hubbm{\hat\ubbm}
\def\hvbbm{\hat\vbbm}
\def\hwbbm{\hat\wbbm}

\def\chg{\check g}
\def\chG{\check G}

\def\chphi{\check\phi}

\def\gama{\check\gamma}

\newcommand{\glueg}{\mathbbm{g}}

\def\gluegi#1{\glueg^{(#1)}}
\def\bfgi#1{{\bf g}^{(#1)}}

\def\CY{\mathfrak{Y}}

\def\ver{v}

\def\hi#1{h^{(#1)}}

\def\han{h^{\rm (an)}}
\def\hh{h^{(0)}}
\def\thh{\tlh^{(0)}}
\def\than{\tlh^{\rm (an)}}
\def\thh{\tlh^{(0)}}

\def\ths#1{\theta^{(#1)}}
\def\vths#1{\vartheta^{(#1)}}
\def\vthls#1{\vartheta^{(#1)||}}
\def\tvths#1{\tilde\vartheta^{(#1)}}
\def\tvthls#1{\tilde\vartheta^{(#1)||}}

\def\vthps#1{\vartheta^{(#1)\perp}}
\def\vthpzi#1#2{\vartheta^{(#1)#2}}
\def\vthA#1{\varTheta^{(#1)}}

\def\di#1{d_{\Nr_#1}}

\def\gi#1{g^{(#1)}}

\def\girf#1{g^{(#1)\rm{ref}}}

\def\chgirf#1{\chg^{(#1)\rm{ref}}}
\def\vwgi#1{\mathfrak{g}_{#1}}

\def\whvwg{\widehat{\mathfrak{g}}}

\def\Er{\Ev^{\rm (ref)}}

\def\Ef{\Ev^{(0)}}
\def\Efrf{\Ev^{(0){\rm ref}}}

\def\trmRi#1{\tilde \rmR^{(#1)}}

\def\trmRVWi#1{\tilde \frR^{(#1)}}

\def\rmRi#1{\rmR^{(#1)}}

\def\rmRirf#1{\rmR^{(#1)\rm ref}}

\def\scRrf{\scR^{\rm ref}}

\def\whgirf#1{\widehat g^{(#1)\rm{ref}}}

\def\whchgirf#1{\lefteqn{\widehat \chg}\hphantom{\chg}^{(#1)\rm{ref}}}

\def\vwh{h^{\rm VW}}

\def\whh{\widehat h}
\def\twhh{\widehat{\tilde h}}

\def\whg{\widehat g}
\def\whgi#1{\widehat g^{(#1)}}

\def\whG{\widehat G}

\def\bfLami#1{\bfLam^{(#1)}}
\def\tbfLami#1{\tilde\bfLam^{(#1)}}

\def\bbLami#1{\bbLambda^{(#1)}}

\def\cvths#1{\lefteqn{\smash{\mathop{\vphantom{<}}\limits^{\;\circ}}}\vartheta^{(#1)}}

\def\Gi#1{G^{(#1)}}
\def\cGi#1{\cG^{(#1)}}
\def\Mi#1{M^{(#1)}}

\def\hcTi#1{\hat\cT^{(#1)}}

\def\whGi#1{\whG^{(#1)}}

\def\EPhi{\Phi^{\,\Ev}}

\def\rPhi{\Phi^{(\bfr)}}

\def\PhiR{\Phi_{\rm R}}
\def\Phid{\Phi_\delta}
\def\rPhiR{\PhiR^{(\bfr)}}

\def\Phii#1{\Phi^{(#1)}}
\def\PhiRi#1{\PhiR^{(#1)}}
\def\Phidi#1{\Phid^{(#1)}}
\def\Fvi#1{\Fv^{(#1)}}
\def\whFvi#1{\widehat\Fv^{(#1)}}

\def\vu{\mathfrak{u}}

\def\cEprf{\cE^{(+){\rm ref}}}

\def\Ep{\Ev^{(+)}}
\def\Eprf{\Ev^{(+){\rm ref}}}

\def\cbfr{c_\bfr}
\def\ci#1{c_{#1}}
\def\bfDr{\bfD^{(\bfr)}}
\def\bfDi#1{\bfD^{(#1)}}

\def\Cvr{\Cv^{(\bfr)}}
\def\hCvr{\hat\Cv^{(\bfr)}}

\def\cCr{\cC^{(\bfr)}}

\def\cAr{\cA^{(\bfr)}}
\def\cIr{\cI^{(\bfr)}}
\def\cIi#1{\cI^{(#1)}}
\def\Dzm{\Delta_0^{(\bfr)}}

\def\nDzm{\overline{\Delta}_0^{\,(\bfr)}}

\def\vnc#1#2{#1^{[#2]}}

\newcommand{\q}{\mbox{q}}

\def\om{\omega}
\def\omb{\omega^{(\beta)}}

\def\Nr{r}

\def\rdcr{\hat r}                  
\def\sss#1#2{\Nr_{#1::#2}}

\def\sl{a}                    
\def\Ms{s}

\newcommand{\symfootnote}[1]{%
\let\oldthefootnote=\thefootnote%
\setcounter{mpfootnote}{2}%
\addtocounter{footnote}{-1}%
\renewcommand{\thefootnote}{\fnsymbol{mpfootnote}}%
\footnote{#1}%
\let\thefootnote=\oldthefootnote%
}

\title{Mock modularity of Calabi-Yau threefolds}

\author{Sergei Alexandrov and Khalil Bendriss
\\
{\it Laboratoire Charles Coulomb (L2C), Universit\'e de Montpellier,
CNRS, \\ F-34095, Montpellier, France}\\

\vspace*{2mm} {\tt e-mail:
\email{sergey.alexandrov@umontpellier.fr},
\email{khalil.bendriss@umontpellier.fr}
}

\vspace*{-3mm}

}

\abstract{Generating functions $h_r(\tau)$ of D4-D2-D0 BPS indices, appearing in Calabi-Yau compactifications
of type IIA string theory and identical to rank 0 Donaldson-Thomas invariants, are known 
to be higher depth mock modular forms satisfying a specific modular anomaly equation, 
with depth determined by the D4-brane charge $r$.
We develop a method to solve the anomaly equation for arbitrary charges, 
in terms of indefinite theta series.
This allows us to find the generating functions up to modular forms that can be fixed by 
computing just a finite number of Fourier coefficients of $h_r$.
}

\begin{document}

\setlength{\parskip}{0.15cm}

\section{Introduction}

The indices counting BPS states in compactifications of type II strings on Calabi-Yau (CY) threefolds $\CY$
play a prominent role both in physics and mathematics.
On the physics side, they represent degeneracies of BPS black holes and encode weights 
of instanton corrections to the low energy effective action.
On the mathematics side, they coincide with the generalized Donaldson-Thomas (DT) invariants 
whose importance for understanding geometry of the CY threefolds can hardly be overestimated.

For non-compact CYs, there are various techniques to compute these BPS indices, which are based on localization, 
quivers, spectral networks and their generalizations, relations to topological and gauge theories, etc., see e.g.
\cite{Douglas:2000qw,Nishinaka:2013mba,Banerjee:2019apt,Beaujard:2020sgs,Descombes:2021snc,DelMonte:2021ytz}.
However, for compact threefolds most of these techniques cannot be applied and the problem becomes much more complicated.

There are actually two classes of BPS indices which, at least in principle, can be systematically calculated.
First, for D6-brane charge equal to $\pm 1$, the BPS indices (at large volume) coincide with the ordinary 
DT (respectively, PT (due to Pandharipande-Thomas)) invariants. Their generating function is given by 
the famous MNOP formula \cite{gw-dt,gw-dt2} in terms of Gopakumar-Vafa (GV) invariants, 
which in turn can be found by
computing the topological string free energy, for example, by the direct integration method 
\cite{Bershadsky:1993ta,Huang:2006hq,Grimm:2007tm}.

Second, for vanishing D6-brane charge, the BPS indices, known also as rank 0 DT invariants, 
count D4-D2-D0 BPS states and can be organized in generating functions $h_p(\tau)$ where 
$p^a$ (with $a=1,\dots,b_2(\CY)$\label{pb2}) is the D4-brane charge which geometrically corresponds to a divisor 
$\cD_p=p^a\gamma_a$ of $\CY$
where $\gamma_a$ is a basis of $H_4(\CY,\IZ)$.\footnote{In fact, the generating functions 
are vector valued so that their components $h_{p,\mu}(\tau)$ are labeled by residue class $\mu_a$
taking values in the discriminant group $\Lambda^\star/\Lambda$ where $\Lambda=H_4(\CY,\IZ)$. For simplicity of exposition,
we drop the vector index in the Introduction. \label{foot-mu}} 
These functions turn out to possess nice modular properties 
\cite{Maldacena:1997de,Alexandrov:2012au,Alexandrov:2016tnf,Alexandrov:2018lgp}
which severely restrict $h_p(\tau)$ and, again at least in principle, 
can be used to fix them up to a finite number of coefficients.

The precise modular properties of $h_p(\tau)$ strongly depend on properties of the divisor $\cD_p$.
If the divisor is {\it irreducible}, the generating function must be 
a weakly holomorphic modular form of weight $-\hf b_2-1$
\cite{Maldacena:1997de}, i.e. it has the expansion 
\be 
h_p(\tau)=\sum_{n\geq n_{\rm min}} c_n \q^n, 
\qquad
\q=e^{2\pi\I \tau},
\label{exp-hp}
\ee
where $n_{\rm min}<0$, and transforms in the usual way under the standard $SL(2,\IZ)$ 
transformations acting on $\tau$. The space of such modular forms is finite dimensional 
and its dimension is bounded from above by the number of {\it polar} terms, i.e. terms with $n<0$.
This is why in this case it is enough to compute only the polar terms in \eqref{exp-hp} to completely fix 
$h_p(\tau)$.
This idea was applied long ago to a few one-parameters CY threefolds in 
\cite{Gaiotto:2006wm,Gaiotto:2007cd,Collinucci:2008ht,VanHerck:2009ww} and revised recently in 
\cite{Alexandrov:2022pgd,Alexandrov:2023zjb}.
In particular, in \cite{Alexandrov:2023zjb} a systematic way to compute first terms in the expansion 
\eqref{exp-hp} has been suggested which is based on 
new wall-crossing relations between PT and rank 0 DT invariants \cite{Feyzbakhsh:2022ydn}.
Combined with the MNOP formula, they allow to express D4-D2-D0 BPS indices in terms of GV invariants
so that, if the latter are known up to sufficiently high genus, all polar (and many non-polar) terms 
can be computed.

If the divisor $\cD_p$ is {\it reducible}, i.e. $p^a=\sum_{i=1}^r p_i^a$ with positive $p_i^a$ and $r>1$, 
the modular properties of $h_p(\tau)$ are more involved. 
It was shown in \cite{Alexandrov:2016tnf,Alexandrov:2018lgp} that the generating functions 
are {\it mock} modular forms of depth $r-1$ with a specific modular anomaly. 
A convenient way to characterize the anomaly is to consider a {\it modular completion} that is
a non-holomorphic function $\whh_p(\tau,\btau)$ that transforms as a usual modular form 
and differs from $h_p(\tau)$ only by terms suppressed in the limit $\Im\tau\to \infty$.
An exact expression for the completion is given below in section \ref{subsec-anom} (see \eqref{exp-whh})
in a simplified form found recently in \cite{Alexandrov:2024jnu}. 
An important feature of this formula is that $\whh_p(\tau,\btau)$ is determined 
by the generating functions $h_{p_i}(\tau)$ of the constituents.

Although for mock modular forms the polar terms alone are not sufficient anymore to 
fix the function uniquely, the missing information can be recovered from the modular anomaly.
Namely, one can follow the two-step strategy. 
First, one finds {\it any} mock modular form $\han_{p}(\tau)$
having the given modular anomaly. Obviously, the generating function $h_p$ can differ from $\han_{p}$
at most by a modular form $\hh_{p}$, i.e. 
\be
h_{p}=\han_{p}+\hh_{p}.
\label{hhh}
\ee
Given this representation, at the second step, 
the modular ambiguity $\hh_{p}$ can be fixed in the usual way 
by computing its polar terms given by the difference of the polar terms of $h_{p}$ and $\han_{p}$.

For one-parameter CYs with the triple intersection number equal to a power of a prime number
and D4-brane charge $r=2$,\footnote{For one-parameter CYs, the D4-brane charge $p^1$
coincides with the degree of reducibility of the divisor $\cD_p$ and therefore
will be denoted by $r$ in the rest of this paper.} 
the first step (solution of the modular anomaly) has been realized in 
\cite{Alexandrov:2022pgd}. Then for two CYs known as decantic $X_{10}$ and octic $X_8$, 
the second step (computing the polar terms) has been done in  
\cite{Alexandrov:2023ltz}, which resulted in explicit mock modular generating functions $h_2$
for this pair of threefolds.

The goal of this paper is, still restricting to the one-parameter case $b_2=1$, 
to find a solution of the modular anomaly, i.e. the functions $\han_{r}$, for higher charges. 
Thus, we reduce the problem of finding the generating functions $h_r$
to just the problem of computing their polar terms. This last problem is left for future research.

The immediate question which arises when one solves the modular anomaly for $h_r$ is how 
this can be done given that the anomaly depends on the generating functions of lower charges 
that remain unknown because their polar terms are not fixed yet?
To address this issue, we disentangle the anomalous parts of all generating functions from 
their modular ambiguities fixed by the polar terms. Namely, we express each $h_r$ as a polynomial 
in $\hh_{r_i}$ with $r_i\leq r$ (see \eqref{genansatz}) and show that the coefficients $\gi{\bfr}$, 
where $\bfr=(r_1,\dots,r_n)$ such that $\sum_{i=1}^n r_i=r$, 
are themselves mock modular forms of depth $n-1$  
satisfying an appropriate anomaly equation \eqref{compl-gi}.
Thus, the problem of solving the modular anomaly for $h_r$ is reformulated
as the problem of solving the modular anomaly equations for the holomorphic functions $\gi{\bfr}(\tau)$
parametrized by $n$ charges $r_i$. We call these functions {\it anomalous coefficients}.

It turns out that it is relatively easy to give a solution for two infinite families 
of the anomalous coefficients. First, in the $n=2$ case with arbitrary $r_1$ and $r_2$, 
the anomaly is characterized by a simple theta series depending 
on a single combination of all parameters which we denote by $\kappa_{12}$. 
A partial solution for such $\gi{r_1,r_2}$ (when $\kappa_{12}$ is a power of a prime number)
has already been given in \cite{Alexandrov:2022pgd}. But, in fact, a solution for generic $\kappa_{12}$
is also known and provided by {\it mock modular forms of optimal growth}
introduced in \cite{Dabholkar:2012nd}.
They are constructed by applying certain Hecke-like operators to a set of ``seed" mock modular functions
$\cGi{d}$ defined for each $d$ that is a square-free positive integer with an even number of prime factors.
In particular, $\cGi{1}$ coincides with the generating series of Hurwitz class numbers, 
which is also known to be the normalized generating function of $SU(2)$ Vafa-Witten (VW) invariants on $\IP^2$
\cite{Vafa:1994tf},
consistently with the results of \cite{Alexandrov:2022pgd}. 
In fact, this is a particular case of the second family of solutions.
Namely, we show that the anomaly equations for $\gi{\bfr}$ with all $r_i=1$ 
form a closed system which, for the intersection number of $\CY$ equal to one, 
coincides with a similar system of anomaly equations 
for the normalized generating functions of $SU(n)$ VW invariants on $\IP^2$,
with $n$ equal to the number of charges.
Thus, these two sets of functions can be simply identified.
Since the generating functions of VW invariants on $\IP^2$ are by now well-known for any rank 
of the gauge group \cite{Manschot:2014cca,Alexandrov:2020dyy}
(see also \cite{Yoshioka:1994,Manschot:2010nc,Manschot:2011ym,Beaujard:2020sgs}),
this identification provides a solution for the subset of anomalous coefficients.

Unfortunately, neither of these solutions seems to have a simple generalization 
to other cases. Therefore, we follow an alternative strategy which is meant to work 
for an arbitrary set of charges and is based on the use of indefinite theta series.
This however requires two preliminary steps. First, the system of anomaly equations
should be extended to include a {\it refinement} encoded by an elliptic parameter $z$. 
This allows to simplify both the equations and their solution, 
but most importantly it provides a regularization 
of certain singularities which would otherwise plague the theta series. 
Second, one should artificially extend the relevant charge lattice
(which is achieved by multiplication by an appropriate combination of Jacobi theta series)
to ensure that the resulting lattice possesses a set of null vectors necessary 
to write down a general solution. Such solution is then given by a combination of 
indefinite theta series and holomorphic modular functions $\phi^{(\bfr)}(\tau,z)$
(see Theorem \ref{thm-gensol}) which ensure that the unrefined limit $z\to 0$ is non-singular.
In fact, it is the proper choice of these functions and the explicit evaluation of the unrefined limit
that are the most non-trivial elements of our construction.

In this paper we perform the construction in detail and derive the final form 
of $\gi{\bfr}$ for the cases of two and three charges, 
while in generic case we find the general form of the refined solution, 
obtain the functions $\phi^{(\bfr)}$ ensuring the existence of the unrefined limit,
but leave the limit itself non-evaluated since it appears quite hard to do this analytically.
Besides, we check that the solutions based on the indefinite theta series are consistent
with the ones obtained by Hecke-like operators and from VW theory. 

The organization of the paper is as follows. In the next section we recall properties
of the generating series of D4-D2-D0 BPS indices, including their behavior under modular transformations.
In section \ref{sec-ansatz} we introduce the anomalous coefficients that disentangle 
the mock modular parts of the generating series, which are fixed by the modular anomaly equations,
from their modular ambiguities fixed by computing the polar terms.  
In section \ref{sec-2sol} we establish relations to the mock modular forms of optimal growth introduced 
in \cite{Dabholkar:2012nd} and to the normalized generating functions 
of VW invariants on $\IP^2$.
In section \ref{sec-sol-indef} we present our main construction 
of the anomalous coefficients in terms of indefinite theta series. 
Finally, in section \ref{sec-concl} we discuss our results and their possible extensions.
Several appendices contain some useful information on various building blocks of the construction,
details of our calculations and some explicit $\q$-series.
For the reader's convenience, the last appendix \ref{sec_index} includes an index of notations.

\section{BPS indices and their modular anomaly}
\label{sec-DT}

In this section we introduce our main objects of interest, the generating series
of D4-D2-D0 BPS indices, and describe their main properties.
We restrict ourselves from the very beginning to the case of $b_2(\CY)=1$.
A more complete discussion of BPS indices can be found, e.g., in \cite{Alexandrov:2023zjb}.

\subsection{D4-D2-D0 BPS indices and their generating series}
\label{subsec-MSWgen}

In type IIA string theory compactified on a CY threefold $\CY$, BPS indices depend 
on the electromagnetic charge $\gamma$ \label{pgamma} which labels elements in the even cohomology of $\CY$ 
and can be represented by a vector $(p^0,r,q,q_0)$ where different components
correspond to D6, D4, D2 and D0-charges, respectively. In this paper we are interested 
in D4-D2-D0 BPS states for which $p^0=0$, which are also known as rank 0 DT invariants. 

In general, BPS indices also depend in a piece-wise linear way on the K\"ahler modulus 
due to the wall-crossing phenomenon and 
thus they take different values in different chambers of the moduli space.
Here we take them to be evaluated in the large volume attractor chamber  \cite{deBoer:2008fk} 
containing the point
\be
\label{lvolatt}
z^a_\infty(\gamma)
= \lim_{\lambda\to +\infty}\(-\frac{q}{\kappa r}+\I\lambda r\),
\ee
where $\kappa$ is the intersection number of $\CY$. In this chamber the BPS indices are invariant
under the so-called spectral flow transformation acting on the charge vector.
It gives rise to the following decomposition of the D2-brane charge
\be
\label{defmu}
q = \mu + \frac12\, \kappa r^2 + \kappa r \eps,
\ee
where $\eps\in \IZ$ is the parameter shifted by the spectral flow, while 
$\mu\in \Lambda^*/\Lambda$ with $\Lambda=\kappa r \IZ$ is the so-called
{\it residue class} taking $\kappa r$ values and staying invariant.
Thus, D4-D2-D0 BPS indices depend only on $r$, $\mu$ and an invariant combination of D2 and D0-charges
\be
\label{defqhat}
\hat q_0 \equiv
q_0 -\frac{q^2}{2\kappa r} \, ,
\ee
and will be denoted by $\Omega_{r,\mu}( \hat q_0)$.

An important fact, known as the Bogomolov-Gieseker bound \cite{bayer2011bridgeland},  
is that $\Omega_{r,\mu}( \hat q_0)$ vanishes
unless the invariant charge $\hat q_0$ satisfies
\be
\hat q_0 \leq \hat q_0^{\rm max}=\frac{1}{24}\,(\kappa r^3+c_2 r),
\label{qmax}
\ee
where $c_2$ is the second Chern class of $\CY$.
It allows to define the generating series 
\be
h_{r,\mu}(\tau) =\sum_{\hat q_0 \leq \hat q_0^{\rm max}}
\bOm_{r,\mu}(\hat q_0)\,\q^{-\hat q_0 }\, ,
\label{defhDT}
\ee
where $\q=e^{2\pi\I \tau}$ and the bar denotes {\it rational} BPS indices defined for any charge as
$\bar\Omega(\gamma) = \sum_{d|\gamma} \Omega(\gamma/d)/d^2$.
Only generating series of rational BPS indices are expected to possess nice modular properties 
\cite{Manschot:2010xp}. Although it does not lead to conceptual simplifications, 
sometimes it is useful to use the symmetry\footnote{Mathematically, this symmetry follows from 
dualization of the coherent sheaf induced by the D4-brane.} 
$h_{r,-\mu}=h_{r,\mu}$ where we extended 
the range of $\mu$ from $[0,\kappa r-1]$ by periodicity.

\subsection{Modular symmetry}
\label{subsec-anom}

The most important feature of the generating series $h_{r,\mu}(\tau)$ is that
they transform as depth $r-1$ vector valued (VV) mock modular forms
under the standard $SL(2,\IZ)$ transformations $\tau\mapsto \frac{a\tau+b}{c\tau+d}$.
While a mathematical proof of this modular behavior is still absent, 
it was derived using duality symmetries of string theory 
\cite{Maldacena:1997de,Alexandrov:2012au,Alexandrov:2016tnf,Alexandrov:2018lgp} 
and obtained recently a striking confirmation by verifying predictions of modularity against 
a direct calculation of DT invariants \cite{Alexandrov:2023zjb,Alexandrov:2023ltz}.

More precisely, in the simplest case $r=1$,  $h_{1,\mu}$ is expected to be 
a weakly holomorphic VV modular form of weight $-3/2$
with the multiplier system closely related to the Weil representation attached to the lattice
$\kappa r\IZ$ with quadratic form $Q(x)=\kappa r x^2$
and determined by the following two matrices for T and S-transformations
\cite[Eq.(2.10)]{Alexandrov:2019rth} 
(see also \cite{Gaiotto:2006wm, deBoer:2006vg, Denef:2007vg, Manschot:2007ha,Alexandrov:2012au})
\be
\begin{split}
	\Mi{h_r}_{\mu\nu}(T)=&\, 
	e^{\frac{\pi \I}{\kappa \Nr}\(\mu+\frac{1}{2}\kappa r^2\)^2+\tfrac{\pi\I}{12}\, c_{2}r}\,\delta_{\mu\nu},
	\\
	\Mi{h_r}_{\mu\nu}(S)=&\, \frac{(-1)^{\chi_r}}{\sqrt{\kappa r}}\,
	e^{-\frac{\pi\I}{4}-2\pi\I \, \frac{\mu\nu}{\kappa r}}\,,
\end{split}
\label{Multsys-hp}
\ee
where $\chi_r$ is the arithmetic genus of the divisor $\cD_r$ given by
\be
\label{defL0}
\chi_r = \frac16\, \kappa r^3+ \frac{1}{12}\, c_{2} r\in \IZ .
\ee

We wrote the multiplier system \eqref{Multsys-hp} for generic $r$ because for $r>1$ 
it also enters the modular transformation of $h_{r,\mu}$. However, in this case
the transformation has a modular anomaly so that the generating series is only mock modular
\cite{Alexandrov:2016tnf,Alexandrov:2018lgp}.
This means that $h_{r,\mu}$ can be promoted to a non-holomorphic modular completion 
$\whh_{r,\mu}(\tau,\btau)$ constructed from iterated integrals of some modular forms 
and transforming itself as a true modular form of the same weight $-3/2$ 
and multiplier system \eqref{Multsys-hp}.
The fact that $h_{r,\mu}$ has depth $r-1$ means that the $\btau$-derivative of
its completion, which is known as {\it shadow} of the mock modular form, 
is itself a completion of a mock modular form of depth $r-2$ 
(see \cite{Bringmann:2021dxg} for the precise definition).

The explicit form of the completion has been found in \cite{Alexandrov:2018lgp} 
and then slightly simplified in \cite{Alexandrov:2024jnu}. 
As a result, it reads as\footnote{We always assume that $r_i$ are
positive integers and do not write explicitly this condition in the sum over decompositions of D4-brane charge.}
\be
\whh_{r,\mu}(\tau,\btau)=h_{r,\mu}(\tau)+ \sum_{n=2}^r \sum_{\sum_{i=1}^n r_i=r}
\sum_{\bfmu}
\rmRi{\bfr}_{\mu,\bfmu}(\tau, \btau)
\prod_{i=1}^n h_{r_i,\mu_i}(\tau),
\label{exp-whh}
\ee
where we use the bold script to denote tuples of $n$ variables like $\bfr=(r_1,\dots, r_n)$. 
Note that the first term can be included into the second by setting $\rmRi{r}_{\mu,\mu'}=\delta_{\mu,\mu'}$.
The other coefficients $\rmRi{\bfr}_{\mu,\bfmu}$ can be represented as non-holomorphic 
theta series defined on a $n-1$-dimensional lattice
\be
\rmRi{\bfr}_{\mu,\bfmu}(\tau, \btau)=
\sum_{\sum_{i=1}^n q_i=\mu+\hf\kappa r^2 \atop q_i\in \kappa r_i \IZ+\mu_i+\hf\kappa r_i^2} 
\Sym\Bigl\{(-1)^{\sum_{i<j} \gamma_{ij} }\scR_n(\bfhgam;\tau_2)\Bigr\}\, e^{\pi\I \tau Q_n(\bfhgam)},
\label{defRn}
\ee
where $\bfhgam$ is the $n$-tuple of reduced charge vectors $\hgam_i=(r_i,q_{i})$
with $q_{i}$ decomposed as in \eqref{defmu} with fixed residue classes $\mu_i$.
Besides, $\Sym$ denotes symmetrization (with weight $1/n!$) with respect to charges $\hgam_i$, 
$\gamma_{ij}$ is the anti-symmetric Dirac product of charges
\be
\gamma_{ij} = \Nr_j q_i - \Nr_i q_j ,
\label{def-gammaij}
\ee
and $Q_n$ denotes the quadratic form, originating from the quadratic term in the definition \eqref{defqhat} of
the invariant charge $\hat q_0$,
\be
Q_n(\bfhgam)= \frac{1}{\kappa r}\(\sum_{i=1}^n q_i\)^2-\sum_{i=1}^n\frac{q_i^2}{\kappa r_i} \, .
\label{defQlr}
\ee
Finally, the coefficients $\scR_n$ determining the kernel of 
the theta series are constructed as suitable combinations of derivatives 
of the so-called generalized error functions introduced in
\cite{Alexandrov:2016enp,Nazaroglu:2016lmr}. We relegate the precise definitions of all these
functions to appendix \ref{ap-E}.

The equation for the completion \eqref{exp-whh} specifies the modular anomaly of $h_{r,\mu}$. 
Equivalently, one can talk about the holomorphic anomaly for $\whh_{r,\mu}$.
One can check that this anomaly given by $\p_{\btau}\whh_{r,\mu}$ 
is manifestly modular since it can be expressed through the completions
$\whh_{r_i,\mu_i}$, see \cite[Eq.(5.35)]{Alexandrov:2018lgp}.
The important feature of all these anomaly equations is that the r.h.s. is expressed
through the generating series for charges $r_i<r$.
Thus, \eqref{exp-whh} can be seen as a recursive system of equations determining the anomalous 
parts of the generating functions.

\subsection{Redefinition}
\label{subsec-shift}

Before we turn to the main goal of this paper, 
which is to solve the anomaly equations \eqref{exp-whh}, let us make a slight redefinition
of the generating series $h_{r,\mu}$ by shifting their vector index and multiplying by a sign factor.
This will allow us to avoid some annoying shifts and signs in what follows. 
More precisely, we set
\be
\tlh_{r,\mu}(\tau) =
(-1)^{(\Nr-1)\mu }h_{r,\tmu(r)}(\tau),
\label{def-tgi}
\ee
where
\be
\tmu(r) = \mu - \frac{\kappa \Nr(\Nr-1)}{2}\, .
\label{def-tmu}
\ee

This redefinition leads to two simplifications.
First, the shift of $\mu$ replaces the quadratic term in the spectral flow decomposition \eqref{defmu} by a linear one
so that now it reads
\be
\label{defmu-shift}
q = \mu + \frac12\, \kappa r + \kappa r \eps.
\ee
As a result, all such terms cancel in the condition on the sum over $q_i$ in \eqref{defRn} and, 
using the decomposition \eqref{defmu-shift}, it can be rewritten as
\be
\kappa\sum_{i=1}^n r_i\eps_i=\Delta\mu, 
\qquad
\Delta\mu=\mu-\sum_{i=1}^n\mu_i.
\label{constr-epsi}
\ee

Second, the sign factor in \eqref{def-tgi} cancels the sign factor in \eqref{defRn}.
Indeed, substituting the decomposition \eqref{defmu-shift} into the Dirac product \eqref{def-gammaij}, 
one finds 
\be
\begin{split}
	\sum_{i<j}\gamma_{ij} &=\, 
	\sum_{i<j} \(\Nr_j \(\kappa \Nr_i \eps_i + \mu_i + \frac{\kappa \Nr_i}{2}\) 
	- \Nr_i \(\kappa \Nr_j \eps_j + \mu_i + \frac{\kappa \Nr_j}{2}\)\)
	\\
	&=\,
	\sum_{i<j}\Bigl( \kappa \Nr_i \Nr_j (\eps_i+\eps_j) + \Nr_j \mu_i + \Nr_i \mu_j\Bigr)
	\mod 2
	\\
	&=\, 
	\sum_{i}\Bigl( \kappa \Nr_i (\Nr-\Nr_i)\eps_i + (\Nr-\Nr_i) \mu_i\Bigr) \mod 2
	\\
	&=\, 
	(\Nr-1)\sum_{i} \(\kappa\Nr_i \eps_i +\mu_i\)-\sum_i (\Nr_i-1) \mu_i
	\mod 2
	\\
	&=\, (\Nr-1) \mu - \sum_i (\Nr_i-1) \mu_i \mod 2 ,
\end{split}
\ee
where at the last step we used \eqref{constr-epsi}.

Given the above observations, the anomaly equation \eqref{exp-whh}
reformulated in terms of the redefined generating series takes the following form
\be
\twhh_{r,\mu}(\tau,\btau)=\sum_{n=1}^r \sum_{\sum_{i=1}^n r_i=r}
\sum_{\bfmu}
\trmRi{\bfr}_{\mu,\bfmu}(\tau, \btau)
\prod_{i=1}^n \tlh_{r_i,\mu_i}(\tau),
\label{exp-twhh}
\ee
where
\be
\trmRi{\bfr}_{\mu,\bfmu}(\tau, \btau)=
\sum_{\sum_{i=1}^n q_i=\mu+\kappa r/2 \atop q_i\in \kappa r_i \IZ+\mu_i+\kappa r_i/2} 
\Sym\Bigl\{\scR_n(\bfhgam;\tau_2)\Bigr\}\, e^{\pi\I \tau Q_n(\bfhgam)}.
\label{redefRn}
\ee
The redefinition \eqref{def-tgi} also affects the multiplier system which is now
given by \eqref{mult-thr}.

\section{Anomalous coefficients}
\label{sec-ansatz}

We expect that for each D4-brane charge $r$, the anomaly equation fixes the generating function
$\tlh_{r,\mu}$ up to a modular ambiguity which in turn can be fixed by other means, e.g.
by computing first few terms in the Fourier expansion of $\tlh_{r,\mu}$.
In other words, we can represent 
\be
\tlh_{r,\mu}=\than_{r,\mu}+\thh_{r,\mu},
\label{han}
\ee
where $\than_{r,\mu}$ is a depth $r-1$ mock modular form satisfying \eqref{exp-twhh},
while $\thh_{r,\mu}$ is pure modular.
The problem however is that the r.h.s. of \eqref{exp-twhh} depends on the full generating functions 
$\tlh_{r_i,\mu_i}$ with $r_i<r$ and hence on all $\thh_{r_i,\mu_i}$ which remain unknown at this point.
Therefore, $\than_{r,\mu}$ must also depend on them, and what we can do at best is to find $\than_{r,\mu}$ 
up to these modular functions.
To achieve this goal, we first parametrize the dependence of $\tlh_{r,\mu}$ on $\thh_{r_i,\mu_i}$
by holomorphic functions $\gi{\bfr}_{\mu,\bfmu}(\tau)$ which we call anomalous coefficients,
characterize them by anomaly equations similar to \eqref{exp-twhh}, and then solve 
these equations. 
In this section we perform the first two steps and leave the third one to the subsequent sections.
The main result is captured by the following

\begin{theorem}\label{thm-ancoef}
Let $\gi{\Nr}_{\mu,\mu'}=\delta_{\mu,\mu'}$ and $\thh_{\Nr,\mu}$ be a set of holomorphic modular forms.
Then 
\be
\tlh_{\Nr,\mu}(\tau)= \sum_{n=1}^\Nr\sum_{\sum_{i=1}^n \Nr_i=\Nr}
\sum_{\bfmu}
\gi{\bfr}_{\mu,\bfmu}(\tau)
\prod_{i=1}^n \thh_{\Nr_i,\mu_i}(\tau),
\label{genansatz}
\ee
is a depth $r-1$ modular form with completion of the form \eqref{exp-twhh}
provided $\gi{\bfr}_{\mu,\bfmu}$ are depth $n-1$ mock modular forms (where $n$ is the number of charges $r_i$)
with completions satisfying
\be
\whgi{\bfr}_{\mu,\bfmu}=
\Sym\Bigg\{\sum_{m=1}^n 
\sum_{\sum_{k=1}^m n_k=n}
\sum_{\bfnu}\trmRi{\bfs}_{\mu,\bfnu}
\prod_{k=1}^m \gi{\frr_k}_{\nu_k,\frm_k}
\Bigg\},
\label{compl-gi}
\ee
where\footnote{Note that while the sets $\bfr$ and $\bfmu$ have $n$ elements, 
	the sets $\bfs$ and $\bfnu$ have only $m\le n$ elements. 
	To comprehend the structure of the equation \eqref{compl-gi}, it might be 
	useful to use the fact that the sum on its r.h.s. is equivalent to the sum over
	rooted trees of depth 2 with leaves labelled by charges $r_i$ and other vertices
	labelled by the sum of charges of their children. Using this labelling, we assign 
	the function $\trmRi{\bfs}_{\mu,\bfnu}$ to the root vertex and 
	the anomalous coefficients to the vertices of depth 1 with arguments determined by the charges 
	of their children.
	Then the contribution of a tree is given by the product of the weights of its vertices.
     See Fig. \ref{fig-g-trees}. } 
\be
j_k=\sum_{l=1}^{k-1} n_l,
\qquad
\Ms_k=\sum_{i=1}^{n_k} \Nr_{j_k+i},
\qquad
\begin{array}{c}
\frr_k=(\Nr_{j_k+1},\dots, \Nr_{j_{k+1}}),
\\
\frm_k=(\mu_{j_k+1},\dots,\mu_{j_{k+1}}).
\vphantom{\sum\limits^{a}}
\end{array}
\label{split-rs}
\ee
\end{theorem}

\lfig{A representation of contributions to the r.h.s. of \eqref{compl-gi} in terms of rooted trees of depth 2.}
{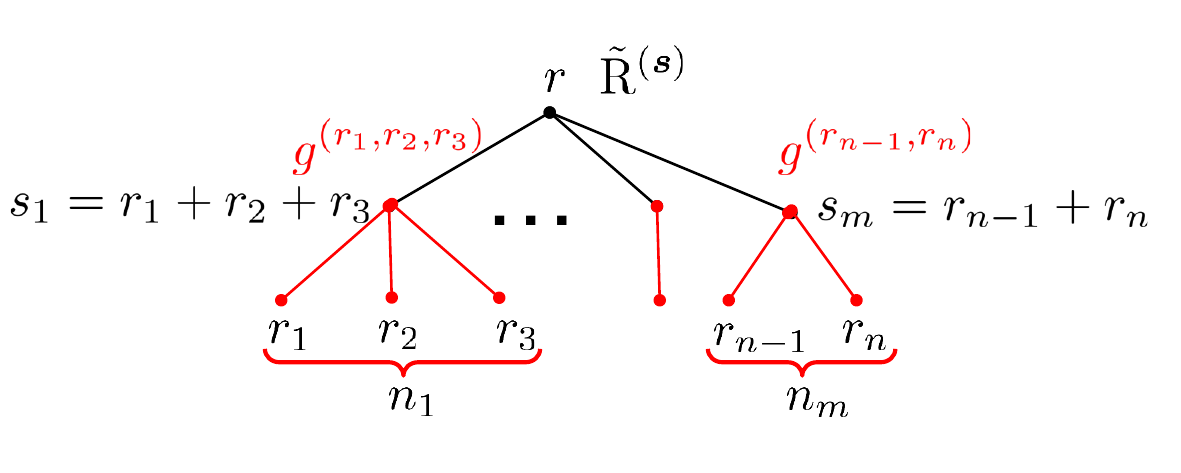}{10cm}{fig-g-trees}{-0.7cm}

\begin{proof}
	To prove the theorem, we must show that \eqref{genansatz} and \eqref{compl-gi} give rise to the same
	modular completions $\twhh_{r,\mu}$.
	On one hand, this completion is obtained by substituting the ansatz \eqref{genansatz}
	into the r.h.s. of \eqref{exp-twhh} which gives
	\begin{equation}
		\twhh_{\Nr,\mu}= \sum_{n=1}^\Nr \sum_{\sum_{k=1}^n \Ms_k=\Nr}
		\sum_{\bfmu}
		\trmRi{\bfs}_{\mu,\bfmu}
		\prod_{k=1}^n
		\Bigg(
		\sum_{n_k=1}^{\Ms_k}\sum_{\sum_{l=1}^{n_k}\Nr_{k,l} = \Ms_k}\sum_{\bfnu_k}
		\gi{\frr_k}_{\mu_k, \bfnu_k}\prod_{l=1}^{n_k}\thh_{\Nr_{k,l}, \nu_{k,l}}
		\Bigg).
    \label{firstexpr}
	\end{equation}
	On the other hand, it is obtained by completing each term in \eqref{genansatz} and then using
	the equation \eqref{compl-gi}. This leads to the following expression
	\be
	\begin{split}
		\twhh_{\Nr,\mu}=&\,  \sum_{n=1}^\Nr\sum_{\sum_{i=1}^n \Nr_i=\Nr}
		\sum_{\bfmu}
		\whgi{\bfr}_{\mu,\bfmu}
		\prod_{i=1}^n \thh_{\Nr_i,\mu_i}
		\\
		=&\, \sum_{n=1}^{\Nr} \sum_{\sum_{i=1}^{n} \Nr_i=\Nr} \sum_{\bfmu}
		\sum_{m=1}^n \sum_{\sum_{k=1}^{m}n_k=n} \sum_{\bfnu}
		\trmRi{\bfs}_{\mu, \bfnu} 
		\prod_k \gi{\frr_k}_{\nu_k, \frm_k}
		\prod_{i=1}^n \thh_{\Nr_i,\mu_i},
		\label{secondexpr}
	\end{split}
	\ee
	where we omitted the sign of symmetrization from \eqref{compl-gi} because it is ensured
	by the sums over decompositions $\Nr=\sum_{i=1}^n \Nr_i$ and residue classes $\mu_i$. 
	Clearly, the two equations coincide if one can identify $(r_{k,l},\mu_k,\nu_{k,l})$ in \eqref{firstexpr} 
	with $(r_i,\nu_k,\mu_i)$ in \eqref{secondexpr} and claim that
	\be
	\sum_{n=1}^\Nr \sum_{\sum_{k=1}^n \Ms_k=\Nr}\sum_{n_k=1}^{\Ms_k}\sum_{\sum_{l=1}^{n_k}\Nr_{k,l} = \Ms_k}
	=\sum_{n=1}^{\Nr} \sum_{\sum_{i=1}^{n} \Nr_i=\Nr} \sum_{m=1}^n \sum_{\sum_{k=1}^{m}n_k=n}.
	\ee
	Here on the l.h.s. the sum goes over double decompositions: first, decompose $\Nr$ into $\Ms_k$ 
	and then each $\Ms_k$ into $\Nr_{k,l}$,
	while on the r.h.s. one first sums over decompositions of $\Nr$ into $\Nr_i$ 
	and then over various groupings of the indices $i$
	into sets preserving their ordering.
	It is obvious that the two sums are identical and hence the two representations of the completion,
	\eqref{firstexpr} and \eqref{secondexpr}, coincide.
\end{proof}

This theorem defines a family of holomorphic functions $\gi{\bfr}_{\mu,\bfmu}(\tau)$ 
restricted to have a modular anomaly determined by \eqref{compl-gi}.
Their weight and multiplier system follow from that of $\tlh_{r,\mu}$ (see \eqref{mult-thr})
and are given in \eqref{mult-gr}.
As usual, together with the anomaly these data fix the anomalous coefficients up to a modular ambiguity. 
However, in contrast to the case of the generating functions, there are no any restrictions
on this ambiguity and therefore we can choose any solution of \eqref{compl-gi}.
It is only important that, once a solution has been chosen for small charges, 
it is this solution that is used in the r.h.s. of \eqref{compl-gi}
to determine $\gi{\bfr}_{\mu,\bfmu}$ for higher charges.
In the rest of the paper, our goal will be to find explicit solutions for these functions.

An important observation is that the parameter $\kappa$ and charges $\Nr_i$ enter the functions $\kappa^{n-1}\trmRi{\bfr}_{\mu,\bfmu}$
and the multiplier system \eqref{mult-gr} of the anomalous coefficients $\gi{\bfr}_{\mu,\bfmu}$
always in the form of the product $\kappa \Nr_i$. Indeed, this is true for the quadratic form 
\eqref{defQlr}, the bilinear form \eqref{biform} and the charges \eqref{defmu-shift}. 
On the other hand, the vectors $\bfv_{ij}$ \eqref{defvij} are linear in $\Nr_i$.
Since the generalized error functions \eqref{generrPhiME} 
are independent of the overall scale of the matrix of parameters,
this implies that the functions \eqref{rescEn} and hence $\Ev_n$ are homogeneous of degree $n-1$ in $\Nr_i$.
Then \eqref{solRn} immediately leads to the same property for $\scR_n$,
which in turn confirms the claim for $\trmRi{\bfr}_{\mu,\bfmu}$.
This implies that we can (although do not have to) choose a solution for all $\gi{\bfr}_{\mu,\bfmu}$
satisfying the following property
\be
\gi{\kappa;\bfr}_{\mu,\bfmu}(\tau)=\kappa^{1-n}\,\gi{1;\kappa\bfr}_{\mu,\bfmu}(\tau),
\label{kappato1}
\ee
where $n$ is the number of charges and we explicitly indicated the dependence on $\kappa$
in the upper index.
Using this property would allow to reduce the problem of finding the anomalous coefficients to the case of $\kappa=1$.
Note that for this feature to hold it was crucial to perform the redefinition of section \ref{subsec-shift}.

\section{Partial solutions}
\label{sec-2sol}

In this section we provide a solution for two infinite families of anomalous coefficients.

\subsection{Two charges and Hecke-like operators}
\label{subsec-sol2ch}

Let us first consider the case of two arbitrary charges $r_1$ and $r_2$.
In this case the formula for the modular completion $\whgi{\Nr_1,\Nr_2}_{\mu,\mu_1,\mu_2}$ \eqref{compl-gi}, 
representing the anomaly equation, takes the simple form
\be
\whgi{\Nr_1,\Nr_2}_{\mu,\mu_1,\mu_2}(\tau,\btau)
=\gi{\Nr_1,\Nr_2}_{\mu,\mu_1,\mu_2}(\tau)+\trmRi{\Nr_1,\Nr_2}_{\mu, \mu_1, \mu_2}(\tau, \btau),
\label{whh2}
\ee
and $\gi{\Nr_1,\Nr_2}_{\mu,\mu_1,\mu_2}$ is required to be a mock modular form of weight 3/2
with the multiplier system \eqref{mult-gr} specialized to $n=2$.
The function $\trmRi{\Nr_1,\Nr_2}_{\mu, \mu_1, \mu_2}$ determining the completion 
is easily computable, but for our purposes it is sufficient to consider its derivative
with respect to $\btau$ which specifies the shadow of $\gi{\Nr_1,\Nr_2}_{\mu,\mu_1,\mu_2}$.
It is given in \eqref{shadowRN1N2} and suggests to look for a solution 
of the form
\be
\gi{\Nr_1,\Nr_2}_{\mu,\mu_1,\mu_2}(\tau)=\Nr_0\delta^{(\kappa\Nr_0)}_{\Delta\mu}
\Gi{\kappa_{12}}_{\mu_{12}}(\tau),
\label{sol-n=2}
\ee
where $\Nr_0=\gcd(\Nr_1,\Nr_2)$, $\Delta\mu$ was defined in \eqref{constr-epsi},
$\kappa_{12}$ and $\mu_{12}$ are effective parameters introduced in \eqref{defmu0}, 
$\mu_{12}$ runs over $2\kappa_{12}$ values,
and $\delta^{(n)}_x$ is the mod-$n$ Kronecker delta defined by
\be
\label{defdelta}
\delta^{(n)}_x=\left\{ \begin{array}{ll}
	1\  & \mbox{if } x=0\!\!\!\mod n,
	\\
	0\ & \mbox{otherwise}.
\end{array}\right.
\ee
In particular, the ansatz \eqref{sol-n=2} is consistent with the property \eqref{kappato1}.
If $\Gi{\kappa}_\mu$ is a VV mock modular form
of weight 3/2 with a modular completion satisfying
\be
\tau_2^{3/2}\p_{\btau}\whGi{\kappa}_\mu(\tau,\btau) 
=\frac{\sqrt{\kappa}}{16\pi\I}\, \overline{\ths{\kappa}_\mu}(\btau),
\label{shadG}
\ee
where $\ths{\kappa}_\mu(\tau)$ is the theta series \eqref{deftheta} at $z=0$,
then it is trivial to see that \eqref{sol-n=2} solves the anomaly equation \eqref{whh2}.
The only non-trivial fact to check is that it has the correct multiplier system. 
But this follows directly from Proposition \ref{prop-multsys} 
because the relation \eqref{shadG} ensures that $\Gi{\kappa}_\mu$ has the multiplier system 
$\Mi{\kappa}_{\mu \nu}$ \eqref{STusual} conjugate to that of $\ths{\kappa}_\mu$. 

As a result, we have reduced the problem of finding the anomalous coefficients for arbitrary two charges 
to exactly the same problem that was studied in \cite{Alexandrov:2022pgd} for charges $r_1=r_2=1$,
in which case $\kappa_{12}=\kappa$.
It was found that for any $\kappa$ equal to a power of a prime integer, $\Gi{\kappa}_\mu$
is determined by the generating series $H_\mu$ ($\mu=0,1$) of Hurwitz class numbers\footnote{An explicit formula
for the generating series can be found in 
\cite[Eq.(1.12)]{bringmann2007overpartitionsclassnumbersbinary}
and its mock modular properties have been established in \cite{Zagier:1975,hirzebruch1976intersection}.}  
through the action on it by a certain modification of the Hecke-like operator introduced in 
\cite{Bouchard:2016lfg,Bouchard:2018pem}.
However, it turns out that a solution of this problem for generic $\kappa$ has already been found
in the seminal paper \cite{Dabholkar:2012nd}. 
More precisely, that paper looked for mock modular forms with 
shadow proportional to $\ths{\kappa}_\mu$
and further restricted to have the slowest possible asymptotic growth of their Fourier coefficients. 
Such functions have been called mock modular forms of optimal growth.
In our case we do not have to impose any restrictions on the asymptotic growth.
But since any solution of \eqref{shadG} is equally suitable, we can take the one provided 
by \cite{Dabholkar:2012nd}. All other solutions should differ just by a pure modular form.

In the rest of this subsection we present the formula for the mock modular forms of optimal growth 
found in \cite{Dabholkar:2012nd} adjusting (and correcting) some normalization 
factors.\footnote{Strictly speaking, \cite{Dabholkar:2012nd} worked in terms of {\it mock Jacobi forms} 
	rather than mock modular forms. However, the latter can be easily extracted from the former 
	by means of the theta expansion \eqref{thetaexp}.}
To this end, let 
\begin{itemize}
	\item
	$\omega(\kappa)$ be the number of distinct prime divisors of $\kappa$, i.e.
	\be
	\omega(\kappa)=\#\{p\, :\ p \mbox{ is prime and }p|\kappa \};
	\ee
	
	\item
	$\mu(d)$ be the M\"obious function given by
	\be
	\mu(d)=\left\{ \begin{array}{ll}
		+1 \quad & \mbox{if $d$ is a square-free with an even number of prime factors},
		\\
		-1 & \mbox{if $d$ is a square-free with an odd number of prime factors},
		\\
		0 & \mbox{if $d$ has a squared prime factor};
	\end{array}
	\right.
	\ee
	
	\item
	$\hcTi{d}_{r}$ be a Hecke-like operator given, when acting on modular forms 
	(not necessarily holomorphic)
	of weight $(w,\bw)$ and multiplier system $\Mi{d}_{\mu \nu}$, by
	\be
	\hcTi{d}_{r}=\sum_{s^2| r \atop \gcd(s,d)=1} \mu(s)\,s^{w+\bw-1/2}\,U_s \cT_{r/s^2}
	\label{def-Hecke-hT}
	\ee
	with
	\be
	(U_s[\hi{\kappa}])_\mu(\tau)=\delta^{(s)}_{\mu}\,  \hi{\kappa}_{\mu/s}(\tau),
	\qquad
	\mu=0,\dots,2s^2 \kappa-1,
	\label{defHeckeU}
	\ee
	and
	\be
	(\cT_r[\hi{\kappa}])_\mu(\tau)=\sum_{a,d>0 \atop ad=r}\(\frac{\sqrt{r}}{d}\)^{w+\bw+\hf}\sum_{b=0}^{d-1}
	\delta^{(1)}_{\mu/a} \, e^{\frac{\pi\I b}{2\kappa a r}\,\mu^2}
	\hi{\kappa}_{\mu/a} \(\frac{a\tau+b}{d}\),
	\qquad
	\mu=0,\dots,2r\kappa-1.
	\label{defHecke-exp}
	\ee
	We relate these operators to the ones defined in \cite{Dabholkar:2012nd} and acting on Jacobi forms 
	in appendix \ref{ap-Hecke}. One can also check that for $r$ prime power, $\hcTi{1}_{r}$ coincides with 
	the modification of $\cT_r$ introduced in \cite{Alexandrov:2022pgd}.
	
\end{itemize}
In terms of these quantities, the mock modular forms of optimal growth are given by
\be
\Gi{\kappa}_\mu=2^{-\omega(\kappa)} \sum_{d|\kappa\atop \mu(d)=1} \(\hcTi{d}_{\kappa/d}\[\cGi{d}\]\)_\mu,
\label{DMZ}
\ee
where
$\cGi{d}$ are VV mock modular forms of weight 3/2 with multiplier system $\Mi{d}_{\mu \nu}$.
Thus, for each square-free integer with an even number of prime factors, such as 1, 6, 10, 14, 15, etc.,
one needs to provide such a mock modular form.
The first two of them turn out to be well-known functions:
for $d=1$ it is (the doublet of) the generating series of Hurwitz class numbers,
\be
\cGi{1}_\mu(\tau)=H_\mu(\tau),
\ee
and for $d=6$ it has the following explicit expression 
\be
\cGi{6}_\mu(\tau)=\frac{\chi_{12}(\mu)}{12}\, \hi{6}(\tau),
\label{def-cG6}
\ee
where
\be
\chi_{12}(\mu)=\left\{ \begin{array}{ll}
	+1 \quad & \mbox{if } \mu=\pm 1 \mod 12,
	\\
	-1 & \mbox{if } \mu=\pm 5 \mod 12,
	\\
	0 & \mbox{if } \gcd(\mu,12)>0,
\end{array}
\right.
\ee
and 
\be
\hi{6}(\tau)=\frac{12 F_2^{(6)}(\tau)-E_2(\tau)}{\eta(\tau)}
\ee
is a mock modular form of weight 3/2 with shadow proportional to the Dedekind eta function $\eta(\tau)$, 
which is defined in terms of the quasimodular Eisenstein series $E_2(\tau)$ and the function
\be
F_2^{(6)}(\tau)=-\sum_{r>s>0}\chi_{12}(r^2-s^2)\, s\, \q^{rs/6}\, . 
\ee
For many other functions $\cGi{d}$, \cite{Dabholkar:2012nd} determined their first Fourier coefficients,
however we are not aware about any explicit expressions for their generating series.

\subsection{Unit charges and Vafa-Witten}
\label{subsec-solVW}

Let us now consider the case of $n$ charges $r_i$ all equal to 1. In addition, we also restrict ourselves
to CYs with the intersection number $\kappa=1$. 
A crucial simplification in this case is that one can drop all indices $\mu_i$ 
because they take only $\kappa r_i=1$ value. Therefore, the corresponding anomalous coefficients
can be denoted simply as $g_{n,\mu}\equiv\gi{1,\dots,1}_{\mu}$.
Another feature of this set of anomalous coefficients is that the anomaly equations for $g_{n,\mu}$
form a closed system and do not involve other anomalous coefficients.
Moreover, it is easy to see that in this sector the anomaly equation \eqref{compl-gi} 
becomes identical to \eqref{exp-twhh} under the identification $g_{n,\mu}\leftrightarrow \tlh_{n,\mu}$ 
and thus takes the form 
\be
\whg_{n,\mu}=
\sum_{m=1}^n \sum_{\sum_{k=1}^m n_k=n}
\sum_{\bfmu}\trmRi{\bfn}_{\mu,\bfmu}
\prod_{k=1}^m g_{n_k,\mu_k}.
\label{exp-whtg}
\ee

The case $n=2$ has already been analyzed in the previous subsection. It follows from the results 
presented there, and in agreement with \cite{Alexandrov:2022pgd}, that
\be  
g_{2,\mu}=H_\mu,
\qquad 
\mu=0,1.
\ee
The vector valued function $H_\mu$ appearing here is known not only as the generating series of Hurwitz class numbers, 
but also as the (normalized) generating series of $SU(2)$ Vafa-Witten invariants on $\IP^2$, namely \cite{Vafa:1994tf}
\be  
\vwh_{2,\mu}=3 (h^{\rm VW}_1)^2 H_\mu ,
\ee
where $\vwh_{n,\mu}$ denotes the generating series of $SU(n)$ VW invariants and $h^{\rm VW}_1=\eta^{-3}$.
Combining the two relations, one obtains 
\be
g_{2,\mu}=\frac13\, \vwgi{2,\mu},
\label{relVW-n=2}
\ee
where we introduced the normalized generating series
\be
\vwgi{n,\mu}(\tau) = \eta^{3n}(\tau)\,\vwh_{n,\mu}(\tau).
\label{def-VWnorm}
\ee
As we show below, the relation \eqref{relVW-n=2} is not an accident, but a particular case 
of a more general relation between $g_{n,\mu}$ and $\vwgi{n,\mu}$.

Let us recall that the VW invariants count the Euler characteristic of moduli spaces of instantons 
in a topological supersymmetric gauge theory on a complex surface $S$ obtained from the usual $N=4$ super-Yang-Mills 
by a topological twist \cite{Vafa:1994tf}. The partition function of the theory reduces to the generating 
series of VW invariants and one could expect that it must be a modular form as a consequence of S-duality of 
the $N=4$ super-Yang-Mills. However, it turns out that on surfaces with $b_2^+(S)=1$, which includes $S=\IP^2$,
there is a modular anomaly \cite{Vafa:1994tf,Dabholkar:2020fde}. 
Its precise form can be established from the fact that the VW invariants on $S$ coincide with 
the D4-D2-D0 BPS indices on the non-compact CY given by the canonical bundle over $S$ 
\cite{Minahan:1998vr,Alim:2010cf,gholampour2017localized},
which in turn can be obtained from a compact CY given by an elliptic fibration over $S$
in the limit of large fiber. Since the modularity of the D4-D2-D0 BPS indices on such compact CY 
is governed by a generalization of \eqref{exp-whh} or \eqref{exp-twhh} to $b_2>1$, 
the generating series of VW invariants are subject to the same anomaly equation \cite{Alexandrov:2019rth}.

Furthermore, since in the local limit where the elliptic fiber becomes large 
the only divisor which remains finite is $[S]$, the D4-brane charges belong to the one-dimensional lattice,
and if $b_2(S)=1$, as is the case for $\IP^2$, the lattice of D2-brane charges is also one-dimensional.
Thus, for $S=\IP^2$ one reduces to the ``one-dimensional" case captured by the anomaly equation \eqref{exp-twhh}
with $\kappa=[H]^2=1$ where $[H]$ is the hyperplane class of $\IP^2$. 
However, the fact that the anomaly equation arises as a limit of a compact CY with $b_2>1$
does lead to two modifications: the second term in the spectral flow decomposition \eqref{defmu-shift}
and the Dirac product of charges \eqref{def-gammaij} both get an additional factor of $-3$, 
which can be traced back to the value of the first Chern class $c_1(\IP^2)=3[H]$
\cite[Ap.F]{Alexandrov:2019rth}.\footnote{Strictly speaking, \cite{Alexandrov:2019rth}
analyzed generating functions of {\it refined} VW invariants (see \S\ref{subsec-ref}) which 
count Betti numbers of moduli spaces of instantons. However, the presented results
are easily recovered in the unrefined limit.}
Thus, if one denotes the functions $\trmRi{\bfn}_{\mu,\bfmu}$
with these two modifications implemented by $\trmRVWi{\bfn}_{\mu,\bfmu}$, 
then the normalized generating series of VW invariants satisfy
\be
\whvwg_{n,\mu}=
\sum_{m=1}^n \sum_{\sum_{k=1}^m n_k=n}
\sum_{\bfmu}\trmRVWi{\bfn}_{\mu,\bfmu}
\prod_{k=1}^m \vwgi{n_k,\mu_k}.
\label{exp-whVW}
\ee

The first modification can actually be undone by a simple shift of the spectral flow parameter.
On the other hand, the second one is equivalent to multiplying the vectors $\bfv_{ij}$ \eqref{defvij}
by $-3$. Under this rescaling of parameters, the functions \eqref{rescEn} simply
get an overall factor $(-3)^{n-1}(-1)^{n-1}=3^{n-1}$. Thus, one concludes that 
\be 
\trmRVWi{\bfn}_{\mu,\bfmu}=3^{m-1}\trmRi{\bfn}_{\mu,\bfmu},
\ee
where $m$ is the number of charges which the functions depend on.
Substituting this into \eqref{exp-whVW} and comparing to \eqref{exp-whtg},
one finds that the two equations become identical provided one identifies\footnote{The freedom
to include in this relation a constant factor $c^n$ allowed by the equations 
is fixed by the normalization conditions $g_{1}=\vwgi{1}=1$.} 
\be
g_{n,\mu}=3^{1-n}\vwgi{n,\mu}(\tau).
\label{rel-tggP2-n}
\ee
This result is consistent with \eqref{relVW-n=2} and provides an explicit solution for 
the anomalous coefficients with $r_i=\kappa=1$.

\section{Solution via indefinite theta series}
\label{sec-sol-indef}

\subsection{Motivation and strategy}
\label{subsec-motiv}

In the previous section we have found solutions for 
two infinite families of anomalous coefficients: with two arbitrary charges and 
with arbitrary number of charges, but all set to 1 together with the intersection number.
It is natural to try extending these solutions to more general cases. In particular, one could expect that 
a solution for the case with $r_i=1$ but $\kappa$ an arbitrary prime number should be 
described by the action of Hecke-like operators similar to \eqref{def-Hecke-hT} on the 
normalized generating functions of $SU(n)$ VW invariants on $\IP^2$.
However, we have not been able to show this and all other attempts to extend the above constructions
also failed.\footnote{At technical level, there are two main complications appearing for $n>2$. First, 
$\gi{\bfr}_{\mu,\bfmu}$ do not reduce to a vector-like object and keep a non-trivial tensor structure
(cf. \eqref{sol-n=2}). Second, the action of Hecke-like operators on a product of functions is not factorized,
so that applying them to the r.h.s. of \eqref{compl-gi}, for terms with $1<m<n$, one cannot proceed in an iterative way.}
Therefore, we change the strategy and suggest an approach which works in general.
It is based on the use of indefinite theta series and is similar 
to the solution of the same kind of modular anomaly equation for the generating functions of VW invariants 
constructed in \cite{Alexandrov:2020bwg,Alexandrov:2020dyy}.

An indefinite theta series is defined as a sum over a lattice $\bbLambda$
endowed with a quadratic form $\kbbm^2$ of indefinite signature,\footnote{Note the unusual minus
	sign in the exponential. The same minus sign appears also in \eqref{gentheta}.  
	This convention follows the conventions used in the previous works on this topic
	and can be traced back to the natural quadratic form induced on D4 and D2-brane charges on a compact CY.
	In this convention, the usual convergent theta series with a trivial kernel correspond to {\it negative} definite 
	quadratic forms. \label{foot-sign}} 
\be
\vth_\bbmu(\tau)=\sum_{\kbbm\in\bbLambda+\bbmu} \Phi(\kbbm;\tau_2)\, e^{-\pi\I\tau \kbbm^2},
\label{def-indef}
\ee 
where $\bbmu\in\bbLambda^*/\bbLambda$ labels its different components.
The kernel $\Phi(\kbbm)$ can be a non-trivial function of $\tau_2=\Im\tau$ 
and must ensure convergence of the sum.
In fact, it is very natural to use such theta series to represent solutions of our modular anomaly equations because
for $b_2>1$ the functions analogous to $\rmRi{\bfr}_{\mu,\bfmu}$ 
have precisely the form \eqref{def-indef}. This is also true for $b_2=1$ (cf. \eqref{defRn} or \eqref{redefRn}), 
but in this case the relevant quadratic form coincides with $-Q_n$
\eqref{defQlr} and has a definite signature. But since it is {\it positive} definite, 
this case also calls for the use of indefinite theta series.

The anomalous coefficients we are looking for, and hence the indefinite theta series representing them, 
must be holomorphic in $\tau$. The only way to make \eqref{def-indef} holomorphic and convergent simultaneously
is to restrict the sum to the negative definite cone of the lattice, 
which can be done by choosing the kernel $\Phi(\kbbm)$ to be an appropriate combination of sign functions.
An example of such kernel is provided by Theorem \ref{th-conv}
and is characterized by two sets of vectors $\{\vbbm_{s,i}\}$, $s=1,2$.
As we will see below, one set is determined by the same vectors $\bfv_e$ \eqref{defue} that 
define the functions $\rmRi{\bfr}_{\mu,\bfmu}$, while the second set must consist of
{\it null} vectors, i.e. satisfying $\vbbm_{2,i}^2=0$.
This immediately implies that the lattice $\bbLambda$ cannot be the one that appears in 
the definition of $\rmRi{\bfr}_{\mu,\bfmu}$ \eqref{defRn} and will be denoted below by $\bfLami{\bfr}$,
since the numbers of positive and negative
eigenvalues of the quadratic form must be both non-vanishing for null vectors to exist.

This can be achieved by the so-called {\it lattice extension}, which is a standard trick in the theory
of mock modular forms \cite{Zwegers-thesis}. The idea is that the original problem defined on a lattice $\Lambda$ 
is reformulated on a larger lattice $\Lambda_{\rm ext}=\Lambda\oplus \Lambda_{\rm ad}$ 
that admits a solution in terms of indefinite theta series and, because $\Lambda_{\rm ext}$ is a direct sum,
such solution is expected to be reducible to a solution on $\Lambda$. 
However, if the discriminant group $D_{\rm ad}=\Lambda_{\rm ad}^*/\Lambda_{\rm ad}$ is non-trivial,
the reduction to the original lattice is possible only if the solution on the extended lattice satisfies 
certain identities ensuring that components of the solution labelled by different elements of 
$D_{\rm ad}$ reduce to the same functions. 
For example, this is the case for the generating functions of VW invariants
where the invariants on $\IP^2$ can be obtained from those on the Hirzebruch surface $\Fb_1$
because the latter satisfy the so-called blow-up identities \cite{BMR2015,Alexandrov:2020dyy}.
However, for a general solution on $\Lambda_{\rm ext}$ this is not the case. 
Therefore, we should require triviality of $D_{\rm ad}$,
which in turn requires that, if $\Lambda_{\rm ad}=\IZ^{d_{\rm ad}}$, then the corresponding 
quadratic form is given by (minus) the identity matrix.

In our case $\Lambda=\bfLami{\bfr}$
with quadratic form $-Q_n$ and $\Lambda_{\rm ad}$ should be chosen so that to ensure the existence 
of a null vector on $\bbLambda=\Lambda_{\rm ext}$. 
One could think that it is sufficient to take $\Lambda_{\rm ad}=\IZ$
with the quadratic form $-x^2$. But actually it is not because, for a theta series to converge, 
the null vector appearing in its kernel (possibly after rescaling) should {\it belong to the lattice}.
Otherwise, the indefinite theta series would diverge due to accumulation of lattice points near 
the null cone leading to an infinite number of terms of the same strength
(see \cite[\S B.3]{Alexandrov:2017qhn} for an illustrative example).
Thus, typically, the dimension of $\Lambda_{\rm ad}$ must be non-trivial.
The simplest possibility would be to take $d_{\rm ad}=\bfv^2$ where $\bfv\in \Lambda$
is a vector with the minimal norm. This would ensure that $(\bfv,\vec 1)\in\Lambda_{\rm ext} $,
where the vector $\vec 1$ has $d_{\rm ad}$ components all equal to 1, is a null vector.
However, this is not always the optimal choice and in our case it is actually inconsistent 
with the iterative structure of the equations \eqref{compl-gi}.
Below in \S\ref{subsec-extlat} we propose a lattice extension satisfying all the requirements
discussed above and adapted to our system of equations.

But this is not the end of the story. The problem is that even if the null vector belongs to the lattice, 
this does not ensure the convergence yet. The additional divergence comes from the sum over the sublattice
$\IZ \vbbm\subset\bbLambda$, where $\vbbm\in \bbLambda$ is the null vector.
This is easy to see for \eqref{def-indef} with quadratic form of signature $(1,n_-)$ and the kernel
$\Phi(\kbbm)=\sgn(\vbbm\ast\kbbm)-\sgn(\vbbm'\ast\kbbm)$.
A way out is to consider {\it Jacobi forms}  
instead of the usual modular forms.
They depend on an additional elliptic parameter $z$ transforming under $SL(2,\IZ)$ as
$z\mapsto z/(c\tau+d)$. For theta series, the elliptic transformation property of Jacobi forms 
fixes the dependence on $z$ as shown in \eqref{gentheta} (with $\zbbm=\bbtheta z$). 
In particular, it shifts the lattice vector $\kbbm$ in the kernel
and introduces an exponential $z$-dependent factor. Together these two changes allow to avoid 
the divergence due to the null vector, which manifests now as a pole at $z=0$.
Since eventually we are interested in the limit $z=0$, these poles should be cancelled 
by combining the indefinite theta series with certain {\it Jacobi-like modular forms}
(see \S\ref{ap-Jacobi} for the definition of Jacobi and Jacobi-like forms).
The latter have the same modular transformations as Jacobi forms, but they are not required 
to satisfy the elliptic property, which in our context is irrelevant since we care only about the behavior near $z=0$.
Note that, apart from relaxing the elliptic property,
exactly the same strategy to combine indefinite theta series constructed from null vectors with 
Jacobi forms cancelling poles has been used in \cite{Alexandrov:2020bwg}
to produce the generating functions of VW invariants on Hirzebruch and del Pezzo surfaces 
as solutions of a modular anomaly equation similar to \eqref{exp-whVW}.

The extension to Jacobi forms is known as the so-called {\it refinement} which has a physical interpretation
as switching on an $\Omega$-background \cite{Nekrasov:2002qd,Nekrasov:2010ka}
and has been investigated in the context of modularity of BPS indices in \cite{Alexandrov:2019rth}.
A quite unexpected result of that analysis is that the refinement considerably simplifies 
the coefficients $\scR_n$ determining the modular anomaly.
Thus, the necessity to introduce the refinement should be considered not as a shortcoming, 
but as a virtue which makes the system of anomaly equations more amenable to solution.

However, new complications arise when the refinement is combined with the lattice extension discussed above.
It turns out that for a solution on $\Lambda_{\rm ext}$ to be reducible to a solution on $\Lambda$,
it should have zero of order $d_{\rm ad}$ at $z=0$, which is very difficult to achieve.
Fortunately, there is a trick that allows to avoid this problem: 
one should introduce multiple refinement parameters combined in a vector $\zbbm$ 
so that the indefinite theta series become multi-variable (mock) Jacobi forms as \eqref{gentheta}.
Then, as will be shown below, 
if one sets $\zbbm=(z\bftet,\vec z)$ where $\bftet\in \Lambda$ and $\vec z$
has $d_{\rm ad}$ components and is such that
$(0,\vec z)$ is orthogonal to all null vectors, the lattice $\Lambda_{\rm ad}$ together with the associated
refinement parameters $\vec z$ decouples and the reduction to $\Lambda$ crucially simplifies.

To summarize, we need to perform the following steps:
\begin{enumerate}
	\item
	introduce refinement, 
	\item 
	extend the charge lattice so that it possesses a set of null vectors
	and is consistent with the anomaly equation,
	\item 
	associate with the extension a vector of additional refinement parameters
	satisfying certain orthogonality properties with the null vectors,
	\item 
	solve the refined system of anomaly equations on the extended lattice,
	\item 
	reduce the solution to the original lattice,
	\item 
	take the unrefined limit.
\end{enumerate}
In the next subsection we perform the first step. 
Steps 2 and 3 are done in \S\ref{subsec-extlat}.
The last 3 steps are realized in \S\ref{subsec-case-n=2} in the case of two charges
and in \S\ref{subsec-case-n=3} in the case of three charges.
Finally, in \S\ref{subsec-case-gen} we consider the generic case
for which we perform steps 4 and 5 explicitly, whereas the last step is too cumbersome 
to be done analytically.

\subsection{Refinement}
\label{subsec-ref}

As was mentioned in the previous subsection, a refinement has its physical origin 
in a non-trivial $\Omega$-background. It introduces a complex parameter $y=e^{2\pi\I z}$
which can be thought of as a fugacity conjugate to the angular momentum $J_3$ 
in uncompactified dimensions. At the same time, the BPS indices, which from the mathematical point of view
(roughly) count the Euler number of the moduli spaces of semi-stable coherent sheaves, are replaced 
by {\it refined} BPS indices which are symmetric Laurent polynomials in $y$ constructed 
from the Betti numbers of these moduli spaces. These refined indices are known 
to satisfy similar and even simpler wall-crossing relations as the usual ones 
\cite{ks,Manschot:2010qz,Alexandrov:2018iao}.
But most importantly is that the refinement preserves the modular properties of the generating series 
of BPS indices \cite{Alexandrov:2019rth}. More precisely, after refinement they become {\it mock Jacobi forms} 
for which the role of the elliptic argument is played by the refinement parameter $z$
and the formula for their modular completions takes exactly the same form as in \eqref{exp-whh}, but with the coefficients 
given now by\footnote{We give the coefficients {\it after} performing the same redefinition as in \eqref{def-tgi},
	so that the formula to compare with is \eqref{redefRn} rather than \eqref{defRn},
	but we omit the tilde on $\rmRirf{\bfr}_{\mu,\bfmu}$ to avoid cluttering.} 
\be
\label{Rirf-to-rmRrf}
\rmRirf{\bfr}_{\mu,\bfmu}(\tau,\btau,z)
=
\sum_{\sum_{i=1}^{n} q_i = \mu + \kappa \Nr/2 
	\atop q_i \in \kappa \Nr_i \IZ + \mu_i + \kappa \Nr_i/2} 
\Sym \Bigl\{ 
\scRrf_n(\bfhgam;\tau_2,\beta) \, y^{\sum_{i<j} \gamma_{ij}}
\Bigr\} \,e^{\pi\I \tau Q_n(\bfhgam)},
\ee
where we set $z=\alpha-\tau\beta$ with $\alpha,\beta\in \IR$.
The main difference here, besides the appearance of a power of $y$, 
lies in the form of the coefficients $\scRrf_n$
which we describe in appendix \ref{subsec-refR}. They turn out to be much simpler than their unrefined 
version $\scR_n$.\footnote{More precisely, while the formula \eqref{refsolRn} looks exactly as \eqref{solRn},
	these are the functions $\Er_n$ that are much simpler than their unrefined analogues $\Ev_n$
	(\eqref{rescEnPhi} and \eqref{rescEn} versus \eqref{Erefsim}).}
In particular, while the coefficients $\scR_n$
involve a sum over two types of trees weighted by generalized error functions and their derivatives,
for $\scRrf_n$ one needs only one type of trees and no derivatives.

It should be stressed that the status of the refined BPS indices for compact CY threefolds, 
the case we are really interested in, is unclear. While in the non-compact case they are well-defined 
due to a certain $\IC^\times$ action carried by the moduli space of semi-stable objects, 
in its absence it seems impossible to refine DT invariants in a deformation-invariant way
(see however \cite{ks}). This is not a problem for our construction because we do {\it not}
use the refined BPS indices or their generating functions, but only the coefficients 
\eqref{Rirf-to-rmRrf} characterizing the refined completions.
In other words, we use the existence and properties of $\rmRirf{\bfr}_{\mu,\bfmu}$ as a mere trick 
to produce solutions to the anomaly equations \eqref{compl-gi}.

In particular, the main property which we need is that in the unrefined limit 
$\rmRirf{\bfr}_{\mu,\bfmu}$ develop a zero of order $n-1$ with a coefficient given by $\trmRi{\bfr}_{\mu,\bfmu}$:
\be
\trmRi{\bfr}_{\mu,\bfmu}(\tau,\btau)=
\lim_{y\xrightarrow{}1} (y-y^{-1})^{1-n} \rmRirf{\bfr}_{\mu,\bfmu}(\tau,\btau,z).
\label{unreflim-Rn}
\ee   
Therefore, if we define {\it refined anomalous coefficients} as solutions of the following 
modular anomaly equation
\be
\begin{split}
	\whgirf{\bfr}_{\mu,\bfmu} 
	=&\,
	\Sym\Bigg\{
	\sum_{m=1}^n \sum_{\sum_{k=1}^m n_k=n} \sum_{\bfnu}
	\rmRirf{\bfs}_{\mu,\bfnu}
	\prod_{k=1}^m \girf{\frr_k}_{\nu_k,\frm_k} \Bigg\},
\end{split}
\label{refexp-whgi}
\ee
where $\whgirf{\bfr}_{\mu,\bfmu}$ is required to be a VV Jacobi-like form of weight $\hf(n-1)$,
index\footnote{The weight 
	is obtained from the relation \eqref{lim-ancoef} by taking into account that the $y$-dependent factor 
	in the limit $y\to 1$ is proportional to $z^{1-n}$ and thus increases the weight by $n-1$.  
	The index instead follows from the index of the generating series of refined BPS indices
	which should be equal to $-\chi_r$ \eqref{defL0}, as was established in \cite{Alexandrov:2019rth}.} 
\be 
m_{\bfr}=-\frac{\kappa}{6}\,\biggl(\Nr^3-\sum_{i=1}^n\Nr_i^3\biggr),
\label{index-mr}
\ee
and the same multiplier system as $\gi{\bfr}_{\mu,\bfmu}$ (see \eqref{multrefg}), then 
a solution of \eqref{compl-gi} is obtained from these refined anomalous coefficients as 
\be
\gi{\bfr}_{\mu,\bfmu}(\tau)=
\lim_{y\xrightarrow{}1} (y-y^{-1})^{1-n} \girf{\bfr}_{\mu,\bfmu}(\tau,z).
\label{lim-ancoef}
\ee   
This is easily checked by multiplying \eqref{refexp-whgi} by $(y-y^{-1})^{1-n}$ and taking the unrefined limit.
As a result, we have reformulated the problem of solving one anomaly equation in terms of solving another equation and subsequent evaluation of the unrefined limit.
Importantly, the relation \eqref{lim-ancoef} implies that the unrefined limit exists only if the refined solution
has a zero of order $n-1$ at $z=0$. Although it might be non-trivial to ensure this property, 
for generic set of charges this reformulation makes the problem more feasible. 

Finally, we note that the refined anomalous coefficients can be chosen to 
satisfy a property similar to \eqref{kappato1}, namely,
\be
\girf{\kappa;\bfr}_{\mu,\bfmu}(\tau,z)=\girf{1;\kappa\bfr}_{\mu,\bfmu}(\tau,z/\kappa).
\label{kappato1-ref}
\ee

\subsection{Lattice extension}
\label{subsec-extlat}

The next step is to reformulate the anomaly equation \eqref{refexp-whgi} 
for the refined anomalous coefficients in a way that involves an extended lattice 
possessing a set of null vectors.
To this end, let us introduce:
\begin{itemize}
	\item 
    integer valued function $d_\Nr$ of the magnetic charge (and intersection number $\kappa$)
    such that $d_\Nr\geq 2$;
    \item 
    $d_\Nr$-dimensional vectors $\frt^{(\Nr)}$ such that their components 
    are all non-vanishing integers and sum to zero, $\sum_{\alpha=1}^{d_\Nr}\frt^{(\Nr)}_\alpha=0$.
\end{itemize}
Note that if $d_\Nr$ could be equal to 1, it would be impossible to satisfy the last condition on $\frt^{(\Nr)}$.
The main features of the construction below do not depend on a specific form of $d_\Nr$ and $\frt^{(\Nr)}$, 
and we return to their choice, which is important for the concrete form of the solution, in the end of the subsection.

Let us now consider the anomaly equation
\be
\whchgirf{\bfr}_{\mu,\bfmu}(\tau,z,\bfz) 
=
\Sym \Bigg\{
\sum_{m=1}^n \sum_{\sum_{k=1}^m n_k=n} \sum_{\bfnu}
\rmRirf{\bfs}_{\mu,\bfnu}(\tau,z)
\prod_{k=1}^m \chgirf{\frr_k}_{\nu_k,\frm_k} (\tau,z,\frz_k),
\Bigg\}.
\label{extmodan-manyz}
\ee
where $\frz_k=(z_{j_k+1},\dots,z_{j_{k+1}})$.
Formally it looks the same as \eqref{refexp-whgi}. 
However, there are two differences. First, the new functions $\chgirf{\bfr}_{\mu,\bfmu}$ 
and their completions $\whchgirf{\bfr}_{\mu,\bfmu}$
depend on a vector of additional refinement parameters $\bfz=(z_1,\dots,z_n)$.
Second, we change the normalization for the case $n=1$ which now reads
\be 
\chgirf{\Nr}_{\mu,\mu'}(\tau,z,z')=\delta_{\mu,\mu'}\, 
\prod_{\alpha=1}^{d_\Nr} \theta_1(\tau,\frt^{(\Nr)}_\alpha z'),
\label{newnorm-g}
\ee 
where $\theta_1(\tau,z)$ is the standard Jacobi theta function \eqref{free-theta-1}.
The additional factor in \eqref{newnorm-g} leads to a change in the modular properties of 
$\chgirf{\bfr}_{\mu,\bfmu}$ compared to $\girf{\bfr}_{\mu,\bfmu}$: they
should be higher depth {\it multi-variable} Jacobi-like forms 
of the weight, index (which is now a matrix since there are several elliptic arguments) 
and multiplier system specified in \eqref{multsys-phi}, which can be easily obtained by combining 
\eqref{multrefg} with the modular properties of the  Jacobi theta function given in \eqref{multi-theta-N}.

The important property of the system of equations \eqref{extmodan-manyz} is that any solution
that is regular at $\bfz=0$ gives rise to a solution of \eqref{refexp-whgi}
with the required modular properties. The relation between the two solutions is given by
\be
\girf{\bfr}_{\mu,\bfmu}(\tau,z)= 
\frac{1}{\(-2\pi\eta^3(\tau)\)^{d_\bfr}} 
\(\prod_{i=1}^{n}\frac{\cD^{(d_{\Nr_i})}_{\hf(\frt^{(\Nr_i)})^2}(z_i)}
{d_{\Nr_i}!\prod_{\alpha=1}^{d_{\Nr_i}}\frt^{(\Nr_i)}_\alpha}\) 
\chgirf{\bfr}_{\mu,\bfmu}|_{\bfz=0},
\label{recover-gref}
\ee
where $d_\bfr=\sum_{i=1}^n d_{\Nr_i}$ and 
the differential operators $\cD_m^{(n)}$ are defined in \eqref{defcDmn}.
Indeed, due to Proposition \ref{prop-Jacobi-n} and the fact that $\theta_1(\tau,z)$ and $\eta^3(\tau)$
have identical multiplier systems, 
the product of the differential operators in \eqref{recover-gref} acting on the completion $\whchgirf{\bfr}_{\mu,\bfmu}$ 
produces a Jacobi-like form with weight, index and multiplier system as in \eqref{multrefg}.
Then to see that $\girf{\bfr}_{\mu,\bfmu}$ defined by \eqref{recover-gref} 
satisfies the anomaly equation \eqref{refexp-whgi}, it is sufficient 
to apply this product of the differential operators
to \eqref{extmodan-manyz} and use the fact that each differential operator acts 
only on one of the functions $\chgirf{\frr_k}$ on the r.h.s. of this equation.\footnote{It was 
	to ensure this factorization property that was the main reason for introducing the additional 
	refinement parameters $z_i$ for each magnetic charge.} 
Finally, the standard normalization for the case $n=1$ is reproduced due to the property 
\be
\frac{\cD^{(d_{\Nr})}_{\hf(\frt^{(\Nr)})^2}}
{d_{\Nr}!\prod_{\alpha=1}^{d_{\Nr}}\frt^{(\Nr)}_\alpha}
\prod_{\alpha=1}^{d_\Nr} \theta_1(\tau,\frt^{(\Nr)}_\alpha z)|_{z=0}=\(\p_z \theta_1(\tau,0)\)^{d_\Nr}
=\(-2\pi\eta^3(\tau)\)^{d_\Nr}.
\ee

The main advantage of the new system of equations \eqref{extmodan-manyz} compared to \eqref{refexp-whgi}
is that it corresponds to a lattice extension of the latter.
To see how it comes about, first note that 
the lattice which one sums over in the definition of $\rmRirf{\bfr}_{\mu,\bfmu}$ \eqref{Rirf-to-rmRrf}
can be defined as (see appendix \ref{ap-mainlat} for details)
\be
\bfLami{\bfr}=\left\{\bfk\in \IZ^n \ :\ \sum_{i=1}^n \Nr_i k_i=0 \right\}
\label{def-bfLam}
\ee
and carries the bilinear form
\be
\bfx\cdot\bfy=\kappa\sum_{i=1}^n \Nr_i x_i y_i.
\label{bf-r}
\ee
The new normalization \eqref{newnorm-g} then effectively gives rise to an additional sum 
over the lattice $\IZ^{d_{r_i}}$ with quadratic form $-\diag(1,\dots,1)$ associated to 
each magnetic charge $r_i$. This is especially easy to see for the term in \eqref{extmodan-manyz} with $m=n$
which contains the product of $n$ functions like \eqref{newnorm-g}.
As a result, the overall effect is that the lattice $\bfLami{\bfr}$ is extended to 
\be
\bbLami{\bfr}=\bfLami{\bfr}\oplus \IZ^{d_\bfr}
\label{extlatNr}
\ee
and the extended lattice carries the bilinear form
\be
\xbbm\ast\ybbm=\sum_{i=1}^n \(\kappa \Nr_i x_i y_i-\sum_{\alpha=1}^{d_{\Nr_i}} x_{i,\alpha} y_{i,\alpha}\),
\label{bb-r}
\ee
where $\xbbm=\{x_i,x_{i,\alpha}\}$ with $i=1,\dots ,n$ and $\alpha=1,\dots,d_{\Nr_i}$.

To discuss null vectors on $\bbLami{\bfr}$, one should first specify the function $d_\Nr$.
To motivate its choice, let us consider the case of two charges. It is easy to see that 
\be 
\bfLami{\Nr_1,\Nr_2}=\{(n\rdcr_2,-n\rdcr_1),\ n\in\IZ\},
\label{Lam2}
\ee
where $\rdcr_i=\Nr_i/\gcd(\Nr_1,\Nr_2)$,
and hence is identical to $\IZ$ with quadratic form $\kappa  \Nr \rdcr_1 \rdcr_2 x^2$.
Therefore, the norm of the vector $(1;a_1,\dots, a_1;a_2,\dots,a_2)\in 
\IZ\oplus \IZ^{d_{\Nr_1}}\oplus \IZ^{d_{\Nr_2}}\simeq \bbLami{\Nr_1,\Nr_2}$ is equal to 
$\kappa(\Nr_1\rdcr_2^2+\Nr_2\rdcr_1^2)-d_{\Nr_1}a_1^2-d_{\Nr_2}a_2^2$.
Thus, the most natural choice is to take
$d_\Nr=\kappa\Nr$ which ensures that the above vector is null for $a_1=\pm\rdcr_2$
and $a_2=\pm\rdcr_1$. However, this choice fails to satisfy the condition $\d_\Nr\geq 2$ 
for $\Nr=\kappa=1$. This introduces a complication that this particular case should be treated differently.
There are two natural ways to do this by setting
\be 
a)\ d_\Nr=\left\{\begin{array}{ll} 
	4\Nr, & \quad\kappa=1,
	\\ 
	\kappa\Nr, & \quad\kappa>1,
\end{array}\right.
\qquad
b)\ d_\Nr=\left\{\begin{array}{ll} 
	4, & \quad\kappa=\Nr=1,
	\\ 
	\kappa\Nr, & \quad\kappa \Nr>1.
\end{array}\right.
\label{defdr}
\ee 
The advantage of the second choice is that it preserves the property \eqref{kappato1-ref} and 
allows to work with lattices of smaller dimensions.
On the other hand, it is more involved at the computational level.
Therefore, in the following we proceed with the first choice 
(despite it spoils the property \eqref{kappato1-ref}).\footnote{Another possibility would be to restrict 
	to the case $\kappa=1$ and use the property \eqref{kappato1-ref}, or its unrefined analogue \eqref{kappato1}, 
	to obtain other cases. We prefer to proceed with generic $\kappa$ because, as we will see, 
	due to the additional factor of 4 in the definition of $d_r$, the case $\kappa=1$ appears to be more complicated
	than $\kappa>1$. Besides, it leads to a larger extended lattice which decreases efficiency of numerical computations.}

For the vectors $\frt^{(\Nr)}$ there are plenty of possible choices.
The following two seem to be the most ``canonical":
\be 
a)\ \frt^{(\Nr)}_\alpha=\left\{\begin{array}{ll} 
	1, & \quad \alpha\le \lceil d_\Nr/2\rceil,
	\\ 
	-1, & \quad \lceil d_\Nr/2\rceil<\alpha\le 2\lfloor d_\Nr/2\rfloor,
	\\
	-2 , & \quad \alpha=d_\Nr \mbox{ if $d_\Nr$ is odd} ,
\end{array}\right.
\qquad
b)\ \frt^{(\Nr)}_\alpha=\left\{\begin{array}{ll} 
	1, & \quad \alpha< d_\Nr,
	\\ 
	1-d_\Nr, & \quad \alpha=d_\Nr.
\end{array}\right.
\label{choicetr}
\ee 
In our calculations we will mostly use the first choice.

In the following we will use two sets of vectors belonging to the extended lattice $\bbLami{\bfr}$.
Both of them are extensions of the vectors $\bfv_{ij}\in \bfLami{\bfr}$ 
defined as in \eqref{defvij}
\be 
(\bfv_{ij})_k=\delta_{ki}\Nr_j-\delta_{kj}\Nr_i
\label{def-bfvij}
\ee 
and are given by 
\be 
\begin{split}
	(\vbbm_{ij})_k=&\,(\bfv_{ij})_k,
	\qquad\ \
	(\vbbm_{ij})_{k,\alpha}= 0,
	\\
	(\wbbm_{ij})_k=&\, 2^\eps(\bfv_{ij})_k,
	\qquad
	(\wbbm_{ij})_{k,\alpha}= (\bfv_{ij})_k,
\end{split}
\label{def-bfcvij}
\ee 
where $\eps=\delta_{\kappa-1}$.
Here the factor of $2^\eps$ compensates the factor of 4 appearing in \eqref{defdr} for $\kappa=1$
and ensures that $\wbbm_{ij}^2=0$.
We will also extensively use their normalized versions
\be 
\hvbbm_{ij}=\vbbm_{ij}/\Nr_{ij},
\qquad
\hwbbm_{ij}=\wbbm_{ij}/\Nr_{ij},
\label{def-bfcvij-hat}
\ee 
where $\Nr_{ij}=\gcd(\Nr_i,\Nr_j)$. 
Their scalar products 
with respect to the bilinear form \eqref{bb-r} are found to be
\be
\begin{split}
	&
	\vbbm_{ij}^2=2^{-\eps}\vbbm_{ij}\ast \wbbm_{ij}= \kappa \Nr_i\Nr_j(\Nr_i+\Nr_j),
	\qquad
	\vbbm_{ij}\ast \vbbm_{jk}=2^{-\eps}\vbbm_{ij}\ast \wbbm_{jk}=-\kappa \Nr_i\Nr_j\Nr_k,
	\\
	&\quad\qquad\quad
	\vbbm_{i+j,k}\ast\vbbm_{ij}=\vbbm_{i+j,k}\ast\wbbm_{ij}=\wbbm_{ij}\ast\wbbm_{kl}=0,
	\quad \mbox{for }\forall i,j,k,l,
	\\
	&\qquad\qquad\qquad\qquad
	\vbbm_{ij}\ast\vbbm_{kl}=0, 
	\qquad 
	\mbox{if }\{i,j\}\cap\{k,l\}=\emptyset,
\end{split}
\label{scpr-vw}
\ee
where $\vbbm_{i+j,k}=\vbbm_{ik}+\vbbm_{jk}$ is an extension of $\bfv_{i+j,k}=\bfv_{ik}+\bfv_{jk}$.
We will also often use the notations 
\be 
\kappa_i=\frac{\kappa\Nr \Nr_i\Nr_{i+1}}{\Nr_0\Nr_{i,i+1}}\, ,
\qquad
\kappa_{ij}=\frac{\kappa \Nr_i\Nr_j}{2\Nr_{ij}^2}\,(\Nr_i+\Nr_j),
\qquad
\kappa_{ijk}=\frac{\kappa \Nr_i\Nr_j\Nr_k}{r_{ij}r_{jk}}\,,
\label{defkapij}
\ee 
where $\Nr_0=\gcd(\Nr_1,\dots,\Nr_n)$,
encoding various scalar products of the normalized vectors.
Finally, it is useful to note the property
$\Nr_k \bfv_{ij}+\Nr_i \bfv_{jk}+\Nr_j \bfv_{ki}=0$ which implies
\be 
\begin{split} 
\Nr_k \vbbm_{ij}+\Nr_i \vbbm_{jk}+\Nr_j \vbbm_{ki}=&\, 0,
\\
\Nr_k \wbbm_{ij}+\Nr_i \wbbm_{jk}+\Nr_j \wbbm_{ki}=&\, 0.
\end{split} 
\label{rel-vw}
\ee 
Below we will see how the existence of the null vectors $\wbbm_{ij}$ gives the possibility to construct 
holomorphic theta series associated with the extended lattice and satisfying 
the anomaly equation \eqref{extmodan-manyz}.

\subsubsection{Lattice factorization}
\label{subsubsec-factor}

Before we proceed with solving the extended anomaly equation \eqref{extmodan-manyz}, let us 
perform an important technical step which will be crucial for determining a solution
that has a well-defined unrefined limit.
Namely, let us decompose 
the extended lattice $\bbLami{\bfr}$ into two orthogonal sublattices which we denote by $\bbLami{\bfr}_{||}$ 
and $\bbLami{\bfr}_{\perp}$. The former is taken to be the span of the vectors 
$\hvbbm_{ij}$ and $\hwbbm_{ij}$ introduced in \eqref{def-bfcvij-hat}, 
i.e. all their linear combinations with integer coefficients.
It is clear that it is a direct sum of two lattices
\be 
\label{defLampar}
\bbLami{\bfr}_{||}=\bfLami{\bfr}\oplus \tbfLami{\bfr},
\ee
where $\bfLami{\bfr}$ is generated by $\hvbbm_{ij}$ and is the same as \eqref{def-bfLam},
while $\tbfLami{\bfr}$ is the span of $\hubbm_{ij}=\hwbbm_{ij}-2^\epsilon\hvbbm_{ij}$
and embedded into $\IZ^{d_\Nr}$.\footnote{Note that 
	for our choice $d_\Nr=4^\eps \kappa \Nr$, one has $d_\bfr=d_\Nr$ where $\Nr=\sum_i\Nr_i$,
	which is not generally true for choice b) in \eqref{defdr}.
	This is one of several complications of the second choice.}
The embedding is given by 
\be
(k_1,\dots k_n)\; \mapsto\; (\vnc{k_1}{\di{1}},\dots, \vnc{k_n}{\di{n}}),
\label{embed-k}
\ee
where $\vnc{x}{n}$ denotes the $n$-dimensional vector with all components equal to $x$,
and the resulting lattice is actually isomorphic to $\bfLami{\bfr}$ with quadratic form rescaled by $-4^\eps$.
The lattice $\bbLami{\bfr}_{\perp}$ is taken to be generated by the vectors 
$\ebbm_0$, $\ebbm_{i,\alpha}$, with $i=1,\dots, n$ and $\alpha=1,\dots,d_{\Nr_i}-1$, given by
\be
\begin{split}
	(\ebbm_0)_k=&\,0,
	\qquad
	(\ebbm_0)_{k,\beta}=1,
	\\
	(\ebbm_{i,\alpha})_k=&\,0,
	\qquad
	(\ebbm_{i,\alpha})_{k,\beta} =\delta_{ik}(\delta_{\alpha+1,\beta}-\delta_{\alpha\beta}).
\end{split}
\label{def-basise}
\ee
Using the bilinear form \eqref{bb-r}, it is easy to check that these vectors are indeed orthogonal to 
$\hvbbm_{ij}$ and $\hwbbm_{ij}$.
Moreover, each of the sets $\{\ebbm_{i,\alpha}\}_{\alpha=1}^{d_{\Nr_i}}$
generates a lattice isomorphic to the $A_{N-1}$ root lattice with $N=d_{\Nr_i}$, and all of them are mutually orthogonal
as well as to the vector $\ebbm_0$.
Therefore, we have in addition the following orthogonal decomposition
\be 
\bbLami{\bfr}_{\perp}=\IZ\oplus \IA_{d_{\Nr_1}-1}\oplus \cdots\oplus \IA_{d_{\Nr_n}-1}.
\label{factor-perp}
\ee 
In contrast, the full lattice $\bbLami{\bfr}$ is {\it not} a direct sum of $\bbLami{\bfr}_{||}$ 
and $\bbLami{\bfr}_{\perp}$ because some of its elements require rational coefficients being decomposed 
in the basis of the two sublattices. 
In such situation, to get the full lattice from the sublattices, 
one has to introduce the so called {\it glue vectors}.

According to the general theory \cite{CSbook}, if $\oplus_{a=1}^n\bbLambda^{(a)}$ is a sublattice of 
$\bbLambda$ of the same dimension, the corresponding glue vectors
are given by the sum of representatives of the discriminant groups $\ID^{(a)}=(\bbLambda^{(a)})^*/\bbLambda^{(a)}$
which at the same time belongs to the original lattice, i.e. 
$\glueg_\Asf=\oplus_{a=1}^n\gluegi{a}_\Asf \in\bbLambda$ where $\gluegi{a}_\Asf\in \ID^{(a)}$.
The number of glue vectors is equal to 
\be
N_g=\left|\frac{\prod_{a=1}^n\det\bbLambda^{(a)}}{\det\bbLambda}\right|^{1/2}\, ,
\label{Ng}
\ee 
where $\det\bbLambda=|\bbLambda^*/\bbLambda|$ is the order of the discriminant group 
and is equal to the determinant of the matrix of scalar products of the basis elements.
The decomposition formula of the lattice $\bbLambda$ then reads
\be 
\bbLambda=\bigcup\limits_{\Asf=0}^{N_g-1} \[\mathop{\oplus}\limits_{a=1}^n \(\bbLambda^{(a)}+\gluegi{a}_\Asf\)\].
\label{lat-glue}
\ee 

In our case it takes the form
\be 
\bbLami{\bfr}=\bigcup\limits_{\Asf=0}^{N_g-1} \Bigl[\(\bbLami{\bfr}_{||}+\glueg^{||}_\Asf\)
\oplus \(\bbLami{\bfr}_{\perp}+\glueg^{\perp}_\Asf\)\Bigr].
\label{lat-glue-our}
\ee 
One finds the following lattice determinants
\be 
\begin{split}
	\det\bbLami{\bfr}=&\, (-1)^{d_\Nr}\, \det\bfLami{\bfr},
	\\
	\det\bbLami{\bfr}_{||}=&\, (-4^\eps)^{n-1}\bigl(\det\bfLami{\bfr}\bigr)^2,
	\\	
	\det\bbLami{\bfr}_{\perp}=&\,(-1)^{d_\Nr+1}d_\Nr \prod_{i=1}^n d_{\Nr_i},
\end{split}
\ee
where $\det\bfLami{\bfr}$ is evaluated in \eqref{detbfLam}.
Substituting this into \eqref{Ng}, one finds that the number of glue vectors is given by
\be 
N_g=
\frac{\Nr}{\Nr_0}\,\prod_{i=1}^{n} \di{i}.
\label{number-glue}
\ee 
There is a natural choice of glue vectors for the decomposition \eqref{lat-glue-our}.
Let us fix a $n$-tuple $\bfrho=(\rho_1,\dots,\rho_n)$ such that
$\sum_{i=1}^n \rho_i\Nr_i=\Nr_0$.
Then we represent glue vectors as a sum of several terms 
\be
\glueg_\Asf = \asf_0\glueg_0 + \sum_{i=1}^{n} \glueg_{i,\asf_i},
\label{gluegA}
\ee 
where 
\be
\begin{split}
	(\glueg_{i,\asf})_k=&\,0,
	\qquad
	(\glueg_{i,\asf})_{k,\alpha} =\delta_{ik}\sum_{\beta=1}^\asf \delta_{\alpha\beta},
	\qquad
	\glueg_0=\sum_{i=1}^n\rho_i \glueg_i, 
	\quad \glueg_i\equiv \glueg_{i,\di{i}}.
\end{split}
\label{def-g}
\ee
Thus, a glue vector is labelled by the set  
$\Asf=\{\asf_0,\asf_1,\dots,\asf_n\}$ and the indices take values in the following ranges:
$\asf_0= 0,\dots,\Nr/\Nr_0-1$ and $\asf_i= 0,\dots,\di{i}-1$.
It is trivial to see that the cardinality of the resulting set agrees with the required number
\eqref{number-glue}.\footnote{The geometric origin of these glue vectors is as follows.
	First, let us combine $\tbfLami{\bfr}$ with the one-dimensional lattice generated by $\ebbm_0$
	using the glue vectors $\asf_0\glueg_0$. It is easy to see that the result is the lattice $\IZ^n$
	with the same embedding into $\IZ^{d_\Nr}$ as in \eqref{embed-k}, i.e.
	$$
	\bigcup\limits_{\asf_0=0}^{\Nr/\Nr_0-1} \bigcup\limits_{\ell\in\IZ}
	\Bigl(\tbfLami{\bfr}+\ell\ebbm_0+\asf_0\glueg_0\Bigr)=\IZ^{n}\subset \IZ^{d_\Nr}.
	$$
	Next, one combines the $i$th factor $\IZ$ with $\IA_{\di{i}-1}$ generated by $\ebbm_{i,\alpha}$
	using the glue vectors $\glueg_{i,\asf_i}$, which gives
	$$ 
	\bigcup\limits_{\asf_i=0}^{\di{i}-1} 
	\Bigl(\IA_{\di{i}-1}+\IZ+\glueg_{i,\asf_i}\Bigr)=\IZ^{\di{i}}.
	$$
	Summing over $i$, one obtains $\IZ^{d_\Nr}$ which together with $\bfLami{\bfr}$
	produces the full extended lattice $\bbLami{\bfr}$.
}

The main application of the lattice decomposition \eqref{lat-glue} is a factorization of theta series.
Let us consider a general indefinite theta series as in \eqref{gentheta} 
with a kernel having a factorized form
$\Phi(\xbbm)=\prod_{a=1}^n\Phi_a(\xbbm^{(a)})$
where the upper index $^{(a)}$ on a vector denotes its projection to $\bbLami{a}$.
Then the lattice factorization formula \eqref{lat-glue} implies that
one can split the sum in the definition of the theta series into 
$n$ sums coupled by the additional sum over the glue vectors so that one arrives 
at the following identity for theta series
\be
\vth_{\bbmu}(\tau, \zbbm;\bbLambda, \Phi, \pbbm)=\sum_{\Asf=0}^{N_g-1}
\prod_{a=1}^n \vth_{\bbmu^{(a)}+\gluegi{a}_\Asf}(\tau, \zbbm^{(a)};\bbLami{a}, \Phi_a, \pbbm^{(a)}).
\label{factortheta-gen}
\ee

In this paper we are interested in theta series associated with the extended lattice $\bbLami{\bfr}$
and with other ingredients given by
\be
\begin{split}
	\bbmu=&\,(\bfhmu;0,\dots, 0),
	\hspace{3.3cm}
	\hmu_i=\frac{\mu_i}{\kappa\Nr_i}-\frac{\mu}{\kappa\Nr}+\frac{\rho_i\Delta\mu}{\kappa\Nr_0} \, ,
	\\
	\zbbm=&\, (\bftet z;-\frt^{(\Nr_1)}z_1;\dots;-\frt^{(\Nr_n)}z_n),
	\qquad
	\smash{\bftet=\sum_{i<j}\bfv_{ij},}
	\\
	\pbbm=&\, \(\boldsymbol{0};-1,\dots, -1\)=-\ebbm_0,
\end{split}
\label{thetadata}
\ee
where $\Delta\mu=\mu-\sum_{i=1}^n\mu_i$ as in \eqref{constr-epsi}.
Note that one has the relations
\be 
\bftet^2=-2m_\bfr,
\qquad
\bftet\cdot \bfq=\sum_{i<j}\gamma_{ij}, 
\label{rel-scpr-th}
\ee 
where $\bfq=\bigl(\frac{q_1}{\kappa \Nr_1}\, ,\dots ,\, \frac{q_n}{\kappa \Nr_n}\bigr)$
in terms of the physical charges. They ensure that the factor $e^{2\pi\I \zbbm\ast\kbbm}$ 
in the theta series reproduces the $y$-dependent factor in \eqref{Rirf-to-rmRrf}
and gives rise to the index \eqref{index-mr}. 
Let us also mention here another useful relation. The argument of the kernel in the theta series \eqref{gentheta} 
is $\xbbm=\sqrt{2\tau_2}(\kbbm+\bbbeta)$ where $\kbbm$ runs over the (shifted) lattice.
Therefore, it is useful to introduce $\xbbm_\bbbeta=\xbbm-\sqrt{2\tau_2}\,\bbbeta$ which in our case 
takes the form $\xbbm_\bbbeta=\sqrt{2\tau_2}(\bfq; k_{1,1},\dots, k_{n,\di{n}})$.\footnote{Note that
	the components of $\bfq$ do not sum to zero. Therefore, to see it as an element of $(\bfLami{\bfr})^*$,
	one should use the identification \eqref{rel-charges} silently assumed here.}
Therefore, with respect to the biliniear form \eqref{bb-r}, one finds that
\be 
\xbbm_\bbbeta\ast\vbbm_{ij}=\sqrt{2\tau_2}\gamma_{ij} .
\label{rel-scpr}
\ee

Let us now assume that a kernel $\Phi$ does not depend
on the summation along $\bbLami{\bfr}_{\perp}$, i.e. it satisfies $\Phi(\xbbm)=\Phi(\xbbm_{||})$,
where the indices $||$ and $\perp$ denote projections on $\bbLami{\bfr}_{||}$ 
and $\bbLami{\bfr}_{\perp}$, respectively.
This property is the main motivation to decompose the lattice into these two orthogonal sublattices.
Applying \eqref{factortheta-gen} to our case, one obtains the following factorization property
\be
\vth_{\bbmu}(\tau, \zbbm;\bbLami{\bfr}, \Phi, \pbbm)=
\sum_{\Asf}
\vthls{\bfr}_{\bbmu,\Asf}(\tau, z)\,
\vthps{\bfr}_\Asf(\tau,\bfz),
\label{extg-fact}
\ee
where we introduced
\be 
\begin{split}
\vthls{\bfr}_{\bbmu,\Asf}(\tau, z)=&\, \vth_{\glueg_\Asf^{||}+\bbmu}\(\tau, \zbbm_{||};\bbLami{\bfr}_{||}, \Phi, 0\),
	\\	
\vthps{\bfr}_\Asf(\tau,\bfz)=&\, \vth_{\glueg_\Asf^{\perp}}\(\tau, \zbbm_\perp;\bbLami{\bfr}_{\perp}, 1, -\ebbm_0\)
\end{split}
\label{defprojtheta}
\ee 
and took into account that $\zbbm_\perp$ is independent of $z$, while $\zbbm_{||}$ is independent of $z_i$
due to $\wbbm_{ij}\ast\zbbm=2^\eps(\bfv_{ij}\cdot \bftet) z$.
It is to achieve this property, we required 
that the components of $\frt^{(\Nr_i)}$ sum to zero, which in turn was the reason to 
introduce the additional factors of 2 in \eqref{defdr} and \eqref{def-bfcvij}.

Furthermore, due to \eqref{factor-perp}, 
the second theta series in \eqref{defprojtheta} has itself a factorized form.
To write it explicitly, let us represent the summation variable as $\kbbm_{\perp}=\ell_0\ebbm_0
+\sum_{i=1}^n\sum_{\alpha=1}^{d_{\Nr_i}-1} \ell_{i,\alpha} \ebbm_{i,\alpha}$
with $\ell_0\in\IZ+\frac{\nu_0}{d_\Nr}+\hf$ and $\ell_{i,\alpha}\in\IZ+\frac{\nu_{i,\alpha}}{d_{\Nr_i}}$.
The variables $\nu_0$ and $\nu_{i,\beta}$ determining the fractional parts
depend on the index $\Asf=\{\asf_0,\asf_1,\dots,\asf_n\}$ of glue vectors. 
The precise dependence can be determined by expanding the glue vectors $\glueg^{\perp}_\Asf$ in the basis \eqref{def-basise}. 
Starting from \eqref{gluegA}, one can then arrive at the following expressions (see \eqref{exp-g})
\be
\nu_0 =4^\eps \kappa\Nr_0  \asf_0 +\sum_{i=1}^n \asf_i,
\qquad
\nu_{i,\alpha}= \alpha\asf_{i}.  
\label{glui-extzi}
\ee
As a result, one obtains
\be
\vthps{\bfr}_\Asf(\tau,\bfz)=\vth^{(d_\Nr)}_{\nu_0(\Asf)}(\tau) 
\prod_{i=1}^{n} \vthA{d_{\Nr_i}}_{\asf_i}(\tau,z_i;\frt^{(\Nr_i)}), 
\label{factor-perptheta}
\ee
where
\bea
\label{3thetaZ}
\vth^{(d)}_{\nu_0}(\tau) &=& \vths{d,1}_{\nu_0}(\tau,0)=
\sum_{\ell_0\in \IZ +\frac{\nu_0}{d}+\hf} 
(-1)^{d\ell_0} \, \q^{\frac{d}{2}\,\ell_0^2} ,
\\
\vthA{N}_\asf(\tau,z;\frt) &=&
\(\prod_{\alpha=1}^{N-1}\sum_{\ell_\alpha\in \IZ +\frac{\alpha\asf}{N}} \)
\q^{\sum\limits_{\alpha=1}^{N-1}\( \ell_{\alpha}^2-\ell_\alpha\ell_{\alpha+1} \)}
y^{\sum\limits_{\alpha=1}^{N-1 } \(\frt_{\alpha+1}-\frt_\alpha \)\ell_{\alpha}} ,
\label{3thetas}
\eea
$\vths{d,p}_{\mu}(\tau,z)$ is the theta series \eqref{Vignerasth},
and in the last equation we used the convention $\ell_N=0$.
A nice feature of this representation is that $z_i$ appears only in 
the theta series defined by the corresponding $A_{N-1}$ lattice.
This significantly simplifies recovering the refined anomalous coefficients $\girf{\bfr}_{\mu,\bfmu}(\tau,z)$
by means of \eqref{recover-gref} because each differential operator acts only on one theta function 
(see \eqref{res-gref}, \eqref{refg3} and \eqref{refgr} below).

Finally, let us note that $n(n-1)/2$ vectors $\hvbbm_{ij}$ (or $\bfhv_{ij}=\bfv_{ij}/\Nr_{ij}$) form an 
overcomplete basis of the $n-1$ dimensional lattice $\bfLami{\bfr}$.
On the other hand, if we restrict to the set $\{\hvbbm_{12},\dots,\hvbbm_{n-1,n}\}$, in general, 
it is {\it not} a basis and its span is only a sublattice in $\bfLami{\bfr}$.
A relation between the two lattices can be described using the same formalism as above based on glue vectors. 
More precisely, one can show that
\be 
\bfLami{\bfr}=
\bigcup\limits_{\bsf_{ij}=0}^{N_{ij}-1} 
\Biggl(\Span\{\hvbbm_{k,k+1}\}_{k=1}^{n-1}+\sum_{i,j=1\atop j-i>1}^{n}\bsf_{ij} \hvbbm_{ij}\Biggr),
\label{recover-L}
\ee
where the number of {\it self-glue vectors} $N_{ij}$ is given in \eqref{valNij}.
The idea is that linear combinations with integer coefficients of $\hvbbm_{i,i+1}$ allow to get 
only multiples of $N_{i,i+2}\hvbbm_{i,i+2}$ and one needs to add $N_{i,i+2}$ glue vectors to get arbitrary multiples of
$\hvbbm_{i,i+2}$. Similarly, all these vectors can be used to get only multiples of $N_{i,i+3}\hvbbm_{i,i+3}$, etc.
Proceeding by iterations, one recovers all vectors of the lattice $\bfLami{\bfr}$.
It is obvious that a formula similar to \eqref{recover-L} holds for $\tbfLami{\bfr}$
with $\hvbbm_{ij}$ replaced by $\hubbm_{ij}$.

Although below we will present a general solution to the extended anomaly equation \eqref{extmodan-manyz}
which will be the subject of the factorization developed in this subsection,
it is instructive first to consider the cases of two and three charges.

\subsection{Two charges}
\label{subsec-case-n=2}

\subsubsection{General solution}

Let us first analyze in detail the case of two charges.
We will use the notations introduced in \S\ref{ap-2charges}:
$\Nr_0=\gcd(\Nr_1,\Nr_2)$, $\rdcr_i=\Nr_i/\Nr_0$, $\rdcr=\Nr/\Nr_0$
and $\kappa_{12}$, $\mu_{12}$ defined in \eqref{defmu0}.\footnote{This definition of $\kappa_{12}$ valid for $n=2$
	obviously agrees with the general one in \eqref{defkapij}.}
Specifying  the anomaly equation \eqref{extmodan-manyz} to $n=2$ and 
substituting the result \eqref{exprR2r} for $\rmRirf{\Nr_1,\Nr_2}_{\mu, \mu_1, \mu_2}$, one finds
\bea
\whchgirf{\Nr_1,\Nr_2}_{\mu, \mu_1, \mu_2} &=& 
\chgirf{\Nr_1,\Nr_2}_{\mu, \mu_1, \mu_2}
+\prod_{i=1}^2\(\prod_{\alpha=1}^{d_{\Nr_i}}\theta_1(\tau,\frt^{(d_{\Nr_i})}_\alpha z_i)\)
\rmRirf{\Nr_1,\Nr_2}_{\mu, \mu_1, \mu_2}
\nn\\
&=&
\chgirf{\Nr_1,\Nr_2}_{\mu, \mu_1, \mu_2}
+ \frac14\, \delta^{(\kappa\Nr_0)}_{\Delta\mu}
\sum_{\sigma=\pm 1}\sum_{k\in \IZ+\frac{\mu_{12}}{2\kappa_{12}}} 
\(\prod_{i=1}^2\prod_{\alpha=1}^{d_{\Nr_i}}\sum_{k_{i,\alpha}\in\IZ+\hf} \)
\biggl[E_1\bigl(2\sqrt{\kappa_{12}\tau_2}(\sigma k+\Nr_0 \beta)\bigr) 
\nn\\
&&\qquad
- \sgn (\sigma k) \biggr] (-1)^{\pbbm\ast\kbbm}
\q^{-\hf\kbbm^2} e^{2\pi \I \zbbm_\sigma \ast \kbbm},
\label{ext-cpmlg2}
\eea
where $\Delta\mu$ is the difference of residue classes defined in \eqref{constr-epsi} 
and we introduced $d_{\bfr}+1$-dimensional vectors 
$\kbbm=(k;k_{1,1},\dots,k_{1,d_{\Nr_1}};k_{2,1},\dots,k_{2,d_{\Nr_2}})$, 
$\pbbm=(0;-1,\dots,-1)$ and 
$\zbbm_\sigma=(\sigma\Nr_0 z ;-\frt^{(\Nr_1)}z_1;-\frt^{(\Nr_2)}z_2)$, 
which are contracted using the bilinear form
\be
\kbbm\ast\kbbm'=2\kappa_{12} kk'-\sum_{i=1}^2\sum_{\alpha=1}^{d_{\Nr_i}} k_{i,\alpha}k'_{i,\alpha}.
\label{qfbfk-zalpha}
\ee
This bilinear form is the image of \eqref{bb-r} upon the isomorphism 
$\bbLami{\Nr_1,\Nr_2}\simeq \IZ\oplus \IZ^{d_{\Nr_1}}\oplus \IZ^{d_{\Nr_2}}$
implied by \eqref{Lam2}. Under the same isomorphism, the vectors \eqref{def-bfcvij-hat} become
\be
 \hvbbm_{12}=\(1;\vnc{0}{d_{\Nr_1}};\vnc{0}{d_{\Nr_2}}\),
	\qquad
 \hwbbm_{12}=\bigl(2^\eps;\vnc{\rdcr_2}{d_{\Nr_1}};-\vnc{\rdcr_1}{d_{\Nr_2}}\bigr),
\label{def-v0-2}
\ee 
where we used the same notation as in \eqref{embed-k}.
Using $\hvbbm_{12}$, or $\vbbm_{12}=\Nr_0\hvbbm_{12}$, 
the argument of the error function can be rewritten as $2\sqrt{\kappa_{12}\tau_2}(k+\sigma\Nr_0 \beta)
=\sqrt{2\tau_2}(\kbbm+\bbbeta_\sigma)\ast\vbbm_{12}/\sqrt{\vbbm_{12}^2}$ where 
we have done the usual decomposition $\zbbm_\sigma=\bbalpha_\sigma-\tau\bbbeta_\sigma$.
As a result, the second term in \eqref{ext-cpmlg2}, up to a $\sigma$-dependent factor 
and a $\beta$-dependent shift in the argument of the sign function, 
acquires the form of the theta series \eqref{gentheta}
associated with the lattice $\bbLami{\Nr_1,\Nr_2}$, residue class 
$\bbmu=\frac{\mu_{12}}{2\kappa_{12}}\hvbbm_{12}$
and kernel 
\be
\PhiRi{\Nr_1,\Nr_2}(\xbbm)=E_1\(\frac{\xbbm\ast\vbbm_{12}}{||\vbbm_{12}||}\)-\sgn(\xbbm_\bbbeta\ast\vbbm_{12}),
\label{kerPhi1?}
\ee
where $||\vbbm||=\sqrt{\vbbm^2}$ is the norm of a vector and $\xbbm_\bbbeta=\xbbm-\sqrt{2\tau_2}\,\bbbeta_\sigma$.
More precisely, we get
\be
\whchgirf{\Nr_1,\Nr_2}_{\mu, \mu_1, \mu_2} = 
\chgirf{\Nr_1,\Nr_2}_{\mu, \mu_1, \mu_2}
+ \frac14\, \delta^{(\kappa\Nr_0)}_{\Delta\mu}
\sum_{\sigma=\pm 1}\sigma \,\vth_{\bbmu}(\tau, \zbbm_\sigma; \bbLami{\Nr_1,\Nr_2}, \PhiRi{\Nr_1,\Nr_2},\pbbm).
\label{ext-cpmlg2th}
\ee

The theta series is not modular because the kernel $\PhiRi{\Nr_1,\Nr_2}(\xbbm)$ fails to satisfy 
the Vign\'eras equation \eqref{Vigdif} due to the presence of the sign function and its weird argument.
This is supposed to be cured by the first term in \eqref{ext-cpmlg2}, which therefore should also be taken 
in the form of a theta series. However, since it must be holomorphic, its kernel must be 
a difference of two sign functions, as required by convergence. 
The first of them can be taken exactly as the one in \eqref{kerPhi1?} so that it cancels 
the sign function spoiling modularity in $\PhiRi{\Nr_1,\Nr_2}$.
But then one remains with the second sign function, say $\sgn(\xbbm\ast\vbbm)$.
It also spoils modularity unless the vector $\vbbm$ is null and, in particular, 
can be identified with $\wbbm_{12}$! 
This is due to the property \eqref{Phinull} of the (generalized) error functions
which is easy to see in the present example: the error function $E_1$ that satisfies the Vin\'eras equation
depends on the normalized vector (see \eqref{kerPhi1?}) and when its norm goes to zero, 
the argument becomes large and $E_1$ reduces to the sign function.

This reasoning implies that a general solution to \eqref{ext-cpmlg2th} is given by 
\be
\chgirf{\Nr_1,\Nr_2}_{\mu, \mu_1, \mu_2}=\chphi^{(\Nr_1,\Nr_2)}_{\mu, \mu_1, \mu_2}
+ \frac14\, \delta^{(\kappa\Nr_0)}_{\Delta\mu}
\sum_{\sigma=\pm 1}\sigma \,\vth_{\bbmu}(\tau, \zbbm_\sigma;\bbLami{\Nr_1,\Nr_2}, \Phii{\Nr_1,\Nr_2}, \pbbm),
\label{sol-cpmlg2th}
\ee
where 
\be
\Phii{\Nr_1,\Nr_2}(\xbbm)=\sgn(\xbbm_\bbbeta\ast\vbbm_{12})-\sgn(\xbbm\ast\wbbm_{12})
\label{kerPhi1}
\ee
and 
$\chphi^{(\Nr_1,\Nr_2)}_{\mu, \mu_1, \mu_2}(\tau, z,\bfz)$ is a holomorphic Jacobi-like form
with the same modular properties as $\whchgirf{\Nr_1,\Nr_2}_{\mu, \mu_1, \mu_2}$.
It represents an inherent ambiguity of solution of the anomaly equation and will be fixed later
by requiring the correct unrefined limit.
The convergence of the theta series is ensured by Theorem \ref{th-conv} and the fact that 
$\vbbm_{12}\ast\wbbm_{12}>0$ (see \eqref{scpr-vw}), while using \eqref{mult-genth}
it is straightforward to check that the weight, index and multiplier system agree with \eqref{multsys-phi}.

\subsubsection{Holomorphic modular ambiguity}
\label{subsubsec-modamb2}

In contrast to the original problem \eqref{compl-gi}, not every solution for $\chgirf{\bfr}_{\mu,\bfmu}$
suits our purposes. The restriction to be imposed is that it must have a well-defined unrefined limit.
More precisely, $\chgirf{\Nr_1,\Nr_2}_{\mu, \mu_1, \mu_2}$ must be regular at $z_i=0$ 
and have a first order zero at $z=0$. It is this condition that should be used to fix the 
holomorphic modular ambiguity $\chphi^{(\Nr_1,\Nr_2)}_{\mu, \mu_1, \mu_2}$.
As we will see below, the second term in \eqref{sol-cpmlg2th} is finite at small $z_i$,
but has a pole at small $z$, so that $\chphi^{(\Nr_1,\Nr_2)}_{\mu, \mu_1, \mu_2}$ does have to be non-trivial.
To extract the pole explicitly, we proceed in several steps.

\paragraph{Factorization and split}\

\noindent
First, note that all ingredients of the theta series in \eqref{sol-cpmlg2th} 
are as in \eqref{thetadata} subject to the isomorphism \eqref{Lam2} and the kernel $\Phii{\Nr_1,\Nr_2}$ 
depends only on the projection $\xbbm_{||}$. 
This allows to apply the factorization property \eqref{extg-fact}.
To this end, note that the theta series $\vthls{\bfr}_{\bbmu,\Asf}$ \eqref{defprojtheta} 
is given by a sum over a two-dimensional lattice 
$\kbbm_{||}=\ell\hvbbm_{12}+\tell\hubbm_{12}$
with $\ell\in\IZ+\frac{\nu}{2\kappa_{12}}$ and $\tell\in\IZ+\frac{\tnu}{2\kappa_{12}}$
where the dependence of the variables $\nu$ and $\tnu$ 
on the index $\Asf=\{\asf_0,\asf_1,\asf_2\}$ of glue vectors follows from \eqref{glueexp2}
and is given by
\be
\nu =\mu_{12},
\qquad	
\tnu =\kappa\Nr_0 \rdcr_1\rdcr_2(\rho_1-\rho_2)\asf_0+4^{-\eps}(\rdcr_2\asf_1-\rdcr_1\asf_2).  
\label{nutnu}
\ee
Therefore, \eqref{sol-cpmlg2th} can be written as 
\be
\chgirf{\Nr_1,\Nr_2}_{\mu, \mu_1, \mu_2}=\chphi^{(\Nr_1,\Nr_2)}_{\mu, \mu_1, \mu_2}
+\frac14\, \delta^{(\kappa\Nr_0)}_{\Delta\mu}
\sum_{\Asf}\(\sum_{\sigma=\pm 1}\sigma \,
\vthls{\kappa_{12}}_{\mu_{12},\tnu(\Asf)}(\tau,\sigma \Nr_0 z) \)\vthps{\bfr}_\Asf(\tau,\bfz).
\label{extg12-fact}
\ee
where $\vthps{\bfr}_\Asf$ is given by \eqref{factor-perptheta}
and 
\be 
\vthls{\kappa}_{\nu,\tnu}(\tau,z)=
\sum_{\ell\in \IZ + \frac{\nu}{2\kappa}} 
\sum_{\tell\in \IZ+ \frac{\tnu}{2\kappa}} 
\Bigl(\sgn(\ell ) - \sgn (\ell- 2^\eps \tell+\beta)\Bigr)
\, \q^{\kappa (4^\eps\tell^2-\ell^2)} y^{2\kappa \ell}.
\label{defvth2}
\ee

Next, we split the theta series \eqref{defvth2} into two parts,
$\vthls{\kappa}_{\nu,\tnu}=
\cvths{\kappa}_{\nu,\tnu}+\tvths{\kappa}_{\nu,\tnu}$, 
where in the first term one sums only over $(\ell,\tell)$ 
satisfying the condition $\ell= 2^\eps \tell$, which can also be written in 
geometric terms as
\be  
\kbbm_{||}\ast \wbbm_{12}=0,
\label{zeromode1}
\ee 
while in the second the sum goes over the rest of the lattice.
Then in $\tvths{\kappa}_{\nu,\tnu}$, for sufficiently small $z$ 
one can drop the shift by $\beta$ in the second sign function
and one obtains
\be  
\sum_{\sigma=\pm 1}\sigma \,\tvthls{\kappa}_{\nu,\tnu}(\tau, \sigma\Nr_0 z)=
\!\!\!\!\sum_{{\ell\in \IZ + \frac{\nu}{2\kappa}\atop 
	 \tell\in \IZ+ \frac{\tnu}{2\kappa}}\; : \; \ell\ne  2^\eps \tell }  \!\!\!\!
\Bigl(\sgn(\ell ) - \sgn (\ell- 2^\eps \tell)\Bigr)
\, \q^{\kappa (4^\eps\tell^2-\ell^2)}\Bigl(y^{2\Nr_0\kappa\ell  }-y^{-2\Nr_0\kappa\ell}\Bigr).
\label{nonzm-yy}
\ee 
This theta series is not only convergent for all $z$, but also 
vanishes at $z=0$. Thus, it has a well-defined unrefined limit and it remains to analyze only 
the function $\cvths{\kappa}_{\nu,\tnu}$ which we call ``zero mode contribution".

\paragraph{Pole evaluation}\

\noindent
The zero mode contribution is characterized by the condition $\ell= 2^\eps \tell$.
Importantly, it also restricts the set of glue vectors by requiring 
$\nu-2^\eps \tnu \in 2\kappa_{12}\IZ$ where $\nu$ and $\tnu$ are given in \eqref{nutnu}. 
We denote the set of the glue vectors satisfying this condition 
by $\cAr_0(\mu_{12})$ and find it explicitly in appendix \ref{subsec-glue2}.

Implementing the zero mode condition in \eqref{defvth2}, one finds 
\be
\sum_{\sigma=\pm 1}\sigma \,\cvths{\kappa_{12}}_{\nu,\tnu}(\tau,\sigma\Nr_0 z) 
=-\sum_{\sigma=\pm 1}\sum_{\tell\in \IZ+ \frac{\tnu}{2\kappa_{12}}} 
\Bigl(\sgn (\beta)-\sigma\sgn(\tell) \Bigr) 
\, y^{\sigma 2^{1+\eps} \Nr_0 \kappa_{12} \tell }.
\label{geompr}
\ee
This is just a simple geometric progression.
Assuming that $\beta>0$, so that $\Im z<0$ and $|y|>1$, it evaluates to
\be
-\sum_{\sigma=\pm 1}\(2\,\frac{y^{\sigma 2^{1+\eps}\Nr_0\kappa_{12}\(\frac{\tnu}{2\kappa_{12}}
		-\left\lceil\frac{\tnu}{2\kappa_{12}}\right\rceil+\frac12\,(1-\sigma)\)}}
{1-y^{-2^{1+\eps}\Nr_0 \kappa_{12}}} -\sigma\delta_{\tnu}^{(2\kappa_{12})}\)
=-2\,\frac{y^{2^{1+\eps}\Nr_0 \kappa_{12}\lambda_{12} }+y^{-2^{1+\eps}\Nr_0 \kappa_{12}\lambda_{12} }}
{y^{2^{\eps}\Nr_0 \kappa_{12}}-y^{-2^{\eps}\Nr_0 \kappa_{12}}}\, ,
\label{zmod-sigmap1}
\ee
where we defined
\be 
\lambda_{12} = \left\lceil \frac{\tnu}{2\kappa_{12}} \right\rceil-\frac{\tnu}{2\kappa_{12}}-\hf\, ,
\label{def-lam0}
\ee
which depends on the glue vector index $\Asf$ through \eqref{nutnu}.\footnote{In fact, due to the zero mode condition,
for $\eps=0$, i.e. $\kappa>1$, $\tnu$ is uniquely fixed by the residue class $\tnu=\mu_{12}$, while
for $\eps=1$ it can take two values $\tnu=\mu_{12}/2$ and $\mu_{12}/2+\kappa_{12}$.} 
The same result actually holds for $\beta<0$ as well.
Thus, the zero mode contribution to \eqref{extg12-fact}
is given by
\be
\begin{split}
&
-\hf\,\delta^{(\kappa\Nr_0)}_{\Delta\mu} \sum_{\Asf \in \cAr_0}  
\frac{y^{2^{1+\eps}\Nr_0 \kappa_{12}\lambda_{12} }+y^{-2^{1+\eps}\Nr_0 \kappa_{12}\lambda_{12} }}
{y^{2^{\eps}\Nr_0 \kappa_{12}}-y^{-2^{\eps}\Nr_0 \kappa_{12}}}
\, \vthps{\bfr}_\Asf(\tau,\bfz)
\\
&
=-\frac{\delta^{(\kappa\Nr_0)}_{\Delta\mu} }{2^{2+\eps} \pi \I \Nr_0 \kappa_{12} z}
\sum_{\Asf \in \cAr_0} 
\(1 +\frac{2}{3}\(1- 12 \lambda_{12}^2\) \(2^\eps\pi \Nr_0 \kappa_{12} z\)^2 +O(z^4)\)  
\vthps{\bfr}_\Asf(\tau,\bfz)
\end{split}
\label{extzi-pole}
\ee
and confirms our claim that it has a pole at $z=0$ which needs to be cancelled by a proper choice of the 
Jacobi-like form $\chphi^{(\Nr_1,\Nr_2)}_{\mu, \mu_1, \mu_2}$.

\paragraph{Fixing the ambiguity}\

\noindent
The result \eqref{extzi-pole} for the singular contribution of the indefinite theta series suggests that 
the VV Jacobi-like form representing the holomorphic modular ambiguity can be chosen in a similar form:
\be
\chphi^{(\Nr_1,\Nr_2)}_{\mu, \mu_1, \mu_2}=
\frac12\, \delta^{(\kappa\Nr_0)}_{\Delta\mu}\, \phi^{(\kappa_{12})}(\tau,\Nr_0 z)\sum_{\Asf \in \cAr_0}  
\vthps{\bfr}_\Asf(\tau,\bfz),
\label{def-phik}
\ee
where $\phi^{(\kappa_{12})}(\tau,z)$ is already a scalar valued Jacobi-like form whose
modular properties can be obtained from \eqref{multsys-phi} and \eqref{def-phik}.
Using that $\frac{\kappa}{6} \(\Nr^3 - \Nr_1^3 - \Nr_2^3\)=\Nr_0^2 \kappa_{12}$,
one finds that it should have weight 1, index $-\kappa_{12}$
and a trivial multiplier system. The last fact follows from the observation that 
the leading coefficient in the small $z$ expansion of a VV Jacobi-like form has the same multiplier system as the form itself.

It is easy to find a function with the required modular properties cancelling the pole in \eqref{extzi-pole}.
The simplest solution is to take
\be
\phi^{(\kappa_{12})}(\tau,z)
= \frac{e^{\frac{\pi^2}{3}\, \kappa_{12} E_2(\tau)z^2}}{2^{1+\eps}\pi \I \kappa_{12}z},
\label{choice-phi2}
\ee
where $E_2(\tau)$ is the quasimodular Eisenstein series whose modular anomaly \eqref{modtr-E2}
ensures the right index for $\phi^{(\kappa_{12})}$. Expanding this function at small $z$, one gets 
\be
	\phi^{(\kappa_{12})}(\tau,z)=
	\frac{1}{2^{1+\eps}\pi \I \kappa_{12} z}+\frac{\pi}{2^\eps 6\I}\,  E_2(\tau)z + O(z^3).
\label{exp-solphi}
\ee

Combining \eqref{extg12-fact} with \eqref{def-phik}, we arrive at the following result 
\be
	\chgirf{\Nr_1,\Nr_2}_{\mu, \mu_1, \mu_2}(\tau,z,\bfz)=\hf\,
	\delta^{(\kappa\Nr_0)}_{\Delta\mu}\sum_{\Asf}\Biggl[
	\frac12 \sum_{\sigma=\pm 1}\sigma \,\vthls{\kappa_{12}}_{\mu_{12},\tnu(\Asf)}(\tau,\sigma \Nr_0 z)
	+\delta_{\Asf \in \cAr_0}\,\phi^{(\kappa_{12})}(\tau,\Nr_0 z)    
	\Biggr]\vthps{\bfr}_\Asf(\tau,\bfz).
\label{res-extg12}
\ee
For what follows, it will be useful to undo the lattice factorization for the second term in \eqref{res-extg12}
and rewrite it as a theta series associated to the extended lattice.
This can be done at the price of having a kernel that does not combine dependence on $\tau_2$, $\kbbm$ 
and $\zbbm=\bbalpha-\tau\bbbeta$ into a single argument $\xbbm=\sqrt{2\tau_2}(\kbbm+\bbbeta)$
as in \eqref{gentheta}. Namely, one finds
\be 
\phi^{(\kappa_{12})}(\tau,\Nr_0 z)\sum_{\Asf\in \cAr_0}\vthps{\bfr}_\Asf(\tau,\bfz)=
\vth_{\bbmu}(\tau, \zbbm;\bbLami{\Nr_1,\Nr_2}, \Phidi{\Nr_1,\Nr_2}, \pbbm),
\label{phi12-theta}
\ee 
where
\be  
\Phidi{\Nr_1,\Nr_2}=
\delta_{\kbbm\ast\wbbm_{12}} \phi^{(\kappa_{12})}(\tau,\Nr_0 z)
\sum_{\eta=1}^{2^{1+\eps}\kappa_{12}}\delta_{\kbbm\ast\hvbbm_{12}-\eta}y^{-\Nr_0\eta} .
\label{Phid}
\ee 
The presence of two Kronecker symbols in the kernel ensures that there is no summation along 
$\bbLami{\Nr_1,\Nr_2}_{||}$ and implies that $\Asf\in \cAr_0(\mu_{12})$, while the sum over $\eta$ 
allows for an arbitrary residue class along $\hvbbm_{12}$ and takes into account 
the factor $2^\eps$ in the zero mode condition $\ell= 2^\eps \tell$.\footnote{In fact, as follows from \eqref{nutnu},
the residue class along $\hvbbm_{12}$ is not arbitrary but equal to $\mu_{12}$. This means that
the Kronecker symbol is non-vanishing only for $\eta=\mu_{12}$ and, if $\eps=1$, for $\eta=\mu_{12}+2\kappa_{12}$.
However, to cover more general cases considered below, it is convenient to write the sum over all possible range of $\eta$.}

\subsubsection{The unrefined limit}

Let us now reduce the solution \eqref{res-extg12} on the extended lattice to 
the anomalous coefficient $\gi{\Nr_1,\Nr_2}_{\mu, \mu_1, \mu_2}(\tau)$ we are really interested in. 
At the first step we obtain the refined anomalous coefficient $\girf{\Nr_1,\Nr_2}_{\mu, \mu_1, \mu_2}(\tau,z)$ 
using the relation \eqref{recover-gref}.
As was already mentioned, the absence of $z_i$-dependence in $\vth^{||}$ and 
the factorized form \eqref{factor-perptheta} of $\vth^{\perp}$ makes the application of \eqref{recover-gref}
almost trivial: one should simply apply each of the differential operators to the corresponding 
$A_{N-1}$ lattice theta series $\vthA{d_{\Nr_i}}_{\asf_i}(\tau,z_i)$.
This gives
\be
\begin{split}
	\girf{\Nr_1,\Nr_2}_{\mu, \mu_1, \mu_2}(\tau,z)= &\, \hf\,
	\delta^{(\kappa\Nr_0)}_{\Delta\mu}\sum_{\Asf=\{\asf_0,\asf_1,\asf_2\}}\Biggl[
	\frac12 \sum_{\sigma=\pm 1}\sigma \,\vthls{\kappa_{12}}_{\mu_{12},\tnu(\Asf)}(\tau,\sigma \Nr_0 z)
	\\
	&\, \qquad
	+ \delta_{\Asf \in \cAr_0}\,\phi^{(\kappa_{12})}(\tau,\Nr_0 z) 
    \Biggr]
    \vth^{(d_\Nr)}_{\nu_0(\Asf)}(\tau) 
    \prod_{i=1}^{2} \cD\vthA{d_{\Nr_i}}_{\asf_i}(\tau;\frt^{(\Nr_i)}),
\end{split}
\label{res-gref}
\ee
where 
\be
\cD\vthA{N}_{\asf}(\tau;\frt) =
\frac{\cD^{(N)}_{\frt^2/2}\vthA{N}_{\asf}(\tau,z;\frt)\bigr|_{z_i=0}}
{N!\(\prod_{\alpha=1}^{N}\frt_\alpha\)\(-2\pi\eta^3(\tau)\)^{N}}\,.
\label{def-cDvth}
\ee

Finally, we take the unrefined limit $z\to 0$ according to \eqref{lim-ancoef}. 
To this end, we split $\vth^{||}$ into the zero mode and the non-zero mode parts.
Their contributions are evaluated in \eqref{extzi-pole} and \eqref{nonzm-yy}, respectively, 
where in the latter equation $\kappa$ should be replaced by $\kappa_{12}$.
Using the expansion \eqref{exp-solphi}, we then get
\be
\begin{split}
	\gi{\Nr_1,\Nr_2}_{\mu, \mu_1, \mu_2}(\tau)
	=&\, \frac{\Nr_0 \kappa_{12}}{2}\,\delta^{(\kappa\Nr_0)}_{\Delta\mu}
	\sum_{\Asf=\{\asf_0,\asf_1,\asf_2\}}\Biggl[ 
	\vthpzi{\kappa_{12}}{0}_{\mu_{12},\tnu(\Asf)}(\tau) 
	\\
	&\, \qquad
	+\delta_{\Asf \in \cAr_0}\,\frac{2^\eps}{12}\(1- 12 \lambda_{12}^2
	- \frac{E_2(\tau)}{2^{1+2\eps} \kappa_{12}}\)\Biggr]
	\vth^{(d_\Nr)}_{\nu_0(\Asf)}(\tau) \prod_{i=1}^{2} \cD\vthA{d_{\Nr_i}}_{\asf_i}(\tau;\frt^{(\Nr_i)}),
\end{split}
\label{res-unref2}
\ee
where $\lambda_{12}$ is defined in \eqref{def-lam0} and 
\be
\vthpzi{\kappa}{0}_{\nu,\tnu}(\tau)=
\!\!\!\!\sum_{{\ell\in \IZ + \frac{\nu}{2\kappa}\atop 
		\tell\in \IZ+ \frac{\tnu}{2\kappa}}\; : \; \ell\ne  2^\eps \tell }  \!\!\!\!
\ell\, \Bigl(\sgn(\ell ) - \sgn (\ell- 2^\eps \tell)\Bigr)
\, \q^{\kappa (4^\eps\tell^2-\ell^2)}.
\ee 
In appendix \ref{subsec-Jacobi2} we verified for several values of $\Nr_1$, $\Nr_2$ and $\kappa$  
that the solution \eqref{res-unref2} is consistent with the one constructed in section \ref{subsec-sol2ch} 
using Hecke-like operators, which amounts to showing that their difference is a VV modular form.

\subsection{Three charges}
\label{subsec-case-n=3}

\subsubsection{General solution}

Next, we analyze the case of three charges. 
The r.h.s. of the anomaly equation \eqref{extmodan-manyz} now gets three contributions
\be
\begin{split} 
\whchgirf{\bfr}_{\mu,\bfmu} 
=&\, \chgirf{\bfr}_{\mu,\bfmu} 
+2\Sym \Biggl\{\prod_{\alpha=1}^{d_{\Nr_3}}\theta_1(\tau,\frt^{(d_{\Nr_3})}_\alpha z_3)
\sum_{\nu}
\rmRirf{\Nr_1+\Nr_2,\Nr_3}_{\mu,\nu,\mu_3}
\chgirf{\Nr_1,\Nr_2}_{\nu,\mu_1,\mu_2} \Biggr\}
\\
&\,
+\prod_{i=1}^3\(\prod_{\alpha=1}^{d_{\Nr_i}}\theta_1(\tau,\frt^{(d_{\Nr_i})}_\alpha z_i)\)
\rmRirf{\bfr}_{\mu,\bfmu}.
\end{split} 
\label{anom-n=3}
\ee
The second contribution is fixed by \eqref{exprR2r}, \eqref{sol-cpmlg2th}, \eqref{def-phik} and \eqref{phi12-theta}, 
while the third one can be written explicitly using \eqref{rmRrf3}, \eqref{qf-n=3} and \eqref{expr-R3}.
A crucial observation is that the quadratic form \eqref{defQlr} satisfies the property
\be  
Q_2(\hgam_1,\hgam_2)+Q_2(\hgam_1+\hgam_2,\hgam_3)=Q_3(\hgam_1,\hgam_2,\hgam_3).
\label{sumQ2}
\ee 
As a result, the second and third terms in \eqref{anom-n=3} can be written as 
a theta series associated with the extended lattice
$\bbLami{\bfr}$ \eqref{extlatNr} defined by three charges $\Nr_i$, $i=1,2,3$.
More precisely, one obtains
\be 
	\whchgirf{\bfr}_{\mu,\bfmu} 
	= \chgirf{\bfr}_{\mu,\bfmu} 
	+\frac14\,\delta^{(\kappa\Nr_0)}_{\Delta\mu}\Sym\bigl\{ \vth_{\bbmu}(\tau, \zbbm; \bbLami{\bfr}, \rPhiR,\pbbm) \bigr\}, 
\label{anom-n=3-th}
\ee
where the variables $\bbmu$, $\zbbm$ and $\pbbm$ are defined in \eqref{thetadata},
while the kernel is given by 
\bea
\rPhiR(\xbbm)&=&
\Phi^E_{2}(\vbbm_{1,2+3},\vbbm_{1+2,3};\xbbm) 
-\sgn(\xbbm_\bbbeta\ast\vbbm_{1,2+3})\,\sgn(\xbbm_\bbbeta\ast\vbbm_{1+2,3}) -\frac13\,\delta_{1,2+3}\delta_{1+2,3}
\label{expr-Phi2?}\\
&-&
\PhiRi{\Nr_1+\Nr_2,\Nr_3}(\xbbm)\(\sgn(\xbbm\ast\wbbm_{12}) - \Phidi{\Nr_1,\Nr_2}\)
- \PhiRi{\Nr_1,\Nr_2+\Nr_3}(\xbbm)\(\sgn(\xbbm\ast\wbbm_{23}) - \Phidi{\Nr_2,\Nr_3}\).
\nn 
\eea 
Here we replaced the vectors $\bfv_{ij}$ and $\bfx$ appearing as arguments of $\Phi^E_{2}$ in \eqref{expr-R3}
by their extended versions, abbreviated $\delta_{ij}=\delta_{\xbbm_\bbbeta\ast\vbbm_{ij}}$ 
and used $\xbbm_\bbbeta$ introduced above \eqref{rel-scpr} 
as well as the functions $\PhiRi{\Nr_i,\Nr_j}(\xbbm)$ and $\Phidi{\Nr_i,\Nr_i}$, 
which are the same as \eqref{kerPhi1?} and \eqref{Phid}, respectively, 
but with indices $12$ replaced by $ij$.\footnote{In \eqref{Phid} one should also replace $\Nr_0$ by $\Nr_{ij}$.}

The result \eqref{anom-n=3-th} suggests to take 
\be
\chgirf{\bfr}_{\mu,\bfmu} =
\chphi^{(\bfr)}_{\mu,\bfmu}
+\frac14\,\delta^{(\kappa\Nr_0)}_{\Delta\mu}\Sym\bigl\{ \vth_{\bbmu}(\tau, \zbbm; \bbLami{\bfr}, \rPhi,\pbbm) \bigr\},
\label{chgi-n=3Sym}
\ee
where 
\bea
\rPhi(\xbbm)\!
&=&\!
\Bigl(\sgn(\xbbm_\bbbeta\ast\vbbm_{1,2+3})-\sgn(\xbbm\ast\wbbm_{12})\Bigr)\!
		\Bigl(\sgn(\xbbm_\bbbeta\ast\vbbm_{1+2,3})-\sgn(\xbbm\ast\wbbm_{23})\Bigr)\!
		+\frac13\,\delta_{1,2+3}\delta_{1+2,3}
\label{expr-Phi2}\\
&+&
\Bigl(\sgn(\xbbm_\bbbeta\ast\vbbm_{1+2,3})-\sgn(\xbbm\ast\wbbm_{1+2,3})\Bigr)\Phidi{\Nr_1,\Nr_2}
+ \Bigl(\sgn(\xbbm_\bbbeta\ast\vbbm_{1,2+3})-\sgn(\xbbm\ast\wbbm_{1,2+3})\Bigr)\Phidi{\Nr_2,\Nr_3}
\nn 
\eea
and 
$\chphi^{(\bfr)}_{\mu, \bfmu}(\tau, z,\bfz)$ is a holomorphic Jacobi-like form
with the same modular properties as $\whchgirf{\bfr}_{\mu, \bfmu}$
representing the ambiguity of solution.
Indeed, the sum of the kernels \eqref{expr-Phi2?} and \eqref{expr-Phi2} involves only generalized error functions, 
sign functions of scalar products with null vectors and $\Phidi{\Nr_i,\Nr_i}$,
so that the corresponding theta series transforms as a modular form without anomaly.
It is also easy to check that the first term in $\rPhi$ satisfies the conditions of Theorem \ref{th-conv}
which ensures the convergence of the theta series. Finally, the weight, index and multiplier system 
follow from \eqref{mult-genth} and agree with \eqref{multsys-phi}.

\subsubsection{Holomorphic modular ambiguity}
\label{subsubsec-amb3}

To fix a solution for $\chgirf{\bfr}_{\mu,\bfmu}$, it remains to determine the holomorphic modular ambiguity
$\chphi^{(\bfr)}_{\mu, \bfmu}$ by requiring the existence of a well-defined unrefined limit.
This is equivalent to the condition that $\chgirf{\bfr}_{\mu,\bfmu}$ is regular at $z_i=0$ 
and has a second order zero at $z=0$. 

To investigate the behavior of the theta series in \eqref{chgi-n=3Sym} at small refinement parameters, 
we first apply the factorization property \eqref{extg-fact}, which is possible since the kernel 
\eqref{expr-Phi2} again depends only on the projection $\xbbm_{||}$.
This immediately implies the absence of divergences at small $z_i$. 
Furthermore, as the theta series $\vthps{\bfr}_{\bbmu,\Asf}$ \eqref{factor-perptheta} 
and the glue vectors $\glueg_\Asf$ \eqref{gluegA} are symmetric under permutations of charges, we can write
\be
\chgirf{\bfr}_{\mu,\bfmu} =
\chphi^{(\bfr)}_{\mu,\bfmu}
+\frac14\,\delta^{(\kappa\Nr_0)}_{\Delta\mu}\sum_{\Asf}
\Sym\bigl\{\vthls{\bfr}_{\bbmu,\Asf}(\tau, z) \bigr\}\,\vthps{\bfr}_\Asf(\tau,\bfz).
\label{chgi-n=3}
\ee

Next, we split the theta series $\vthls{\bfr}_{\bbmu,\Asf}$ 
into several parts determined by the vanishing of $\kbbm\ast \wbbm_{ij}$.
While in \S\ref{subsubsec-modamb2} there was just one null vector $\wbbm_{12}$ leading to 
the split of the theta series into two parts, now there are three different null vectors.
Hence, we define {\it zero mode order} of a contribution as the number of linearly independent
vanishing scalar products $\kbbm\ast \wbbm_{ij}$ and split
$\vthls{\bfr}_{\bbmu,\Asf}$ into parts with different zero mode order.
In our case this order can be 0, 1 or 2 because due to \eqref{rel-vw} the three null vectors are linearly dependent 
and the vanishing of two scalar products implies the vanishing of the third.
Note also that each $\Phid$ factor contains one of the vanishing conditions and thus adds 1 to the zero mode order.

In appendix \ref{subsec-ch3-zm01} we demonstrate that the symmetrization ensures that
the contributions of zero mode order equal to 0 and 1 both have a zero of second order at $z=0$.
Thus, they have a well-defined unrefined limit and it remains to analyze only the zero modes of order 2.
To this end, we decompose the lattices $\bfLami{\bfr}$ and $\tbfLami{\bfr}$ 
in $\bbLami{\bfr}_{||}=\bfLami{\bfr}\oplus \tbfLami{\bfr}$ as in \eqref{recover-L}
using $\hvbbm_{ij}$ and $\hubbm_{ij}$, respectively.
In other words, we expand
\be 
\kbbm_{||}=\sum_{i=1}^{n-1}(\ell_i\hvbbm_{i,i+1}+\tell_i\hubbm_{i,i+1}),
\label{expkparal}
\ee 
where in our case $n=3$. The coefficients satisfy $\ell_i\in\IZ+\frac{\nu_i}{\kappa_i}$
and $\tell_i\in\IZ+\frac{\tnu_i}{\kappa_i}$ with $\kappa_i$ defined in \eqref{defkapij}.
The variables $\nu_i$ and $\tnu_i$ determining the rational parts are fixed by $\bbmu$ and by the glue vectors
labelled by $\Asf=\{\asf_0,\asf_1,\asf_2,\asf_3\}$ and $\Bsf=\{\bsf_{13},\tbsf_{13}\}$.
Their explicit expressions can be obtained from \eqref{exp-g}, \eqref{exp-vw} 
and by evaluating scalar products of the vectors \eqref{glue-projectors-n} with $\bbmu$. 
In terms of the variables $\rdcr={\Nr}/{\Nr_0}$ and $\rdcr_i={\Nr_i}/{\Nr_0}$, this gives
\be 
\begin{split}
\nu_1=&\, \rdcr\mu_1-\rdcr_1\mu+\rdcr\rdcr_1\rho_1\Delta\mu  +\kappa\rdcr\,\frac{\Nr_1\Nr_3}{\Nr_{13}}\,\bsf_{13},
\\
\nu_2=&\, \rdcr_3\mu-\rdcr\mu_3-\rdcr\rdcr_3\rho_3\Delta\mu+\kappa\rdcr\,\frac{\Nr_1\Nr_3}{\Nr_{13}}\,\bsf_{13},
\\
\tnu_1=&\, 	\kappa\rdcr\,\frac{\Nr_1\Nr_3}{\Nr_{13}}\, \tbsf_{13}-\kappa\Nr_1(1-\rdcr\rho_1)\asf_0
+4^{-\eps}\((\rdcr_2+\rdcr_3)\asf_1-\rdcr_1(\asf_2+\asf_3)\),
\\
\tnu_2=&\, 	\kappa\rdcr\,\frac{\Nr_1\Nr_3}{\Nr_{13}}\, \tbsf_{13}+\kappa\Nr_3(1-\rdcr\rho_3)\asf_0
+4^{-\eps}\(\rdcr_3(\asf_1+\asf_2)-(\rdcr_1+\rdcr_2)\asf_3\).
\end{split}
\label{nu-n3}
\ee

In terms of the coefficients $\ell_i,\tell_i$, the second order zero mode condition
$\kbbm\ast \wbbm_{12}=\kbbm\ast \wbbm_{23}=0$ is equivalent to $\ell_1-2^\eps\tell_1=\ell_2-2^\eps\tell_2=0$.
Therefore, when one rewrites the expression for
the contribution of second order zero modes to $\Sym\bigl\{\vthls{\bfr}_{\bbmu,\Asf} \bigr\}$
in terms of these coefficients, the sum over $\ell_i$'s disappears. Furthermore, $\kbbm_{||}$ becomes null 
and one remains with
\be 
\Sym\Biggl\{\,\sum_\Bsf
\Dzm(\Asf,\Bsf)\,\cvths{\bfr}_{\tnu_1,\tnu_2}(z)\Biggr\}, 
\label{Symtheta-par-n3-z2}
\ee
where $\Dzm=\delta^{(\kappa_1)}_{\nu_1-2^\eps\tnu_1}\, \delta^{(\kappa_2)}_{\nu_2-2^\eps\tnu_2}$
is the Kronecker symbol imposing the second order zero mode condition and
\bea
\cvths{\bfr}_{\tnu_1,\tnu_2}(z) &=& \!\!
	\sum_{\tell_1\in\IZ+\frac{\tnu_1}{\kappa_1}}
	\sum_{\tell_2\in\IZ+\frac{\tnu_2}{\kappa_2}}
	\Biggl[	\biggl(
	\Bigl(\sgn(\tell_1)-\sgn(\beta)\Bigr) \Bigl(\sgn(\tell_2)-\sgn(\beta)\Bigr)
	+\frac13\, \delta_{\tell_1} \delta_{\tell_2} \biggr) 
	y^{2^{1+\eps}(\Nr_{12}\kappa_{12}\tell_1+\Nr_{23}\kappa_{23}\tell_2)}
	\nn\\
	&&
	+\phi^{(\kappa_{12})}(\tau,\Nr_{12} z)
	\sum_{\eta=1}^{2^{1+\eps}\kappa_{12}}\delta_{2^\eps(2\kappa_{12}\tell_1-\kappa_{123}\tell_2)-\eta}
	\sum_{\sigma=\pm 1} 
	\Bigl(\sigma \sgn(\tell_2)-\sgn(\beta)\Bigr)y^{\sigma 2^\eps\Nr_0\kappa_2\tell_2}
	\Biggr].
\label{theta-par-n3-z2}
\eea
To get this expression, we applied the permutation $1\leftrightarrow 3$ to the last term in \eqref{expr-Phi2}
before substituting \eqref{expkparal}.
The sum over $\tell_i$'s can be evaluated explicitly. First, we note that 
\be  
\sum_{\tell_1\in\IZ+\frac{\tnu_1}{\kappa_1}}\sum_{\eta=1}^{2^{1+\eps}\kappa_{12}}
\delta_{2^\eps(2\kappa_{12}\tell_1-\kappa_{123}\tell_2)-\eta}
=\delta^{(1)}_{\frac{2\kappa_{12}}{\kappa_1}\nu_1-\frac{\kappa_{123}}{\kappa_2}\nu_2}=1,
\ee
where in the first equality we used that $2^\eps\tnu_i=\nu_i\mod \kappa_i$, 
as required by the second order zero mode condition,
while the last equality follows by substituting the explicit expressions for $\nu_i$ \eqref{nu-n3}.
As a result, all sums in \eqref{theta-par-n3-z2} become simple geometric progressions
and for $\beta>0$ one finds\footnote{Although \eqref{theta-par-n3-z2} does depend on the sign of $\beta$,
	 its symmetrized version \eqref{Symtheta-par-n3-z2} does not. So all expressions written below starting from \eqref{Symthetaexp} 
	 will be valid for both signs.}
\bea
\cvths{\bfr}_{\tnu_1,\tnu_2}(z) &=& 
\frac13 \delta^{(\kappa_1)}_{\tnu_1}\delta^{(\kappa_2)}_{\tnu_2}+ 
\(\delta^{(\kappa_1)}_{\tnu_1}
-\frac{2y^{-2^{1+\eps}\Nr_{12}\kappa_{12} \lambda_1}}
{y^{2^{\eps}\Nr_{12}\kappa_{12}}-y^{-2^{\eps}\Nr_{12}\kappa_{12}}}\)
\(\delta^{(\kappa_2)}_{\tnu_2}
-\frac{2y^{-2^{1+\eps}\Nr_{23}\kappa_{23} \lambda_2}}
{y^{2^{\eps}\Nr_{23}\kappa_{23}}-y^{-2^{\eps}\Nr_{23}\kappa_{23}}}\)
\nn\\
&&\qquad
-2\phi_{12}(\tau,\Nr_{12} z)
\,\frac{y^{2^\eps\Nr_0 \kappa_2 \lambda_2} +y^{-2^\eps\Nr_0 \kappa_2 \lambda_2} }
{y^{2^{\eps-1}\Nr_0 \kappa_2}-y^{-2^{\eps-1}\Nr_0 \kappa_2}}\,,
\label{dzm-n3}
\eea
where we defined 
\be 
\lambda_i = \left\lceil \frac{\tnu_i}{\kappa_i} \right\rceil-\frac{\tnu_i}{\kappa_i}-\hf\, .
\label{def-lami}
\ee

In fact, it turns out to be convenient to symmetrize this expression with respect to 
the permutation $1\leftrightarrow 3$.
This could be done before performing the above calculations, but it can also be done directly for \eqref{dzm-n3}
because under this permutation the basis vectors in \eqref{expkparal} are mapped to each with a flip of sign,
$\hvbbm_{12}\leftrightarrow -\hvbbm_{23}$ and $\hubbm_{12}\leftrightarrow -\hubbm_{23}$,
while the glue vectors just flip their sign. This implies $\tnu_1\leftrightarrow -\tnu_2$ 
and $\lambda_1\leftrightarrow -\lambda_2-\delta^{(\kappa_2)}_{\tnu_2}$, so that
in the second term one should flip the relative signs inside the brackets and the signs in the power of $y$ in the numerators, 
whereas in the last term in \eqref{dzm-n3} one should just replace the indices 12 by 23 and 2 by 1.
As a result, if one expands at small $z$ the symmetrized expression, one obtains
\be
\hf \,\Bigl(\cvths{\bfr}_{\tnu_1,\tnu_2} +\sigma_{13}\bigl[\cvths{\bfr}_{\tnu_1,\tnu_2}\bigr] \Bigr)
=\cCr_{-2}z^{-2} + \Cvr_0(\tau)  + \Cvr_2(\tau) z^2 + O(z^4),
\label{Symthetaexp}
\ee
where we introduced functions of $\tau$
\be
\begin{split} 
\Cvr_0(\tau) =&\,\cCr_0+ \frac{\Nr_{12}\Nr_{23}(\Nr_1+\Nr_3)}{4^\eps\, 6\kappa \Nr \Nr_1 \Nr_2\Nr_3}\, E_2(\tau),
\\
\Cvr_2(\tau) =&\,\cCr_2 +\frac{\pi^2\Nr_0}{36}\,\Bigl(\Nr_{23}\kappa_1(1-12\lambda_1^2)+\Nr_{12}\kappa_2(1-12\lambda_2^2)\Bigr)E_2(\tau)
\\
&\, +\frac{\pi^2\Nr_{12}\Nr_{23}}{4^\eps\, 72\Nr \Nr_1 \Nr_3}\,\Bigl(\Nr_1^2(\Nr_1+\Nr_2)+\Nr_3^2(\Nr_2+\Nr_3)\Bigr) E_2^2(\tau)
\end{split}
\label{def-cCr}
\ee
and constant coefficients
\bea
\cCr_{-2}
&=& \frac{\Nr_2}{4^{1+\eps}\pi^2 \Nr \Nr_{12}\Nr_{23}\kappa_{12}\kappa_{23}}\, ,
\nn\\
\cCr_0 &=&  \frac13\delta^{(\kappa_1)}_{\tnu_1}\delta^{(\kappa_2)}_{\tnu_2}
+\(2\lambda_1+\delta^{(\kappa_1)}_{\tnu_1}\) \(2\lambda_2+\delta^{(\kappa_2)}_{\tnu_2}\) 
+\frac16 \(\frac{\Nr_{23}\Nr_1(1-12\lambda_1^2)}{ \Nr_{12}(\Nr_2+\Nr_3)} 
+\frac{\Nr_{12}\Nr_3(1-12\lambda_2^2)}{\Nr_{23}(\Nr_1+\Nr_2)} \),
\nn\\
\cCr_2 &=& 4^{\eps}\pi^2
\Biggl[ \frac{(\kappa\Nr_2)^2}{360}\biggl(
\frac{\Nr_{23}\Nr_1^3(\Nr^3 -(\Nr_1+\Nr_2)^3)}{\Nr_{12}^3\Nr_3(\Nr_2+\Nr_3)}\,(7-120\lambda_1^2+240\lambda_1^4)
\\ 
&&\quad 
+ \frac{\Nr_{12}\Nr_3^3(\Nr^3 -(\Nr_2+\Nr_3)^3)}{\Nr_{23}^3\Nr_1(\Nr_1+\Nr_2)}\,(7-120\lambda_2^2+240\lambda_2^4)\biggr)
-\frac19\, \Nr_{12}\Nr_{23}\kappa_{12}\kappa_{23} (1-12\lambda_1^2)(1-12\lambda_2^2)
\nn\\
&&\quad 
+\frac43\Bigl((\Nr_{12}\kappa_{12})^2 \lambda_1(1-4\lambda_1^2)\(2\lambda_2 + \delta^{(\kappa_2)}_{\tnu_2}\) 
+(\Nr_{23}\kappa_{23})^2\lambda_2(1-4\lambda_2^2)\(2\lambda_1 + \delta^{(\kappa_1)}_{\tnu_1}\)  \Bigr)
\Biggr].
\nn  
\eea

The main conclusion of all this analysis is that the only contribution of the theta series term in \eqref{chgi-n=3}
that does not have a zero of second order at $z=0$ and needs to be cancelled by the holomorphic modular ambiguity 
originates from the second order zero modes and is given by 
\be
\frac14\,\delta^{(\kappa\Nr_0)}_{\Delta\mu}\sum_{\Asf}
\Sym\biggl\{\sum_\Bsf \Dzm(\Asf,\Bsf)
\(\cCr_{-2}z^{-2} + \Cvr_0(\tau)\) \biggr\}
\,\vthps{\bfr}_\Asf(\tau,\bfz).
\label{sing-n=3}
\ee
Furthermore, it turns out that 
\be 
\Sym\biggl\{\sum_\Bsf \Dzm \cCr_0 \biggr\}=0,
\label{vanishC0}
\ee  
while the sum over $\Bsf$ in the first term can be evaluated using Corollary \ref{cor-sumB}.
Thus, one remains with 
\be
-\frac14\,\delta^{(\kappa\Nr_0)}_{\Delta\mu}\Sym\{ \cbfr \}
\(\frac{1}{z^2}-\frac{\pi^2}{3}\,m_{\bfr} E_2(\tau)\)\sum_{\Asf\in\cAr_0}\,
\vthps{\bfr}_\Asf(\tau,\bfz),
\label{sing-n=3new}
\ee
where $\cAr_0$ is the set characterized by the conditions \eqref{cond-glueA}
implementing the second order zero mode condition on $\Asf$ indices,
and  
\be  
\cbfr=-\frac{\Nr_0\Nr_2}{\Nr_{12}\Nr_{23}}\,\cCr_{-2}
=-\frac{\Nr_0}{4^{\eps}\pi^2 \kappa^2\Nr\Nr_1\Nr_3(\Nr_1+\Nr_2)(\Nr_2+\Nr_3) }\, .
\label{cr3}
\ee 
We have also taken into account that 
\be  
\frac{\Nr_{12}\Nr_{23}(\Nr_1+\Nr_3)}{4^\eps\, 6\kappa \Nr \Nr_1 \Nr_2\Nr_3}=-\frac{\pi^2}{3}\,m_{\bfr}\cCr_{-2},
\label{relcoef-Cm}
\ee
where
\be 
m_{\bfr}=-\frac{\kappa}{6}\,\biggl(\Nr^3-\sum_{i=1}^3\Nr_i^3\biggr)
=-\frac{\kappa}{2}\,(\Nr_1+\Nr_2)(\Nr_1+\Nr_3)(\Nr_2+\Nr_3).
\ee
We do not provide here a proof of the vanishing property \eqref{vanishC0} (which has been extensively checked on a computer)
because we will prove its generalization to arbitrary number of charges in \S\ref{subsec-case-gen}.
As we will see, it turns out to be a direct consequence of modularity.

The contribution \eqref{sing-n=3new} can be cancelled by the holomorphic modular ambiguity
chosen as 
\be
\chphi^{(\bfr)}_{\mu,\bfmu}=
\frac14\, \delta^{(\kappa\Nr_0)}_{\Delta\mu}
\phi^{(\bfr)}(\tau,z)
\sum_{\Asf\in\cAr_0}\,\vthps{\bfr}_\Asf(\tau,\bfz),
\label{def-phi3}
\ee
where
\be  
\phi^{(\bfr)}(\tau,z)=\frac{\Sym\{ \cbfr \}}{z^{2}}\,e^{-\frac{\pi^2}{3}\,m_{\bfr} E_2(\tau)z^2}.
\label{exp-solphi3}
\ee
Indeed, $\phi^{(\bfr)}$ is a Jacobi-like form of weight 2 and index $m_\bfr$, so that the weight and index of 
$\chphi^{(\bfr)}_{\mu,\bfmu}$ agree with \eqref{multsys-phi} specialized to $n=3$. 
The multiplier system must also agree because it is the same as the one of the leading terms in the small $z$ expansion
of a theta series with the right multiplier system. Finally,
the first two non-trivial terms in the small $z$ expansion of $\chphi^{(\bfr)}_{\mu,\bfmu}$ cancel \eqref{sing-n=3new},
which ensures that \eqref{chgi-n=3} has a zero of second order at $z=0$.

\subsubsection{The unrefined limit}

To get the anomalous coefficient $\gi{\bfr}_{\mu,\bfmu}$ from the solution \eqref{chgi-n=3} on the extended lattice, 
we again proceed in two steps. First, we apply the relation \eqref{recover-gref} which gives 
the following refined anomalous coefficient
\be
\begin{split}
	\girf{\bfr}_{\mu,\bfmu}(\tau,z) = &\, \frac14\, \delta^{(\kappa\Nr_0)}_{\Delta\mu} \sum_{\Asf}
	\(\phi^{(\bfr)}(\tau,z)\delta_{\Asf\in\cAr_0}+\Sym \bigl\{\vthls{\bfr}_{\bbmu,\Asf}(\tau, z) \bigr\}\)
	\vth^{(d_\Nr)}_{\nu_0(\Asf)}(\tau) 
	\prod_{i=1}^{3} \cD\vthA{d_{\Nr_i}}_{\asf_i}(\tau;\frt^{(\Nr_i)}),
\end{split}
\label{refg3}
\ee
where $\cD\vthA{d_{\Nr}}_{\asf}$ was defined in \eqref{def-cDvth}.

To evaluate the remaining unrefined limit, we use the results of our analysis which showed 
the existence of a zero of second order. Thus, we represent $\Sym\bigl\{\vthls{\bfr}_{\bbmu,\Asf} \bigr\}$
as a sum of three contributions corresponding to different orders of zero modes:
the vanishing order with kernel given in \eqref{nzmod-n3}, the first order with kernel given by the sum of \eqref{doublez-contr1} 
and \eqref{doublez-contr2}, and the second order given by \eqref{Symtheta-par-n3-z2} and \eqref{dzm-n3}
or its expansion \eqref{Symthetaexp}. The last contribution is to be combined with the first term in \eqref{refg3}.
Applying the relation \eqref{lim-ancoef}, one then finds
\be
\begin{split}
	\gi{\bfr}_{\mu, \bfmu}(\tau)
	=&\, \frac{1}{4}\,\delta^{(\kappa\Nr_0)}_{\Delta\mu}
	\sum_{\Asf}\biggl[\vthpzi{\bfr}{0}_{\bbmu,\Asf}(\tau) +\vthpzi{\bfr}{1}_{\bbmu,\Asf}(\tau)
	\\
	&\, \qquad
	-\Sym \biggl\{\frac{1}{16\pi^2}\sum_\Bsf \Dzm\hCvr_2(\tau)\biggl\}\biggr]
	\vth^{(d_\Nr)}_{\nu_0(\Asf)}(\tau) \prod_{i=1}^{3} \cD\vthA{d_{\Nr_i}}_{\asf_i}(\tau;\frt^{(\Nr_i)}),
\end{split}
\label{res-unref3}
\ee
where the three terms in the square brackets correspond to the three contributions described above.
For the first two given by theta series we provide explicit expressions in appendix \ref{subsec-unrefn=3},
while the function $\hCvr_2$ determining the third term in \eqref{res-unref3} is obtained
by combining the $O(z^2)$-terms in \eqref{Symthetaexp} and in the expansion of \eqref{exp-solphi3}:
\bea
\hCvr_2(\tau)&=& \Cvr_2(\tau)-\frac{\pi^4}{18}\,\cC_{-2}m_\bfr^2E_2^2(\tau)
\\
&=&\cCr_2 +\frac{\pi^2\Nr_0}{36}\,\Bigl(\Nr_{23}\kappa_1(1-12\lambda_1^2)+\Nr_{12}\kappa_2(1-12\lambda_2^2)\Bigr)E_2(\tau)
-\frac{\pi^2\Nr_{12}\Nr_{23}}{4^\eps\, 72\Nr_2}\,(\Nr+\Nr_2) E_2^2(\tau).
\nn
\eea 
 
The formula \eqref{res-unref3} represents an explicit expression for the anomalous coefficients with arbitrary three charges.
In appendix \ref{subsec-Jacobi3} we verified that for $\kappa$ and all $\Nr_i$ equal to 1, it is consistent with 
the solution proportional to the normalized generating function of $SU(3)$ VW invariants on $\IP^2$
presented in section \ref{subsec-solVW}.

\subsection{General case}
\label{subsec-case-gen}

\subsubsection{General solution}
\label{subsubsec-gengen}

Now we turn to the most general case and, as usual, we start by presenting a solution of the anomaly equation 
\eqref{extmodan-manyz}. Of course, for any set of charges this solution involves a holomorphic modular ambiguity
parametrized by a Jacobi-like form $\phi^{(\bfr)}_{\mu, \bfmu }$.
From the very beginning we will take into account that it can be chosen in the factorized form (cf. \eqref{def-phi3})
\be
\chphi^{(\bfr)}_{\mu, \bfmu}=
\frac{\delta^{(\kappa\Nr_0)}_{\Delta\mu}}{2^{n-1}}
\sum_{\Asf} \phi^{(\bfr)}_{\glueg_\Asf^{||}+\bbmu}(\tau,z) \,
\vthps{\bfr}_\Asf(\tau,\bfz), 
\label{def-phin}
\ee
where $\phi^{(\bfr)}_{\bbnu}(\tau,z)$ is a VV Jacobi-like form labelled 
by $\bbnu\in\ID^{(\bfr)}_{||}=(\bbLami{\bfr}_{||})^*/\bbLami{\bfr}_{||}$,
and characterized by weight $n-1$, index $m_{\bfr}$ \eqref{index-mr},
and the multiplier system given by the Weil representation \eqref{mult-genth} 
associated with the lattice $\bbLami{\bfr}_{||}$.
As we did before for $n=2$, formally one can rewrite \eqref{def-phin} as a theta series over the full extended lattice
$\bbLami{\bfr}$
\be 
\chphi^{(\bfr)}_{\mu, \bfmu}=\frac{\delta^{(\kappa\Nr_0)}_{\Delta\mu}}{2^{n-1}}\,
\vth_{\bbmu}(\tau, \zbbm;\bbLami{\bfr}, \Phidi{\bfr}, \pbbm),
\label{phin-theta}
\ee 
where $\bbmu$, $\zbbm$ and $\pbbm$ are as in \eqref{thetadata}, and
the kernel is given by
\be  
\Phidi{\bfr}(\xbbm,\tau,z)=\q^{\hf\kbbm_{||}^2}\,y^{-\bftet\cdot\bfk}\sum_{\bbnu\in\ID^{(\bfr)}_{||}}\delta_{\kbbm_{||}-\bbnu}
\phi^{(\bfr)}_\bbnu(\tau,z) ,
\label{Phidn}
\ee 
where $\kbbm_{||}$ is the projection of $\kbbm={\xbbm}/{\sqrt{2\tau_2}}-\bbbeta$ on $\bbLami{\bfr}_{||}$
and $\bfk$ is the component of $\kbbm$ along $\bfLami{\bfr}$.
It is chosen to ensure that
\be  
\vth_{\bbnu}(\tau, \zbbm_{||};\bbLami{\bfr}_{||}, \Phidi{\bfr},0)=\phi^{(\bfr)}_\bbnu(\tau,z).
\label{prop-Phid}
\ee

Although the restriction to \eqref{def-phin} may not describe the most general solution of the anomaly equation, 
which is not our goal anyway, it allows us to represent a solution in terms of a theta series on the extended lattice.
To this end, let us define 
\be
\rPhi(\xbbm;\{\Fvi{\bfs}\})=
	\sum_{m=2}^{n} \sum_{\sum_{k=1}^m n_k = n} 
	\Fvi{\bfs}(\xbbm^{(0)})	\prod_{k=1}^{m} \Phidi{\frr_k}(\xbbm^{(k)},\tau,z),
\label{ker-thm}
\ee
where $\bfs$ are $\frr_k$ are the notations from \eqref{split-rs},
the upper indices $^{(0)}$ and $^{(k)}$ denote projections to $\bbLami{\bfs}$ and $\bbLami{\frr_k}$, respectively,
and for a single charge we set $\Phidi{\Nr}=1$.
Then we have the following

\begin{theorem}\label{thm-gensol}
	A solution of the anomaly equation \eqref{extmodan-manyz} and its modular completion 
	can be expressed as
	\be
	\begin{split}
		\chgirf{\bfr}_{\mu, \bfmu }  = &\, \chphi^{(\bfr)}_{\mu,\bfmu}
		+\frac{\delta^{(\kappa\Nr_0)}_{\Delta\mu}}{2^{n-1}}\,
		\Sym\Bigl\{\vth_{\bbmu}(\tau, \zbbm;\bbLami{\bfr}, \rPhi(\{\Fvi{\bfs}\}), \pbbm)\Bigr\} ,
		\\ 
		\whchgirf{\bfr}_{\mu, \bfmu}  = &\,\chphi^{(\bfr)}_{\mu,\bfmu}+
		\frac{\delta^{(\kappa\Nr_0)}_{\Delta\mu}}{2^{n-1}}\,
		\Sym\Bigl\{\vth_{\bbmu}(\tau, \zbbm;\bbLami{\bfr}, \rPhi(\{\whFvi{\bfs}\}), \pbbm) \Bigr\},
	\end{split}
	\ee
	where the functions $\Fvi{\bfr}$ and $\whFvi{\bfr}$ are given by 
	\be
	\begin{split} 
	\Fvi{\bfr}(\xbbm)=&\,
	\sum_{\cJ\subseteq\Zv_{n-1}} e_{|\cJ|} \delta_\cJ
	\prod_{\ell \in \Zv_{n-1} \backslash \cJ}
	\Bigl(\sign (\xbbm_\bbbeta\ast\vbbm_\ell)-\sign ( \xbbm\ast\wbbm_{\ell,\ell+1}) \Bigr),
		\\
	\whFvi{\bfr}(\xbbm)=&\,
	\sum_{\cJ \subseteq \Zv_{n-1}}\Phi_{|\cJ|}^E \(\{\vbbm_l\}_{l \in \cJ}; \xbbm \) 
	\prod_{\ell \in \Zv_{n-1}\backslash\cJ}\bigl(-\sign(\xbbm\ast\wbbm_{\ell,\ell+1}) \bigr).
	\end{split} 
	\label{kern-manyz}
	\ee
	Here $\xbbm_\bbbeta=\xbbm-\sqrt{2\tau_2}\,\bbbeta$,
	$\Zv_{n}=\{1,\dots,n\}$, 
	\be
	\label{def-genthm}
	e_m=\left\{\begin{array}{ll}
		0 & \mbox{\rm if $m$ is odd},
		\\
		\frac{1}{m+1}\ & \mbox{\rm if $m$ is even},
	\end{array}\right.
	\qquad
	\delta_\cJ=\prod_{\ell\in\cJ}\delta_{\xbbm_\bbbeta\star\vbbm_\ell},
	\quad\mbox{\rm and}\quad
	\vbbm_\ell=\sum_{i=1}^{\ell}\sum_{j=\ell+1}^{n} \vbbm_{ij} .
	\ee 
\end{theorem}

Although the functions \eqref{kern-manyz} might seem to be complicated, their structure is easy to understand.
First, if all scalar products $\xbbm_\bbbeta\ast\vbbm_\ell$ are non-vanishing, then the function $\Fvi{\bfr}(\xbbm)$
simplifies to 
\be  
\Fvi{\bfr}(\xbbm) =\prod_{\ell =1}^{n-1}
\Bigl(\sign (\xbbm_\bbbeta\ast\vbbm_\ell)-\sign ( \xbbm\ast\wbbm_{\ell,\ell+1}) \Bigr),
\ee 
which is the standard kernel ensuring convergence of indefinite theta series 
with quadratic form having $n-1$ positive eigenvalues
(see Theorem \ref{th-conv}).
If however some of the scalar products vanish, it is not sufficient 
to set the corresponding sign functions to zero. Instead, one gets additional contributions
manifestly visible in \eqref{defST}
(for $n=3$ this the term $\frac13\,\delta_{1,2+3}\delta_{1+2,3}$ in \eqref{expr-Phi2}).
In the presence of refinement only the linear tree is relevant (see \eqref{Efref})
and one can apply a simple recipe that $\sgn(0)^m\to e_m$ \cite{Alexandrov:2019rth}. 
This gives rise to the expression in \eqref{kern-manyz}. It is useful to note that,
using \eqref{rel-scpr}, the notation \eqref{def-bij} and the function \eqref{defST},
it can also be rewritten as
\be 
\Fvi{\bfr}(\xbbm) =	\sum_{\cJ \subseteq \Zv_{n-1}}S_{\cT_{\rm lin}}(\{\bfhgam_l\}_{l \in \cJ})  
\prod_{\ell \in \Zv_{n-1}\backslash\cJ}\bigl(-\sign(\omb_{\ell,\ell+1}) \bigr).
\ee 
Now it should become obvious where the second function $\whFvi{\bfr}(\xbbm)$ comes from:
it is obtained from $\Fvi{\bfr}(\xbbm)$ by applying the recipe to construct modular completions of indefinite theta series 
explained in \S\ref{ap-generr} which amounts to replacing each product of 
sign functions\footnote{Since the large $\tau_2$ limit of the generalized error function is precisely the function 
$S_{\cT}$ and not the simple product of signs \cite{Alexandrov:2024jnu}, 
it is actually this function that should be replaced by $\Phi_{n-1}^E$.} 
by the (boosted) generalized error function with parameters determined by the vectors entering 
the sign functions. If some of the vectors are null, in addition one applies the property \eqref{Phinull}.

The proof of Theorem \ref{thm-gensol} is completely analogous to the proof of Theorem 1 in \cite{Alexandrov:2020bwg}. 
The similarity between these theorems become particularly obvious 
if one applies the factorization property \eqref{extg-fact}. 
It allows to rewrite the solution given in the theorem as
\be
\chgirf{\bfr}_{\mu,\bfmu} =\frac{\delta^{(\kappa\Nr_0)}_{\Delta\mu}}{2^{n-1}}\,\sum_{\Asf}
\(\phi^{(\bfr)}_{\glueg_\Asf^{||}+\bbmu}(\tau,z)+\Sym\bigl\{\vthls{\bfr}_{\bbmu,\Asf}(\tau, z) \bigr\}\)
\vthps{\bfr}_\Asf(\tau,\bfz),
\label{chgi-n}
\ee
where $\vthps{\bfr}_\Asf$ has the universal form \eqref{factor-perptheta}.
The key point is that $\vthls{\bfr}_{\bbmu,\Asf}$ can be significantly simplified due to the property 
\eqref{prop-Phid} of the kernels $\Phidi{\bfr}$.
However, to be able to use this property, one should somehow obtain theta series with a kernel given by only one $\Phidi{\bfr}$.
This can be achieved by factorizing the lattice $\bbLami{\bfr}_{||}$ according to 
the factorized form of the kernel $\rPhi$ in \eqref{ker-thm}, which in turn will induce the corresponding
factorization of theta series.
Since  $\bbLami{\bfr}_{||}=\bfLami{\bfr}\oplus \tbfLami{\bfr}$, this amounts to using 
\eqref{lat-decomp-r} for both factors, which leads to a generalization of \eqref{factortheta-rs}
where the theta series are associated to the product of two lattices. 
As a consequence, the indices labelling them are doubled and can be represented as $(\mu,\bfmu;\tmu,\bftmu)$.
This set of indices can be thought of as a label of the elements of the discriminant group $\bbnu\in\ID^{(\bfr)}_{||}$
which itself can be seen as a pair of vectors $(\bfhmu,\bfhtmu)$.
A relation between $(\mu,\bfmu)$ and $\bfhmu$ is given in \eqref{thetadata},
and to get its tilded version it is enough to replace $\kappa$ by $4^\eps \kappa$.
In particular, using \eqref{exp-g}, one finds that the element $\glueg_\Asf^{||}+\bbmu\in\ID^{(\bfr)}_{||}$
corresponds to $(\mu,\bfmu;\tmu,\bftmu)$ with
\be 
\tmu=4^\eps \kappa\Nr_0\asf_0 +\sum_{i=1}^n\asf_i,
\qquad 
\tmu_i=\asf_i.
\label{tildemus}
\ee 
As a result, one arrives at the following expression
\be
\vthls{\bfr}_{\bbmu,\Asf}(\tau, z)
=
\sum_{m=2}^{n} \sum_{\sum_{k=1}^m n_k = n} 
\sum_{\bfnu,\bftnu} \vth_{\mu,\bfnu;\tmu,\bftnu} (\tau, \zbbm_{||}^{(0)};\bbLami{\bfs}_{||}, \Fvi{\bfs},0) 
\prod_{k=1}^{m} \phi^{(\frr_k)}_{\nu_k,\frm_k;\tnu_k,\tfrm_k}(\tau,z).
\label{thet-phi-thm}
\ee
It is a simple generalization of \cite[Eq.(3.11)]{Alexandrov:2020bwg} which is twofold: 
the doubling of indices discussed above and arbitrary set of charges instead of $n$ charges all equal to 1.
Furthermore, the kernels \eqref{kern-manyz} are exactly the same as the ones given in Eqs. (3.13) and (3.14) of that paper.
Given this similarity, we refrain from repeating the proof.

\subsubsection{Holomorphic modular ambiguity}

As usual, it remains to fix the holomorphic modular ambiguities which are now encoded in 
the VV Jacobi-like forms $\phi^{(\bfr)}_{\bbnu}$ entering the construction via \eqref{chgi-n} and \eqref{thet-phi-thm}.
To this end, as in \S\ref{subsubsec-amb3}, we split the theta series $\vthls{\bfr}_{\bbmu,\Asf}$ 
into contributions with different zero mode order equal to the number of
linearly independent vanishing scalar products $\kbbm\ast \wbbm_{ij}$.
On the basis of our results for $n=3$ (and some checks done at $n=4$ which are too cumbersome to be presented here)
we make the following 
\begin{conj}\label{conj-zm}
Let us fix an integer $n_0$, and assume that for all sets of charges $\bfr$ with the number of charges $n<n_0$, 
the functions $\phi^{(\bfr)}_{\bbnu}(\tau,z)$ are Jacobi-like forms that ensure the existence of the unrefined limit 
for all $\chgirf{\bfr}_{\mu,\bfmu}$ so that they behave as $O(z^{n-1})$ at small $z$.	
Then for $n=n_0$, the contributions to $\Sym\bigl\{\vthls{\bfr}_{\bbmu,\Asf}(\tau, z) \bigr\}$ 
of any zero mode order different from the maximal one, given by $n-1$, behave as $O(z^{n-1})$.
\end{conj}

For the contributions of vanishing zero mode order this conjecture is equivalent to
Conjecture 1 in \cite{Alexandrov:2020bwg}, which has been extensively tested numerically,
and there is a simple argument why it is expected to hold for other zero mode orders as well.
Note that due to Proposition \ref{prop-JacobiE2}, the expansion coefficients at small $z$
of the function $e^{\frac{m_\bfr}{3}\pi^2 E_2(\tau) z^2}\,\whgirf{\bfr}_{\mu,\bfmu}$,
where $m_\bfr$ is the index \eqref{index-mr} of the refined anomalous coefficient, 
transform as modular forms without any anomaly.
On the other hand, due to the induction hypothesis, all terms in the refined anomaly equation \eqref{refexp-whgi}, 
except the one with $m=1$, behave as $O(z^{n-1})$. Thus, the only term that can spoil this behavior 
is $\girf{\bfr}_{\mu,\bfmu}$. Combining the two conclusions, one obtains that if 
\be 
\cL_{n-2}\[e^{\frac{m_\bfr}{3}\pi^2 E_2(\tau) z^2}\,\girf{\bfr}_{\mu,\bfmu}\]
\label{trunc-exp}
\ee 
is the part of the Laurent series in $z$ truncated at $z^{n-2}$, then all its coefficients 
must be modular forms. Furthermore, from the representation \eqref{chgi-n} of our solution, 
it is clear that the same is true for the coefficients of 
\be 
\cL_{n-2}\[e^{\frac{m_\bfr}{3}\pi^2 E_2(\tau) z^2}\Sym\bigl\{\vthls{\bfr}_{\bbmu,\Asf}(\tau, z) \bigr\}\].
\label{trunc-exp}
\ee
But all contributions to $\Sym\bigl\{\vthls{\bfr}_{\bbmu,\Asf}(\tau, z) \bigr\}$ 
of any zero mode order different from the maximal one involve holomorphic indefinite theta series.
They are (higher depth) mock modular forms whose anomalies cannot be cancelled by quasimodular forms.
Therefore, they cannot generate pure modular forms and  should vanish. 
For the zero modes of maximal order the situation is different because, as we will see below, their contribution
does not involve indefinite theta series.
Of course, this argument is far from being a proof and we hope to return to this issue in a future work.

The above conjecture reduces the problem to evaluating the zero mode contribution of maximal order.
To get it, let us first analyze the zero mode contribution of maximal order to 
$\vth_{\bbnu} (\tau, \zbbm_{||};\bbLami{\bfr}_{||}, \Fvi{\bfr},0) $
which we will denote by $\cIr_{\bbnu}$.
The zero mode condition implies that $n-1$ scalar products $\kbbm\ast \wbbm_{ij}$ are vanishing, 
but since among vectors $\wbbm_{ij}$ 
there are only $n-1$ linearly independent, actually all such scalar products are vanishing.
Under this condition, $\kbbm_{||}^2=0$ so that $\cIr_{\bbnu}$ 
does not depend on $\tau$ and is a function of $z$ only.
To find it explicitly, we perform the lattice decomposition \eqref{recover-L}
for the two factors in $\bbLami{\bfr}_{||}=\bfLami{\bfr}\oplus \tbfLami{\bfr}$ 
and substitute the expansion \eqref{expkparal} with $\ell_i\in\IZ+\frac{\nu_i}{\kappa_i}$
and $\tell_i\in\IZ+\frac{\tnu_i}{\kappa_i}$ where $\nu_i$ and $\tnu_i$ are determined by the residue class $\bbnu$. 
As a result, we end up with the following expression (cf. \eqref{Symtheta-par-n3-z2})
\be 
\cIr_{\bbnu}(z)=\sum_\Bsf
\Dzm(\Asf,\Bsf)\,
\cvths{\bfr}_{\bftnu(\Asf,\Bsf)}(z), 
\label{theta-mzm}
\ee
where
is the Kronecker symbol imposing the maximal order zero mode condition and
\be 
\cvths{\bfr}_{\bftnu}(z)=
\(\prod_{i=1}^n \sum_{\tell_i\in\IZ+\frac{\tnu_i}{\kappa_i}}\)
	\[\sum_{\cJ\subseteq\Zv_{n-1}} e_{|\cJ|} \prod_{i\in\cJ}\delta_{\tell_i}
\prod_{i \in \Zv_{n-1} \backslash \cJ} \Bigl(\sign (\tell_i)-1 \Bigr)\]
y^{2^{1+\eps}\sum\limits_{i=1}^{m-1}\Nr_{i,i+1}\kappa_{i,i+1}\tell_i}.
\ee 
Here for simplicity we restricted to $\beta>0$ (as usual, the final result after symmetrization
will be valid for both signs of $\beta$).
The remaining sums produce geometric progressions resulting in
\be 
\cvths{\bfr}_{\bftnu}(z)=
\sum_{\cJ\subseteq\Zv_{n-1}}e_{|\cJ|}\prod_{i\in\cJ} \delta^{(\kappa_i)}_{\tnu_i}
\prod_{i \in \Zv_{n-1} \backslash \cJ}
\(\delta^{(\kappa_i)}_{\tnu_i}
-\frac{2y^{-2^{1+\eps}\Nr_{i,i+1}\kappa_{i,i+1} \lambda_i}}
{y^{2^{\eps}\Nr_{i,i+1}\kappa_{i,i+1}}-y^{-2^{\eps}\Nr_{i,i+1}\kappa_{i,i+1}}}\),
\label{cvths}
\ee 
where we used the notation $\lambda_i$ from \eqref{def-lami}. 

Let us now concentrate on the leading singular contribution in the small $z$ limit which will allow us to fix 
the most singular term in the holomorphic modular ambiguity.
From \eqref{cvths}, one finds that 
\be
\cvths{\bfr}_{\bftnu}(z)=\frac{(-2^{1+\eps}\pi\I z)^{1-n}}{\prod_{i=1}^{n-1}\Nr_{i,i+1}\kappa_{i,i+1}}
+O(z^{2-n})\, .
\ee
Importantly, the leading term does not depend on the glue vector indices $\Bsf$.
Therefore, due to Corollary \ref{cor-sumB},
the leading term in \eqref{theta-mzm} is given by
\be 
\cIr_{\bbnu}(z)=
\frac{(-2^{\eps}\pi\I \kappa z)^{1-n}\,\Nr_0}{\prod_{i=1}^{n}\Nr_i\prod_{i=1}^{n-1}(\Nr_i+\Nr_{i+1})} 
\, \delta_{\Asf\in\cAr_0}
+O(z^{2-n})\, ,
\label{theta-mzm-lead}
\ee
where $\cAr_0$ is the set from Proposition \ref{prop-zm}
implementing the maximal order zero mode condition on $\Asf$ indices.
Next, we assume that the leading term of the Jacobi-like form $\phi^{(\bfr)}_\bbnu(\tau,z)$
also depends on the residue class $\bbnu$ only through the zero mode condition. Namely, representing 
$\bbnu=(\bfhmu,\bfhtmu)$ and defining $\Delta\bfhmu=\bfhmu-2^\eps\bfhtmu$ (cf. \S\ref{subsec-zm})
we take 
\be 
\phi^{(\bfr)}_\bbnu(\tau,z)=\frac{\Sym\{ \cbfr \}}{ z^{n-1}}
\prod_{i=1}^n \delta^{(1)}_{\Delta\hmu_i}+O(z^{2-n})\,.
\label{phi-lead}
\ee 
For $\bbnu=\glueg_\Asf^{||}+\bbmu$, the product of Kronecker symbols is nothing but $\delta_{\Asf\in\cAr_0}$.
Substituting \eqref{theta-mzm-lead} and \eqref{phi-lead} into \eqref{thet-phi-thm},
one obtains 
\be 
\vthls{\bfr}_{\bbmu,\Asf}(\tau, z)=
\frac{\delta_{\Asf\in\cAr_0}}{z^{n-1}}\sum_{m=2}^n \sum_{\sum_{k=1}^m n_k=n}\sum_{\bfnu} 
\frac{\gcd(s_1,\dots,s_m)\prod_{k=1}^m \ci{\frr_k} }{(-2^{\eps}\pi\I \kappa)^{n-1}\prod_{k=1}^m s_k \prod_{k=1}^{m-1}(s_k+s_{k+1})}
+O(z^{2-n})\,.
\ee 
Since the summand does not depend on $\bfnu$, the sum over these indices produces
just a numerical factor equal to the number of their independent values.
In appendix \ref{ap-mainlat} we showed that these indices can be identified with
the glue vectors of the lattice decomposition \eqref{lat-decomp-r} and therefore the factor is given by
\eqref{number-nus}.
Finally, the consistency requires that the resulting singular contribution should be exactly cancelled by 
adding the leading singular term \eqref{phi-lead} of the holomorphic modular ambiguity 
with $\bbnu=\glueg_\Asf^{||}+\bbmu$ as in \eqref{chgi-n}. Thus, the coefficients $\cbfr$
should be solutions of the following system of equations
\be 
\sum_{m=1}^n \sum_{\sum_{k=1}^m n_k=n} \(-\frac{1}{2^{\eps}\pi\I \kappa}\)^{m-1}
\frac{\prod_{k=1}^m \ci{\frr_k} }{\prod_{k=1}^m \gcd(\frr_k)\prod_{k=1}^{m-1}(s_k+s_{k+1})}=0.
\label{eq-cphi}
\ee 

We claim that this system is solved by
\be  
\cbfr
=\frac{\Nr_0}{(2^{\eps}\pi\I \kappa)^{n-1}\Nr}
\prod\limits_{k=1}^{n-1}\(\sum\limits_{i=1}^k\Nr_i\sum\limits_{j=n-k+1}^n\Nr_j\)^{-1}\, ,
\label{constphi}
\ee 
which agrees with \eqref{exp-solphi} and \eqref{cr3}.
Although we were not able to find an analytic proof, we have checked on a computer 
that this formula does solve the equations for arbitrary charges up to $n=8$, which leaves no doubt that 
this is the right solution.

Having found the leading term in the expansion of $\phi^{(\bfr)}_\bbnu$, we can now provide the full function.
It is given by the following 

\begin{theorem} \label{thm-ambig}
Provided Conjecture \ref{conj-zm} holds, the holomorphic modular ambiguity given by \eqref{def-phin} with 
\be  
\phi^{(\bfr)}_\bbnu(\tau,z)=\frac{\Sym\{ \cbfr \}}{ z^{n-1}}\,e^{-\frac{\pi^2}{3}\,m_{\bfr} E_2(\tau)z^2}
\prod_{i=1}^n \delta^{(1)}_{\Delta\hmu_i}
\label{solfullphi}
\ee
ensures the existence of the unrefined limit.
\end{theorem}
\begin{proof}
The proof will follow the logic of the argument used to justify Conjecture \ref{conj-zm}.
Namely, let us consider the function 
\bea 
&&
e^{\frac{m_\bfr}{3}\pi^2 E_2(\tau) z^2}
\(\phi^{(\bfr)}_{\glueg_\Asf^{||}+\bbmu}(\tau,z)+\Sym\bigl\{\vthls{\bfr}_{\bbmu,\Asf}(\tau, z) \bigr\}\)
\label{funforproof}\\
&=&
\delta_{\Asf\in\cAr_0}\Sym\biggl\{ \sum_{m=1}^{n} \frac{1}{z^{n-m}} 
\!\!\sum_{\sum_{k=1}^m n_k = n} \!\!
\,e^{\frac{\pi^2}{3}\(m_\bfr-\sum\limits_{k=1}^m m_{\frr_k}\) E_2(\tau)z^2}
\prod_{k=1}^{m} \ci{\frr_k}\sum_{\bfnu,\bftnu} \cIi{\bfs}_{\mu,\bfnu;\tmu,\bftnu}(z) \biggr\}
+O(z^{n-1}),
\nn
\eea 
where to get the second line, we used Conjecture \ref{conj-zm}, the representation \eqref{thet-phi-thm} and 
the proposed form of $\phi^{(\bfr)}_\bbnu$.
To prove the theorem, one needs to show that the first term is also $O(z^{n-1})$. 
Let us assume that this is not the case and there are terms $\sim z^k$ with $k<n-1$.
Here $k$ must be large than $1-n$ because we have already shown the cancellation of the leading singularity $\sim z^{1-n}$
for $\cbfr$ given by \eqref{constphi}. Then the same chain of reasoning as above \eqref{trunc-exp} 
allows to conclude that the coefficient of the $z^k$-term must be a modular form of weight $n-1+k>0$.
However, as follows from \eqref{funforproof}, the only dependence of $\tau$ of these coefficients is through 
the polynomial dependence on the quasimodular form $E_2(\tau)$. But no polynomial of $E_2(\tau)$
can produce a modular form of positive weight. Hence all of them must vanish, which proves the statement of 
the theorem.
\end{proof}

\subsubsection{The unrefined limit}

The two theorems \ref{thm-gensol} and \ref{thm-ambig} provide a solution for the functions $\chgirf{\bfr}_{\mu,\bfmu}$ 
satisfying the anomaly equation \eqref{extmodan-manyz} and having a well-defined unrefined limit. 
It remains just to reduce it to the original anomalous coefficients $\gi{\bfr}_{\mu,\bfmu}$. 
The first step, the reduction to the refined anomalous coefficients $\girf{\bfr}_{\mu,\bfmu}$, is trivial
and done by applying the relation \eqref{recover-gref} to the expression \eqref{chgi-n}. 
This affects only the $A_{N-1}$ lattice theta series and results in
\be
	\girf{\bfr}_{\mu,\bfmu}(\tau,z) =  \frac14\, \delta^{(\kappa\Nr_0)}_{\Delta\mu} \sum_{\Asf}
	\(\phi^{(\bfr)}_{\glueg_\Asf^{||}+\bbmu}(\tau,z)+\Sym\bigl\{\vthls{\bfr}_{\bbmu,\Asf}(\tau, z) \bigr\}\)
	\vth^{(d_\Nr)}_{\nu_0(\Asf)}(\tau) 
	\prod_{i=1}^{n} \cD\vthA{d_{\Nr_i}}_{\asf_i}(\tau;\frt^{(\Nr_i)}).
\label{refgr}
\ee

The last step is to evaluate the unrefined limit $z\to 0$. Unfortunately, we cannot accomplish it analytically in full generality
because this would require rewriting $\Sym\bigl\{\vthls{\bfr}_{\bbmu,\Asf}(\tau, z) \bigr\}$ in a form which makes manifest 
the existence of zero of order $n-1$ at small $z$ for all contributions except the zero modes of maximal order
and, in particular, would automatically provide a proof of Conjecture \ref{conj-zm}.
However, since evaluating a limit of a function should certainly be simpler than solving non-trivial anomaly equations,
we see this problem as just a technical obstacle and hope to return to it elsewhere.

\section{Conclusions}
\label{sec-concl}

In this paper we solved the modular anomaly equation for the generating functions $h_{r,\mu}(\tau)$ of D4-D2-D0 BPS indices,
the same as rank 0 DT invariants of Calabi-Yau threefolds, restricting to the case of threefolds with one K\"ahler modulus. 
Since for a fixed D4-brane charge $r$, 
the anomaly equation fixes the generating function only up to a modular form $\hh_{r,\mu}$ 
and involves all generating functions $h_{r_i,\mu_i}$ with $r_i<r$, 
the solution is expressed as a polynomial in $\hh_{r_i,\mu_i}$, which at this stage remain unknown functions.
The coefficients of this polynomial, called anomalous coefficients $\gi{\bfr}_{\mu,\bfmu}(\tau)$,
satisfy their own anomaly equation (Theorem \ref{thm-ancoef}), and it is solving this equation that 
takes the main part of our work. 

In two particular cases (arbitrary two charges $\bfr=(r_1,r_2)$ 
or arbitrary number of charges but all, together with the intersection number $\kappa$, equal to 1),
it is immediate to write a solution for $\gi{\bfr}_{\mu,\bfmu}(\tau)$ given by 
the mock modular forms of optimal growth from \cite{Dabholkar:2012nd} and the normalized generating functions of 
$SU(n)$ VW invariants on $\IP^2$, respectively.
The generic case is treated using indefinite theta series.
In fact, this requires introducing auxiliary functions, $\girf{\bfr}_{\mu,\bfmu}(\tau,z)$ and 
$\chgirf{\bfr}_{\mu,\bfmu}(\tau,z,\bfz)$, depending on additional variables and satisfying 
a proper set of anomaly equations. The former is a refined version of the anomalous coefficients 
and the latter is obtained from the former by a lattice extension.
As a result, we have solved the anomaly equation for $\chgirf{\bfr}_{\mu,\bfmu}$ 
(Theorems \ref{thm-gensol} and \ref{thm-ambig}),
reduced it $\girf{\bfr}_{\mu,\bfmu}$ (see \eqref{refgr}), and evaluated its unrefined limit $z\to 0$
producing $\gi{\bfr}_{\mu,\bfmu}$ in the cases of two and three charges.
Unfortunately, the last step turns out to be too complicated to be done analytically in generic case.
The realized strategy is schematically presented in Fig. \ref{fig-strategy}.

\lfig{Construction of the refined anomalous coefficients through the refinement and lattice extension
and their relation to the generating functions of BPS indices.}
{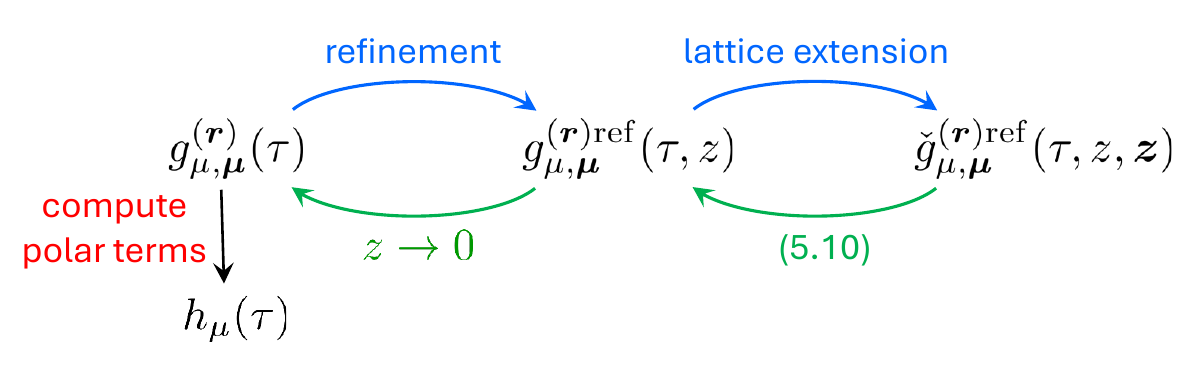}{12cm}{fig-strategy}{-0.5cm}

In fact, the existence of the unrefined limit of our solution for generic charges remains 
conjectural since it relies on Conjecture \ref{conj-zm}.
Although we gave a strong argument why we expect it to be true, it leaves a gap in our construction. 
It would certainly be desirable to fill this gap and to find a way to rigorously prove the conjecture.
This might also suggest a way to explicitly evaluate the unrefined limit in generic case,
the other missing step for having a complete result.

For a few cases of small charges and small intersection number, 
we presented first terms in $\q$-series of the anomalous coefficients in appendix \ref{ap-expand}.
Besides, in appendix \ref{ap-consist}, we have shown that our solution in terms of indefinite theta series 
is consistent with the other solutions mentioned above, provided by mock modular forms of optimal growth 
and generating functions of VW invariants.
To this end, we constructed a combination of the solutions that is required to be a Jacobi form
and then explicitly evaluated it confirming that the requirement is indeed satisfied.

This calculation opens an interesting possibility.
The point is that the mock modular forms of optimal growth introduced in \cite{Dabholkar:2012nd}
are constructed in terms of certain seed functions $\cGi{d}$ (see \eqref{DMZ}), but only
for the first two of them, $\cGi{1}$ and $\cGi{6}$, one knows analytic expressions.
Our results can be used to obtain such analytic expressions for other seed functions. 

It is useful to note that our construction resembles a lot the solution of a similar modular anomaly equation 
for the generating functions of refined VW invariants in \cite{Alexandrov:2020bwg}.
But there is an important difference that here we could fix the holomorphic modular ambiguity for
the anomalous coefficients in terms of a Jacobi-like form, whereas in the VW case
a similar function had to be a Jacobi form. It is this possibility to use Jacobi-like forms that is responsible 
for a very simple form of the solution \eqref{solfullphi} for the ambiguity.

Our solution for the anomalous coefficients reduces the problem of finding the generating functions $h_{r,\mu}$
to the problem of finding just a finite number of Fourier coefficients.
For example, it would be sufficient to compute their polar terms,
which allows us to fix the modular ambiguities $\hh_{r,\mu}$ and thereby the whole generating functions.
Moreover, for $r>1$, typically, the polar terms must satisfy non-trivial constraints to produce a mock modular form
\cite{Manschot:2007ha,Manschot:2008zb}. Therefore, such computation would provide an extremely 
strong test of mock modularity, which at the physical level appears as a consequence of S-duality,
whereas at the mathematical level remains completely mysterious.

The problem however is that the existing techniques are likely insufficient to do this.
For example, the naive extension of the approach based on wall-crossing,
as outlined in \cite[\S3.3]{Alexandrov:2023ltz}, would require knowledge of GV invariants up to genus 232
to compute all polar terms of $h_{3,\mu}$ for CY $X_{10}$, which is given by 
the degree 10 hypersurface in weighted projective space $\IP^{5,2,1,1,1}$ (one of the two CYs 
for which $h_{2,\mu}$ has been explicitly found). 
This seems unrealistic in the current state of affairs.
Therefore, one needs to look for alternative methods to compute polar terms
which would not be limited to extremely small charges.

Finally, although we restricted in this paper to the one-modulus case, 
our construction seems to be easily generalizable to CYs with multiple moduli.
In this case CYs can also have additional structures, such as elliptic or K3 fibrations,
and it would be interesting to study an interplay between these structures and the mock modularity 
of rank 0 DT invariants exploited here.

\section*{Acknowledgements}

The authors are grateful to Abhiram Kidambi for valuable discussions.
SA thanks the Galileo Galilei Institute for Theoretical Physics for the hospitality 
and the INFN for partial support during the completion of this work.
Besides, SA would like to thank the Isaac Newton Institute for Mathematical Sciences, Cambridge, 
for support and hospitality during the programme ``Twistor theory",
supported by EPSRC grant no EP/R014604/1, where work on this paper was undertaken.

\appendix

\section{Jacobi and Jacobi-like forms}
\label{ap-Jacobi}

\subsection{Jacobi forms}

Jacobi forms have been introduced and studied in detail by Eichler and Zagier in \cite{MR781735}.
Here we provide a definition which generalizes the original one in several aspects:
it allows the function to have multiple elliptic arguments, to be non-holomorphic, 
to be vector valued and to possess a non-trivial multiplier system. All these generalizations
are well-known and play an important role in numerous physical problems.

Let $\vph_\mu(\tau, \bfz)$ be a finite set of (in general, non-holomorphic) functions, 
labelled by $\mu$, on $\IH\times\IC^n$, and 
$\bfx\cdot \bfy=\sum_{i,j=1}^n Q_{ij} x_i y_j$ denotes a bilinear form on $\IC^n$.
Then $\vph_\mu(\tau, \bfz)$ is a vector valued (multi-variable) Jacobi form of weight $(w,\bw)$ and 
(matrix valued) index $mQ_{ij}$
if it satisfies the following transformation properties 
\begin{subequations}
\bea
	\vph_{\mu}(\tau,\bfz+\bfa\tau+\bfb)&=&
	e^{-2\pi\I m \( \bfa^2\tau + 2 \bfa \cdot \bfz\)} \,\vph_{\mu}(\tau,\bfz),
	\qquad \bfa,\bfb\in \IZ^n,
	\label{Jacobi-ell}
	\\
	\vph_{\mu}\(\frac{a\tau+b}{c\tau+d}, \frac{\bfz}{c\tau+d}\)
	&=& (c\tau+d)^w(c\btau+d)^{\bw} \, e^{\frac{2\pi\I m c \bfz^2}{c\tau+d}}
	\sum_\nu M_{\mu\nu}(\rho)\,\vph_{\nu}(\tau,\bfz),
\label{Jacobi-mod}
\eea
\label{Jacobi}
\end{subequations}
where $\rho=\scriptsize{\(\begin{array}{cc}
		a & b \\ c & d
	\end{array}\)}\in SL(2,\IZ)$ and $M_{\mu\nu}(\rho)$ is a multiplier system. 

Setting the elliptic variables $\bfz=0$, \eqref{Jacobi}
reduces to the definition of a vector valued modular form.
Note that since $M_{\mu\nu}(\rho)$ must furnish a representation of the group $SL(2,\IZ)$ 
generated by two transformations,
$T=\scriptsize{\(\begin{array}{cc}
		1 & 1 \\ 0 & 1
	\end{array}\)}$
and
$S=\scriptsize{\(\begin{array}{cc}
		0 & -1 \\ 1 & 0
	\end{array}\)}$,
to define the multiplier system, it is enough to specify it for $\rho=T$ and $S$.
Thus, to characterize the modular behavior of a Jacobi form, it is sufficient to provide
its modular weight $(w,\bw)$, index $m$ and two matrices $M_{\mu\nu}(T)$ and $M_{\mu\nu}(S)$.

\subsection{Jacobi-like forms}

Jacobi-like forms have been first introduced in \cite{Zagier:1994,Cohen1997}
and further studied in the mathematical literature (see, e.g., \cite{Lee2001,Lee2006,2010arXiv1007.4823C}).
In the simplest case they are defined as formal power series in $X$, with coefficients in functions on 
$\IH$, satisfying the following transformation property under $SL(2,\IZ)$
\be  
\Phi\(\frac{a\tau+b}{c\tau+d}, \frac{X}{(c\tau+d)^2}\)=e^{\frac{cX}{c\tau+d}}\,\Phi(\tau,X).
\label{JLtrans}
\ee 

Given a Jacobi-like form, let us consider the function $X^{-w}\Phi(\tau,X)$.
Then it is easy to see that under the identification $X=2\pi\I m z^2$, one gets a function
$\vph(\tau,z)$ which satisfies the modular transformation property \eqref{Jacobi-mod} of a Jacobi form of weight $w$ and index $m$.
However, the elliptic property \eqref{Jacobi-ell} is in general missing. This justifies the name ``Jacobi-like".

In the applications relevant to our work, it is more convenient to work directly with functions of $z$
that have a definite weight and index and to allow for all the generalizations (multiple variables, non-holomorphicity, 
vector valuedness) that we allowed for Jacobi forms. Therefore, for the purposes of this paper, 
we will call Jacobi-like form any function $\vph_\mu(\tau, \bfz)$ that satisfies \eqref{Jacobi-mod}.\footnote{The original
	definition of Jacobi-like forms implies that they have an expansion in {\it even} powers of $z$.
	However, once one allows for a non-trivial multiplier system, there is no much sense keeping this condition.
	In practice, the functions appearing in the main text are functions of $\tau$, $z$ and $\bfz=(z_1,\dots,z_n)$
	which behave as Jacobi-like forms with respect to $z$ with an expansion in even powers, up to an overall shift in the power, 
	and as usual Jacobi forms with respect to $\bfz$.}

Next, we are interested in modular properties of the expansion coefficients of 
a Jacobi-like form around a point where one of the (would be) elliptic variables, say $z_1$, vanishes.
For simplicity, we restrict ourselves to the case $n=1$ of only one elliptic variable
and set $Q_{11}=1$, but the propositions below are trivially generalized to $n>1$
provided the quadratic form is factorized, i.e. $Q_{1i}=0$ for $i>1$.

In fact, it is well-known that the coefficients of the expansion in $X\sim z^2$ of a Jacobi-like form 
are in one-to-one correspondence with modular forms which can be constructed as linear combinations of the $\tau$-derivatives 
of the coefficients \cite{Zagier:1994,Cohen1997}.
In particular, the coefficient of the leading term, say, $z^{n_0}$ in the expansion  
is a modular form of weight $w+n_0$.

It is also known that Jacobi-like forms are closely related to quasimodular forms \cite{2010arXiv1007.4823C},
the simplest example of which is the Eisenstein series
$
E_2(\tau) = 1- 24 \sum_{n=1}^{\infty} \sigma_1(n) e^{2\pi \I n \tau}
$
satisfying the transformation property
\be
E_2\(\frac{a\tau+b}{c \tau +d}\)  = \(c\tau+d\)^2 \(E_2(\tau) + \frac{6}{\pi\I}\, \frac{c}{c\tau+d}\).
\label{modtr-E2}
\ee
In particular, one can note that the anomalous term in this transformation has the same form as 
the logarithm of the exponential factor in the transformation of a Jacobi-like form.
This immediately implies the following

\begin{proposition}\label{prop-JacobiE2}
Let $\vph_\mu(\tau,z)$ be a Jacobi-like form of modular weight $w$ and index $m$. Then
\be
\tilde \vph_\mu(\tau,z) = e^{\frac{m}{3}\pi^2 E_2(\tau) z^2} \vph_\mu(\tau,z)
\label{def-tphi}
\ee	
is a Jacobi-like form of the same weight and vanishing index, and the coefficients of its Laurent expansion
$\tilde \vph_\mu(\tau,z)=\sum_{n=n_0}^\infty h_n(\tau) z^n$
are modular forms of weight $w+n$.
\end{proposition}

This simple observation can be used to prove

\begin{proposition}
	\label{prop-Jacobi-n}
	Let $\vph_\mu(\tau,z)$ be a Jacobi-like form of modular weight $w$ and index $m$, and having a smooth limit at $z\to 0$.
	We define the following differential operator
	\be 
	\cD_m^{(n)}=\sum_{k=0}^{\lfloor n/2\rfloor}c_{n,k}E_2^k(\tau) \,\p_z^{n-2k},
	\label{defcDmn}
    \qquad
	c_{n,k}=\frac{n!\(\frac{2m}{3}\pi^2\)^k }{(2k)!!(n-2k)!}\, .
	\ee
	Then 
	\be
	\phi^{(n)}_\mu(\tau)\equiv \cD_m^{(n)}\vph_\mu(\tau,z)|_{z=0},	
	\label{coeffJac}
	\ee
	is a vector valued modular form of weight $w+n$.
\end{proposition}
\begin{proof}
	If $\phi_\mu$ is smooth at small $z$, the same is true for the function $\tilde \vph_\mu$ \eqref{def-tphi}
	and hence its expansion coefficients are given by the derivatives with respect to $z$ evaluated at $z=0$.
	According to Proposition \ref{prop-JacobiE2}, such derivatives $\partial_z^n\tilde \vph_\mu(\tau,0)$
	transform as modular forms of weight $w+n$.
	On the other hand, we have
	\be
	\partial_z^n\tilde \vph_\mu(\tau,0) = \sum_{k=0}^{\lfloor n/2 \rfloor} \frac{n!}{(2k)!(n-2k)!} \, 
	\left.\(e^{-\frac{x^2}{2}}\frac{\de^{2k}}{\de^{2k} x}\, e^{\frac{x^2}{2}}\)\right|_{x=0}\, 
	\(\frac{2m}{3} \pi^2 E_2(\tau)\)^k \partial_z^{n-2k} \vph_\mu(\tau,0).
	\ee
	Taking into account that $\left.\(e^{-\frac{x^2}{2}}\frac{\de^{2k}}{\de^{2k} x}
	\, e^{\frac{x^2}{2}}\)\right|_{x=0}=(2k-1)!!$, 
	we conclude that
	\be
	\partial_z^n\tilde \vph_\mu(\tau,0) =\cD_m^{(n)}\vph_\mu(\tau,z)|_{z=0},
	\ee
	which proves the statement of the proposition.
\end{proof}

\section{Summary of modular properties}
\label{modular-prop}

In this appendix we collect the modular weights, indices and multiplier systems of the main
modular functions appearing in the paper.

\paragraph{Generating functions of BPS indices:} 
\be
\begin{split}
	w(h_{\Nr}) &=-3/2,
	\\
	\Mi{h_\Nr}_{\mu \nu}(T) &= 
	e^{\frac{\pi \I}{\kappa \Nr}\(\mu + \hf \kappa \Nr^2\)^2+\frac{\pi \I}{12}\,c_2 \Nr}\,
	\delta_{\mu\nu},
	\\
	\Mi{h_\Nr}_{\mu \nu}(S) &= \frac{(-1)^{\chi_\Nr}}{\sqrt{\kappa \Nr }} \,
	e^{-\frac{\pi\I}{4}-2\pi \I \,\frac{\mu \nu}{\kappa \Nr}}, 
\end{split}
\label{mult-hr}
\ee
where $\chi_r$ is defined in \eqref{defL0}.

\paragraph{Redefined generating functions:} 
\be
\begin{split}
	w(\tlh_{\Nr}) &=-3/2,
	\\
	\Mi{\tlh_\Nr}_{\mu \nu}(T) &= 
	e^{\frac{\pi\I}{\kappa\Nr}(1-\kappa\Nr)\mu^2 + \frac{\pi \I}{4}\(\kappa+\frac{c_2}{3}\)\Nr} \,
	\delta_{\mu\nu},
	\\
	\Mi{\tlh_\Nr}_{\mu \nu}(S) &= \frac{e^{\frac{\pi\I}{4}((2\kappa+c_2)\Nr-1)}}{\sqrt{\kappa\Nr}} 
	\,e^{-2\pi \I \frac{\mu\nu}{\kappa\Nr}}.
\end{split}
\label{mult-thr}
\ee

\paragraph{Anomalous coefficients:} 
\be
\begin{split}
	w\(\gi{\bfr}\) &=3(n-1)/2, 
	\\
	\Mi{\gi{\bfr}}_{\mu, \bfmu, \nu, \bfnu}(T) &= 
	e^{\pi \I \( \mu - \sum_i \mu_i\)+ \pi \I \(\frac{\mu^2}{\kappa \Nr}-\sum_i \frac{\mu_i^2}{\kappa \Nr_i}\) }\,
    \delta_{\mu \nu} \delta_{\bfmu \bfnu},
	\\
	\Mi{\gi{\bfr}}_{\mu, \bfmu, \nu, \bfnu}(S) &= 
    \frac{e^{\frac{\pi\I}{4}(n-1)}}{\sqrt{\kappa^{n+1} \Nr\prod_i \Nr_i }} \,
    e^{-2\pi \I \(\frac{\mu \nu}{\kappa \Nr} - \sum_i \frac{\mu_i\nu_i}{\kappa \Nr_i}\)}.
\end{split}
\label{mult-gr}
\ee

\paragraph{Refined anomalous coefficients:}
\be
\begin{split}
	w\(\girf{\bfr}\) &= (n-1)/2,
	\\
	m\(\girf{\bfr}\) &= -\frac{\kappa}{6}\biggl(\Nr^3 - \sum_{i=1}^n \Nr_i^3\biggr),
	\\
	\Mi{\girf{\bfr}}_{\mu, \bfmu, \nu, \bfnu}(T) &= 
	e^{\pi \I \( \mu - \sum_i \mu_i\)+ \pi \I \(\frac{\mu^2}{\kappa \Nr}-\sum_i \frac{\mu_i^2}{\kappa \Nr_i}\)}\,
	\delta_{\mu\nu} \delta_{\bfmu\bfnu},
	\\
	\Mi{\girf{\bfr}}_{\mu, \bfmu, \nu, \bfnu}(S) &= 
	\frac{e^{\frac{\pi\I}{4}(n-1)}}{\sqrt{\kappa^{n+1} \Nr\prod_i \Nr_i }} \,
	e^{-2\pi \I\(\frac{\mu \nu}{\kappa \Nr}-\sum_i\frac{\mu_i \nu_i}{\kappa \Nr_i} \) }.
\end{split}
\label{multrefg}
\ee

\paragraph{Extended refined anomalous coefficients:}
\be
\begin{split}
	w\(\chgirf{\bfr}\)=&\, (n-1+d_\bfr)/2,
    \\
    m\(\chgirf{\bfr}\) =&\, 
   \hf\, \diag\biggl(-\frac{\kappa}{3}\biggl(\Nr^3 - \sum_{i=1}^n \Nr_i^3\biggr),(\frt^{(\Nr_1)})^2,\dots,(\frt^{(\Nr_n)})^2\biggr),
    \\
	\Mi{\chgirf{\bfr}}_{\mu, \bfmu, \nu, \bfnu}(T) =&\, 
    e^{\pi \I \( \mu - \sum_i \mu_i\)+ \pi \I \(\frac{\mu^2}{\kappa \Nr}-\sum_i \frac{\mu_i^2}{\kappa \Nr_i}\)
	+\frac{\pi\I}{4}\, d_\bfr}\,
    \delta_{\mu\nu} \delta_{\bfmu\bfnu},
    \\
    \Mi{\chgirf{\bfr}}_{\mu, \bfmu, \nu, \bfnu}(S) =&\, 
    \frac{e^{\frac{\pi\I}{4}(n-1-3d_\bfr)}}{\sqrt{\kappa^{n+1} \Nr\prod_i \Nr_i }} \,
    e^{-2\pi \I\(\frac{\mu \nu}{\kappa \Nr}-\sum_i\frac{\mu_i \nu_i}{\kappa \Nr_i} \) },
\end{split}
\label{multsys-phi}
\ee
where $(\frt^{(\Nr)})^2=\sum_{\alpha=1}^{d_\Nr}(\frt^{(\Nr)}_\alpha)^2$ and $d_\bfr=\sum_{i=1}^n d_{\Nr_i}$.

\section{Theta series}
\label{modular-theta}

In this appendix we define some useful theta series and describe their modular properties. 

\subsection{Generalized theta series and modularity condition}

Let us define
\be
\vth_{\bbmu}(\tau, \zbbm;\bbLambda, \Phi, \pbbm) = \sum_{\kbbm\in \bbLambda + \bbmu + \hf \pbbm} 
(-1)^{\pbbm\ast \kbbm} \Phi\(\sqrt{2\tau_2}\(\kbbm+\bbbeta \) \) \q^{-\hf \kbbm^2} 
e^{2\pi\I\zbbm\ast \kbbm},
\label{gentheta}
\ee
where $\q=e^{2\pi\I\tau}$, 
$\bbLambda$ is a $d$-dimensional lattice equipped with a bilinear form $\xbbm\ast \ybbm$ 
such that the associated quadratic form has signature $(n,d-n)$ and is integer valued, 
$\pbbm$ is a characteristic vector satisfying 
$\kbbm\ast(\kbbm+\pbbm) =0 \mod 2$ for $\forall \kbbm\in \bbLambda$, 
$\bbmu \in \bbLambda^{\star}/\bbLambda$, and $\zbbm=\bbalpha-\tau\bbbeta\in\IC^d$
with $\bbalpha,\bbbeta\in\IR^d$.
(We follow the convention to denote $d$-dimensional quantities by blackboard letters.)
The Vign\'eras theorem \cite{Vigneras:1977} asserts that if the kernel $\Phi(\xbbm)$ satisfies 
suitable decay properties as well as
the following differential equation 
\be
\[\p_\xbbm^2+2\pi (\xbbm\ast\p_\xbbm-\lambda)\]\Phi(\xbbm)=0,
\label{Vigdif}
\ee
where $\lambda$ is an integer parameter,
then the theta series is a vector valued (multi-variable) Jacobi form\footnote{More precisely, 
the elliptic transformation \eqref{Jacobi-ell} can generate an additional sign factor 
$(-1)^{\pbbm\ast(\abbm+\bbbm)}$.} 
with the following modular properties: 
\be
\begin{split}
	w(\vth) &= \(\hf \,(d+\lambda), -\hf\, \lambda\),
	\qquad\qquad
	m(\vth) = - \hf\, \ast,
	\\
	\Mi{\vth}_{\bbmu \bbnu}(T) &= e^{-\pi \I\(\bbmu + \hf \pbbm\)^2 } \delta_{\bbmu\bbnu},
	\qquad
	\Mi{\vth}_{\bbmu \bbnu}(S) = \frac{e^{(2n-d)\frac{\pi\I}{4}}}{\sqrt{|\bbLambda^\ast/\bbLambda|}} \,
	e^{\frac{\pi \I}{2}\pbbm^2} e^{2\pi \I \bbmu \ast \bbnu},
\end{split}
\label{mult-genth}
\ee
where by $\ast$ in the formula for the index we mean the matrix representing the bilinear form.
The multiplier system here forms the Weil representation of the modular group defined by the lattice $\bbLambda$.
A particularly interesting case is when the multi-variable Jacobi form is reduced to the usual
Jacobi form by choosing $\zbbm=\bbtheta z$ where $\bbtheta\in \bbLambda$.
Then the index is a scalar and is given by
\be 
m(\vth)=-\hf\, \bbtheta^2.
\ee

\subsection{Unary theta series}

Let us specialize \eqref{gentheta} to the case where $d=1$, $n=0$ and $\bbLambda= m\IZ$ so that 
the bilinear form is $\xbbm\ast\ybbm=-\xbbm\ybbm/m$.
We also take $\pbbm=-mp$ where $p$ is odd for odd $m$ and arbitrary integer otherwise, 
$\zbbm=-m z$ and $\Phi=1$ (hence $\lambda=0$).
Then the theta series reduces to 
\be
\label{Vignerasth}
\vths{m,p}_{\mu}(\tau,z)=
\!\!\!\!
\sum_{{k}\in \IZ+\frac{\mu}{m}+\frac{p}{2}}\!\!
(-1)^{mpk}\, \q^{m k^2/2}\, y^{mk},
\ee
where we introduced $y=e^{2\pi\I z}$.
Its modular properties follow from \eqref{mult-genth} and are given by
\be
\begin{split}
	&w(\vths{m,p}) = 1/2,
	\qquad\qquad
	m(\vths{m,p}) = m/2,
	\\
	\Mi{m,p}_{\mu\nu}(T)&= e^{\frac{\pi\I}{m} \(\mu+\tfrac{mp}{2} \)^2}\,\delta_{\mu\nu}\, ,
	\qquad
	\Mi{m,p}_{\mu\nu}(S)=
	\frac{e^{-\frac{\pi\I}{2}\, m p^2}}{\sqrt{\I m}}\,
	e^{-2\pi\I\,\frac{\mu\nu}{m}}.
\end{split}
\label{eq:thetatransforms}
\ee

For even $m=2\kappa$, we can choose $p=0$. Then \eqref{Vignerasth} gives
\be
\ths{\kappa}_\mu(\tau,z)\equiv \vths{2\kappa,0}_{\mu}(\tau,z)
=\sum_{k\in 2\kappa\IZ+\mu}\q^{\frac{k^2}{4\kappa} }\,y^{2\kappa k}.
\label{deftheta}
\ee
If $z=0$, we will simply drop the last argument and write $\ths{\kappa}_\mu(\tau)$.
The multiplier system \eqref{eq:thetatransforms} reduces to
\be
	\Mi{\ths{\kappa}}_{\mu \nu}(T) =
	e^{\frac{\pi \I}{2\kappa} \,\mu^2}\delta_{\mu \nu},
	\qquad
	\Mi{\ths{\kappa}}_{\mu \nu}(S)=
	\frac{1}{\sqrt{2\I\kappa}} \, e^{-\frac{\pi \I}{\kappa}\,\mu \nu} .
\label{mult-theta}
\ee

On the other hand, specifying $m=p=1$ in \eqref{Vignerasth},
we reproduce the standard Jacobi theta function 
\be
\label{free-theta-1}
\theta_1(\tau,z) =\vths{1,1}_{0}(\tau,z)= \sum_{k\in Z+\hf} \q^{k^2/2} (-y)^k
\ee
whose modular properties are
\be
\begin{split}
	w(\theta_1) &= 1/2,
	\qquad\quad\ \
	m(\theta_1) = 1/2,
	\\
	\Mi{\theta_1}(T) &= e^{\pi \I/4} ,
	\qquad
	\Mi{\theta_1}(S) = e^{-3\pi \I/4}.
\end{split}
\label{multi-theta-N}
\ee
It has the following expansion around $z=0$
\be
\theta_1(\tau,z) = -2\pi \eta(\tau)^3 z -4\pi^2 \I \eta'(\tau) \eta(\tau)^2 z^3 + O(z^5) . 
\label{theta1-z3}
\ee

\subsection{Convergence of indefinite theta series}
\label{ap-converge}

Let us now consider theta series with a quadratic form of indefinite signature.
In this case the kernel $\Phi(\xbbm)$ cannot be trivial anymore since otherwise the theta series would be divergent.
On the other hand, a non-trivial kernel would spoil holomorphicity in $\tau$ unless $\Phi(\xbbm)$ is a piece-wise constant 
function.\footnote{It is possible also to multiply it by a homogeneous polynomial in $\xbbm$ since 
the non-holomorphic dependence can then be canceled by multiplying by a power of $\tau_2$.}
Thus, the only way to get a convergent and holomorphic theta series is to take $\Phi(\xbbm)$ 
to be a combination of sign functions.
The following theorem from \cite{Alexandrov:2020bwg} (generalizing results of 
\cite{Nazaroglu:2016lmr,Alexandrov:2017qhn,funke2017theta}) 
provides the simplest choice of such kernel

\begin{theorem}
	\label{th-conv}
	Let the signature of the quadratic form be $(n,d-n)$ and
	\be
	\Phi(\xbbm)=\prod_{i=1}^n\Bigl(\sgn(\vbbm_{1,i}\ast \xbbm)-\sgn(\vbbm_{2,i}\ast\xbbm)\Bigr).
	\label{kerconverge}
	\ee
	Then the theta series \eqref{gentheta} is convergent provided:
	\begin{enumerate}
		\item
		for all $i\in \Zv_{n}=\{1,\dots,n\}$,
		$\vbbm_{1,i}^2,\vbbm_{2,i}^2\geq 0$;
		\item
		for any subset $\cI\subseteq \Zv_{n}$ and any set of $s_i\in \{1,2\}$, $i\in\cI$,
		\be
		\Delta_{\cI}(\{s_i\})\equiv \mathop{\det}\limits_{i,j\in \cI}(\vbbm_{s_i,i}\ast \vbbm_{s_j,j})\geq 0;
		\label{condDel}
		\ee
		\item
		for all $\ell\in\Zv_n$ and any set of $s_i\in \{1,2\}$, $i\in\Zv_n\setminus\{\ell\}$,
		\be
		\vbbm_{1,\ell\perp\{s_i\}}\ast \vbbm_{2,\ell\perp\{s_i\}}>0,
		\label{condscpr}
		\ee
		where $_{\perp\{s_i\}}$ denotes the projection on the subspace orthogonal to the span of
		$\{\vbbm_{s_i,i}\}_{i\in \Zv_n\setminus\{\ell\}}$;
		\item
		if $\vbbm_{s,i}^2=0$, then $\exists \alpha_{s,i}\in \IR$ such that $\alpha_{s,i}\vbbm_{s,i}\in \bbLambda$.
	\end{enumerate}
	\label{th-converg}
\end{theorem}

Note that the last condition requiring that the (rescaled) null vectors, i.e. satisfying $\vbbm_{s,i}^2=0$, 
that appear in the definition of the kernel belong to the lattice
is important. If such a null vector is present, it is also important to keep the elliptic variable $\zbbm$ generic
because the theta series has poles at the points where $\exists {\kbbm}\in \bbLambda+\bbmu+\hf\,\pbbm$ such that
$\vbbm_{s,i}\ast (\kbbm+\bbbeta)=0$. In particular, theta series involving null vectors are typically 
divergent in the limit $\zbbm\to 0$.

\section{Functions determining completions}
\label{ap-E}

\subsection{Generalized error functions}
\label{ap-generr}

In section \ref{ap-converge}, we provided a class of functions $\Phi(\xbbm)$
that define convergent and holomorphic indefinite theta series.
However, in contrast to the usual theta series with negative\footnote{Usually, the convergent case corresponds 
to positive definite quadratic forms. In our conventions it is negative due to the minus sign in the power of $\q$
in \eqref{gentheta}. See also footnote \ref{foot-sign}.} 
definite quadratic form, 
they are not modular. This can be seen, for example, from the fact that 
the discontinuities of the signs spoil the Vign\'eras equation \eqref{ap-converge}.
Nevertheless, there is a simple recipe to construct their modular completions 
\cite{Zwegers-thesis,Alexandrov:2016enp,Nazaroglu:2016lmr}.

This is achieved with help of the generalized error functions
introduced in \cite{Alexandrov:2016enp,Nazaroglu:2016lmr} (see also \cite{kudla2016theta}).
They are defined by
\bea
E_n(\cM;\vu)&=& \int_{\IR^n} \de \vu' \, e^{-\pi\sum_{i=1}^n(u_i-u'_i)^2} \prod_{i=1}^n \sign(\cM^{\rm tr} \vu')_i\, ,
\label{generr-E}
\eea
where $\vu=(u_1,\dots,u_n)$ is $n$-dimensional vector and $\cM$ is $n\times n$ matrix of parameters.
To get kernels of indefinite theta series, we need however functions depending on a $d$-dimensional vector
rather than $n$-dimensional one. To define such functions, let $\cV$ be $d\times n$ matrix 
which can be viewed as a collection of $n$ vectors, $\cV=(\vbbm_1,\dots,\vbbm_n)$, 
and it is assumed that these vectors span a positive definite subspace in $\IR^d$ endowed with the quadratic form $\ast$,
i.e. $\cV^{\rm tr}\ast\cV$ is positive definite. We also introduce a $n\times d$ matrix
$\cB$ whose rows define an orthonormal basis for this subspace.
Then we set
\be
\Phi_n^E(\cV;\xbbm)=E_n(\cB\ast \cV;\cB\ast\, \xbbm).
\label{generrPhiME}
\ee
The detailed properties of these functions can be found in \cite{Nazaroglu:2016lmr}.
Most importantly, they do not depend on $\cB$,
solve the Vign\'eras equation \eqref{Vigdif} with $\lambda=0$ and at large $\xbbm$ reduce to
$\prod_{i=1}^n \sgn (\vbbm_i\ast\,\xbbm)$.
Thus, to construct a completion of the theta series whose kernel is a combination of sign functions,
it is sufficient to replace each product of $n$ sign functions by $\Phi_n^E$ with matrix of parameters $\cV$
given by the corresponding vectors $\vbbm_i$.

Finally, if one of the vectors is null, it reduces the rank of the generalized error function.
Namely, for $\vbbm_\ell^2=0$, one has
\be
\Phi_n^E(\{\vbbm_i\};\xbbm)=\sgn (\vbbm_\ell\,\ast\xbbm)\,\Phi_{n-1}^E(\{\vbbm_i\}_{i\in \Zv_{n}\setminus\{\ell\}};\xbbm).
\label{Phinull}
\ee
In other words, for such vectors the completion is not required.

\subsection{Coefficients $\scR_n$}
\label{subsec-R}

The generalized error functions defined in the previous subsection play the role of building blocks
in the definition of the coefficients $\scR_n(\bfhgam;\tau_2)$ appearing in \eqref{redefRn}.
Here we provide their expression found in \cite{Alexandrov:2024jnu}.

The construction proceeds in two steps.
At the first step, we introduce functions of $\tau_2$ parametrized by $n$ reduced charges $\hgam_i=(r_i,q_{i})$.
To this end, let $\IT_n^\ell$ be the set of unrooted labelled trees with $n$ vertices
decorated by charges from the set $\bfhgam=(\hgam_1,\dots,\hgam_n)$.
Given a tree $\cT\in \IT_{n}^\ell$, we denote the set of its edges by $E_{\cT}$, the set of vertices by $V_{\cT}$, 
the source and target vertex\footnote{The orientation 
	of edges on a given tree can be chosen arbitrarily, the final result does not depend on this choice.} 
of an edge $e$ by $s(e)$ and $t(e)$, respectively,
and the two disconnected trees obtained from $\cT$ by removing the edge $e$ by $\cT_e^s$ and $\cT_e^t$.
Furthermore, to each edge we assign the vector 
\be
\bfv_e=\sum_{i\in V_{\cT_e^s}}\sum_{j\in V_{\cT_e^t}}\bfv_{ij},
\label{defue}
\ee
where $\bfv_{ij}$ are $n$-dimensional vectors with the following components
\be
(\bfv_{ij})_k=\delta_{ki} \Nr_j-\delta_{kj} \Nr_i.
\label{defvij}
\ee
Using these notations, we define 
\be
\Ev_n(\bfhgam;\tau_2)=
\frac{\EPhi_n(\sqrt{2\tau_2}\, \bfq)}{(\sqrt{2\tau_2})^{n-1}}\, ,
\label{rescEnPhi}
\ee
where $\bfq=\bigl(\frac{q_1}{\kappa \Nr_1}\, ,\dots ,\, \frac{q_n}{\kappa \Nr_n}\bigr)$,
\be
\EPhi_n(\bfx)=
\frac{1}{n!}\sum_{\cT\in\, \IT_n^\ell}
\[\prod_{e\in E_\cT} \cD(\bfv_{s(e) t(e)},\bfy)\]
\Phi^E_{n-1}(\{ \bfv_e\};\bfx)\Bigr|_{\bfy=\bfx},
\label{rescEn}
\ee
and
\be
\cD(\bfv,\bfy)=\bfv\cdot\(\bfy+\frac{1}{2\pi}\,\p_\bfx\).
\label{defcDif}
\ee
The dot in \eqref{defcDif} denotes the bilinear form
\be
\bfx\cdot\bfy=\kappa \sum_{i=1}^n \Nr_ix_iy_i.
\label{biform}
\ee
In particular, this implies that $\bfv_{ij}\cdot\bfq=\gamma_{ij}$.

Importantly, each function $\Ev_n(\bfhgam;\tau_2)$ defined by \eqref{rescEnPhi} has a canonical decomposition
\be
\Ev_n(\bfhgam;\tau_2)=\Ef_n(\bfhgam)+\Ep_n(\bfhgam;\tau_2),
\label{twocEs}
\ee
where the first term $\Ef_n$ does not depend on $\tau_2$,
whereas the second term $\Ep_n$ is exponentially suppressed as $\tau_2\to\infty$ keeping
the charges $\hgam_i$ fixed. In \cite{Alexandrov:2024jnu} it was shown that 
\be 
\Ef_n(\bfhgam)
= \frac{1}{n!}
\sum_{\cT\in\, \IT_n^\ell} S_\cT(\bfhgam)\prod_{e\in E_{\cT}}\gamma_{s(e) t(e)},
\label{newexprcEf}
\ee
where 
\be
S_\cT(\bfhgam)=
\sum_{\cJ\subseteq E_\cT}e_{\cT_\cJ}
\,\prod_{e\in \cJ}\delta_{\Gamma_e}
\prod_{e\in E_\cT\setminus \cJ} \sgn (\Gamma_e),
\qquad
\Gamma_e=\sum_{i\in V_{\cT_e^s}}\sum_{j\in V_{\cT_e^t}}\gamma_{ij},
\label{defST}
\ee
Here $\cT_\cJ$ denotes the tree obtained from $\cT$ by contracting the edges $e\in E_\cT\backslash \cJ$
and $e_{\cT}$ are some rational numbers depending only on topology of $\cT$. In particular,
they vanish for trees with even number of vertices, $e_\bullet=1$, 
$e_{\bullet\!\mbox{\scriptsize --}\!\bullet\!\mbox{\scriptsize --}\!\bullet}=1/3$,
and an iterative formula for generic tree can be found in \cite{Alexandrov:2024jnu}.
If all $\Gamma_e$ are non-vanishing, $S_\cT$ reduces to the product of their signs.
This shows that the functions $\Ev_n$ have a meaning of kernels providing completions 
for holomorphic theta series constructed from signs of Dirac products of charges.

At the second step, we introduce another type of trees, the so-called Schr\"oder trees.\label{pSchr}
They are defined as rooted planar trees such that all vertices $\ver\in V_T$ 
(the set of vertices of $T$ excluding the leaves) have $k_\ver\geq 2$ children.
The set of such trees with $n$ leaves will be denoted by $\IT_n^{\rm S}$. Besides, we take
$n_T$ to be the number of elements in $V_T$ and $\ver_0$ to denote the root vertex.
The vertices of $T$ are labelled by charges so that the leaves carry charges $\hgam_i$, 
whereas the charges assigned to other vertices
are given recursively by
the sum of charges of their children, $\hgam_\ver\in\sum_{\ver'\in\Ch(\ver)}\hgam_{v'}$.
Then, given a Schr\"oder tree $T$,
we set $\Ev_{\ver}\equiv \Ev_{k_\ver}(\{\hgam_{\ver'}\})$ (and similarly for $\Ef_{\ver}, \Ep_{\ver}$)
where $\ver'\in \Ch(\ver)$ runs over the $k_\ver$ children of the vertex $\ver$. 
In terms of these notations, the coefficients $\scR_n$ 
are given by
\be
\scR_n(\bfhgam;\tau_2)= \frac{1}{2^{n-1}}\sum_{T\in\IT_n^{\rm S}}(-1)^{n_T-1} 
\Ep_{\ver_0}\prod_{\ver\in V_T\setminus{\{\ver_0\}}}\Ef_{\ver}.
\label{solRn}
\ee

\subsection{Coefficients $\scRrf_n$}
\label{subsec-refR}

The refined version of the coefficients $\scR_n$ has been introduced in \cite{Alexandrov:2019rth}.
It is given by the same sum over Schr\"oder trees as in \eqref{solRn}, 
\be
	\scRrf_n\(\bfhgam;\tau_2,\beta\) = \frac{1}{2^{n-1}}\sum_{T\in\IT_n^{\rm S}}(-1)^{n_T-1} 
	\Eprf_{v_0} \prod_{v\in V_T \backslash \{v_0\}}\Efrf_v,
	\label{refsolRn}
\ee
but now with the weights assigned to vertices determined by new functions $\Er_n(\bfhgam;\tau_2,\beta)$.
Although they depend on an additional parameter $\beta$, they are actually much simpler than their unrefined 
analogues $\Ev_n$ because in their definition there is no any sum over trees.
Namely, they are given by
\be
\Er_n(\bfhgam;\tau_2,\beta)= \Phi^E_{n-1}\(\{ \bfv_{\ell}\};\sqrt{2\tau_2}\,(\bfq+\beta\bftet )\),
\label{Erefsim}
\ee
where
\be
\bfv_\ell= \sum_{i=1}^\ell\sum_{j=\ell+1}^n\bfv_{ij},
\qquad
\bftet = \sum_{i<j} \bfv_{ij}.
\label{def-bfvk}
\ee
As in the unrefined case, $\Eprf_n=\Er_n-\Efrf_n$, while $\Efrf_n$
is the large $\tau_2$ limit of $\Er_n$. However, before taking the limit, 
one should first set $\beta=0$, i.e. 
\be 
\Efrf_n(\bfhgam)\equiv  \lim_{\tau_2\to\infty}\Er_n(\{\gama_i\};\tau_2,0)
= S_{\cT_{\rm lin}}(\bfhgam),
\label{Efref}
\ee
where $\cT_{\rm lin}=
\bullet\!\mbox{---}\!\bullet\!\mbox{--}\cdots \mbox{--}\!\bullet\!\mbox{---}\!\bullet\,$
is the simplest linear tree. The last relation follows from the observation that
the vectors $\bfv_\ell$ can be seen as the vectors \eqref{defue} assigned to edges of $\cT_{\rm lin}$.
Note that for the linear tree $e_{\cT_{\rm lin}}=\delta^{(2)}_{n-1}/n$ where $n$ is the number of vertices.

\subsection{Explicit expressions}

Let us now compute explicitly the functions $\rmRirf{\bfr}_{\mu,\bfmu}$ \eqref{Rirf-to-rmRrf} 
for $n=2$ and 3.

\subsubsection{Two charges}
\label{ap-2charges}

In this case we have 
\be
\rmRirf{\Nr_1,\Nr_2}_{\mu, \mu_1, \mu_2}(\tau, \btau,z) 
= \hf\sum_{q_1+q_2=\mu+\frac{\kappa\Nr}{2} }
\Bigl(
\scRrf_2(\gama_1,\gama_2) \, y^{\gamma_{12}}
+\scRrf_2(\gama_2,\gama_1) \, y^{-\gamma_{12}}\Bigr)\, \q^{\hf Q_2(\gama_1,\gama_2)},
\label{Rref2}
\ee
where, using the restriction on $q_1+q_2$, one easily finds from \eqref{defQlr} that
\be 
Q_2(\hgam_1,\hgam_2)=-\frac{\gamma_{12}^2}{\kappa \Nr \Nr_1 \Nr_2 }\, ,
\ee
while the function $\scRrf_2$ follows from the definitions in \S\ref{subsec-refR} to be
\be 
\scRrf_2(\gama_1,\gama_2)=\hf\, \Ep_2(\gama_1,\gama_2)=
\hf\[
E_1\(\frac{\sqrt{2\tau_2}
	\( \gamma_{12} + \kappa \Nr \Nr_1 \Nr_2\beta \)}
{\sqrt{\kappa\Nr \Nr_1 \Nr_2}}\)
-\sgn (\gamma_{12}) \] .
\label{expr-R2}
\ee
Note also that the function $E_1$ coincides with the usual error function: $E_1(u)=\mbox{Erf}(\sqrt{\pi}\, u)$.

Next, let us rewrite the sum over electric charges in \eqref{Rref2} as an unconstrained sum. 
Upon substituting the spectral flow decomposition \eqref{defmu-shift}, the condition on the sum of charges
becomes 
\be
\Nr_1\eps_1+\Nr_2\eps_2=\Delta\mu/\kappa.
\label{condcharge2}
\ee
Let us define $\Nr_0=\gcd(\Nr_1,\Nr_2)$, $\rdcr_i=\Nr_i/\Nr_0$ and $\rdcr=\Nr/\Nr_0$.
Then the condition \eqref{condcharge2} is solvable for $\eps_i\in\IZ$ only if $\Delta\mu=0 \mod \kappa \Nr_0$.
If this is the case, let $\rho_i$ be integers such that $\rdcr_1\rho_1+\rdcr_2\rho_2=1$.
Then a general solution to \eqref{condcharge2} is given by
\be
\eps_1=\frac{\Delta\mu}{\kappa\Nr_0}\,\rho_1+\rdcr_2\ell,
\qquad
\eps_2=\frac{\Delta\mu}{\kappa\Nr_0}\,\rho_2-\rdcr_1\ell,
\qquad
\ell\in \IZ.
\ee
Using this solution in the formula for the Dirac product \eqref{def-gammaij}
after plugging there the spectral flow decomposition, 
one finds that
\be  
\gamma_{12}=\Nr_0(2\kappa_{12}\ell+\mu_{12}), 
\ee
where we introduced 
\be
\begin{split}
	\kappa_{12} =&\, \hf\, \kappa  \Nr \rdcr_1 \rdcr_2,
	\\
	\mu_{12}=&\, \rdcr_2 \mu_1-\rdcr_1 \mu_2+\rdcr_1\rdcr_2(\rho_1-\rho_2)\Delta\mu
	\\
	=&\,\rdcr \mu_1-\rdcr_1 \mu + \rdcr \rdcr_1\rho_1\Delta\mu.
\end{split}
\label{defmu0}
\ee
Thus, we arrive at the following result
\be
\rmRirf{\Nr_1,\Nr_2}_{\mu, \mu_1, \mu_2} =
\frac14\, \delta^{(\kappa \Nr_0)}_{\Delta\mu}\sum_{\sigma=\pm 1}
\sum_{k\in \IZ+\frac{\sigma\mu_{12}}{2\kappa_{12}}}
\Bigl[ E_1\bigl(2\sqrt{\kappa_{12}\tau_2}\,( k+ \Nr_0 \beta )\bigr)-\sgn(k)\Bigr]
\q^{-\kappa_{12} k^2}\, y^{2\Nr_0\kappa_{12} k}.
\label{exprR2r}
\ee
The corresponding holomorphic anomaly is found to be\footnote{To obtain the holomorphic anomaly, 
	one takes the derivative with respect to $\btau$ while keeping fixed $\tau$, $y$ and
	$\by$ (and not $\alpha$, $\beta$).}
\be
\p_{\btau}\rmRirf{\Nr_1,\Nr_2}_{\mu, \mu_1, \mu_2} =
\frac{\I}{4}\,\sqrt{\frac{\kappa_{12}}{\tau_2}}\,
\delta^{(\kappa \Nr_0)}_{\Delta\mu}
\sum_{\sigma=\pm 1}\sum_{k\in \IZ+\frac{\sigma\mu_{12}}{2\kappa_{12}}}
(k-\Nr_0\beta)\,
(\q\bqq)^{\kappa_{12}(k+\Nr_0\beta)^2}
\q^{-\kappa_{12} k^2}\, y^{2\Nr_0\kappa_{12} k}.
\label{exprR2r-hanom}
\ee

The unrefined function $\trmRi{\Nr_1,\Nr_2}_{\mu, \mu_1, \mu_2}$ can be computed either 
from \eqref{redefRn} and equations in \S\ref{subsec-R}, 
or by taking the unrefined limit of \eqref{exprR2r} and using \eqref{unreflim-Rn}.
In this paper we will not use the resulting function, which has a form similar to 
\eqref{exprR2r} but with a little bit more complicated summand. Instead, 
we give here its holomorphic anomaly, 
which turns out to be simpler than its refined counterpart \eqref{exprR2r-hanom}
and is proportional to the complex conjugate of the theta series $\theta^{(\kappa)}_\mu(\tau)$ 
introduced in \eqref{deftheta} and evaluated at $z=0$: 
\be
\p_{\btau}\trmRi{\Nr_1,\Nr_2}_{\mu, \mu_1, \mu_2}(\tau, \btau) =
\frac{\Nr_0\sqrt{\kappa_{12}}}{16\pi\I \tau_2^{3/2}}\,\delta^{(\kappa \Nr_0)}_{\Delta\mu} 
\, \overline{\theta^{(\kappa_{12})}_{\mu_{12}}(\tau)}.
\label{shadowRN1N2}
\ee

Note that $\p_{\btau}\trmRi{\Nr_1,\Nr_2}_{\mu, \mu_1, \mu_2}$ must be a modular form
of weight $(3/2,2)$ with the same multiplier system as $\gi{\Nr_1,\Nr_2}_{\mu, \mu_1, \mu_2}$
given in \eqref{mult-gr}. Therefore, the result \eqref{shadowRN1N2} immediately implies

\begin{proposition}
	\label{prop-multsys}
	If $\Gi{\kappa}_\mu$ ($\mu=0,\dots,2\kappa-1$) transforms with the multiplier system
\be
\Mi{\kappa}_{\mu \nu}(T) =
e^{-\frac{\pi \I}{2\kappa} \,\mu^2}\delta_{\mu \nu},
\qquad
\Mi{\kappa}_{\mu \nu}(S)=
\frac{e^{\frac{\pi\I}{4}}}{\sqrt{2\kappa}} \, e^{\frac{\pi \I}{\kappa}\,\mu \nu},
\label{STusual}
\ee
which is the complex conjugate of the multiplier system \eqref{mult-theta} of $\theta^{(\kappa)}_{\mu}$,
	then $\delta^{(\kappa\Nr_0)}_{\Delta\mu} \Gi{\kappa_{12}}_{\mu_{12}}$
	transforms with the multiplier system \eqref{mult-gr} specified for $n=2$.
\end{proposition}

\subsubsection{Three charges}

In this case we have 
\be
\label{rmRrf3}
\rmRirf{\bfr}_{\mu,\bfmu}(\tau,\btau,z)
=
\sum_{\sum_{i=1}^{3} q_i = \mu + \kappa \Nr/2 } 
\Sym \Bigl\{ 
\scRrf_3(\bfhgam;\tau_2,\beta) \, y^{\gamma_{1+2,3}+\gamma_{12}}
\Bigr\} \,e^{\pi\I \tau Q_3(\bfhgam)},
\ee
where the quadratic form can be written as
\be
Q_3(\bfhgam) = 
-\frac{ \Nr_1 \gamma_{23}^2+ \Nr_2 \gamma_{13}^2 + \Nr_3\gamma_{12}^2 }{\kappa \Nr\Nr_1 \Nr_2 \Nr_3}\,  ,
\label{qf-n=3}
\ee
while the coefficient $\scRrf_3$ follows from \eqref{refsolRn} and is given by the sum of three Schr\"oder trees
resulting to
\be
\begin{split}
	\scRrf_3(\bfhgam;\tau_2,\beta)=&\,\frac{1}{4}\Bigl[ \Eprf_{3}(\bfhgam)
	-\Eprf_2(\hgam_{1+2}, \hgam_3)\, \Efrf_2(\hgam_1, \hgam_2)
	\\
	&\,\qquad
	-\Eprf_2(\hgam_{1}, \hgam_{2+3} ) \,\Efrf_2(\hgam_2, \hgam_3) \Bigr].
\end{split}
\label{expr-cR3}
\ee
The function $\Eprf_2$ was computed in \eqref{expr-R2} which implies that 
$\Efrf_2(\hgam_1, \hgam_2)=\sgn(\gamma_{12})$.
Finally, $\Eprf_{3}$ is found to be
\be
	\cEprf_3(\bfhgam) = 
	\Phi^E_{2}(\bfv_{1,2+3},\bfv_{1+2,3};\bfx) 
	- \sgn(\gamma_{1,2+3})\,\sgn(\gamma_{1+2,3})	
	-\frac13\,\delta_{\gamma_{1,2+3}}\delta_{\gamma_{1+2,3}},
\ee
where $\bfx=\sqrt{2\tau_2}\,(\bfq+\beta\bftet )$ and $\Phi^E_{2}(\bfx)$ 
can be expressed through the generalized error function $E_2$.
Substituting these results into \eqref{expr-cR3}, one obtains
\bea
\scRrf_3(\bfhgam;\tau_2,\beta) &=& 
\frac{1}{4}\,
\Biggl[
\Phi^E_{2}(\bfv_{1,2+3},\bfv_{1+2,3};\bfx) 
-\sgn(\gamma_{1,2+3})\,\sgn(\gamma_{1+2,3}) -\frac13\,\delta_{\gamma_{1,2+3}}\delta_{\gamma_{1+2,3}}
\nn\\
&&
- \(E_1 \(\frac{\sqrt{2\tau_2}\(\gamma_{1+2,3}+\kappa \Nr\Nr_{1+2}\Nr_3\beta \)}{\sqrt{\kappa \Nr\Nr_{1+2}\Nr_3}} \) 
-\sgn(\gamma_{1+2,3})\)\sgn(\gamma_{12}) 
\label{expr-R3}\\
&&
- \(E_1 \(\frac{\sqrt{2\tau_2}\(\gamma_{1,2+3}+\kappa \Nr\Nr_1\Nr_{2+3}\beta \)}{\sqrt{\kappa \Nr\Nr_1\Nr_{2+3}}} \) 
-\sgn(\gamma_{1,2+3})\)\sgn(\gamma_{23})
\Biggr].
\nn
\eea 
where $\Nr_{i+j}=\Nr_i+\Nr_j$.

\section{Hecke-like operators}
\label{ap-Hecke}

In this appendix we define two operators acting on Jacobi forms 
and then derive their induced action on modular forms appearing as coefficients in the theta expansion, 
which for a (mock) Jacobi form of index $m$
reads \cite{MR781735}
\be
\vph(\tau,z)=\sum_{\mu=0}^{2m-1} \hi{m}_\mu(\tau) \ths{m}_\mu(\tau,z),
\label{thetaexp}
\ee
where $\ths{m}_\mu$ is the index $m$ theta function \eqref{deftheta}.
Although this will not be needed in this paper, we will not assume that $\vph$ is holomorphic in $\tau$,
but will omit the complex conjugate argument.

The first Hecke-like operator is defined by a simple rescaling of the elliptic argument 
\cite[Eq.(4.36)]{Dabholkar:2012nd}:
\be
(U_s[\vph])(\tau,z)=\vph(\tau, sz).
\ee
The action on the modular forms follows from the following property of the theta functions 
\be
\ths{m}_\mu\(\frac{a\tau+b}{d}\, ,az\)=\sum_{\nu=0}^{d-1} e^{\frac{\pi\I b}{2m d}\, (\mu+2m \nu)^2} \, 
\ths{ad m}_{a(\mu+2m \nu)}(\tau,z)
\ee
valid for $a,d\in\IN$ and $b\in\IZ$.
In particular, choosing $a=d=s$ and $b=0$, one obtains
\be
\ths{m}_\mu\(\tau ,s z\)=\sum_{\nu = 0}^{s-1} \ths{m s^2}_{s\(\mu+2m \nu\)}(\tau,z).
\ee
This result immediately implies that the action of $U_s$ on $\hi{m}_\mu$ is given by \eqref{defHeckeU}.

The second operator is defined by \cite[Def. 3.1]{Bouchard:2018pem}\footnote{We multiplied 
	the operator defined in \cite{Dabholkar:2012nd,Bouchard:2018pem}
	by the factor $r^{1-\hf(w+\bw)}$ in order to make it 
	commuting with the simple operations changing the weight such as $\p_\tau$
	or multiplication by $\tau_2$, but leaving it intact for $w+\bw=2$. This does not affect the property
	\label{foot-Hecke}
	$$
	\cT_r\cT_s=\cT_{rs},
	\qquad
	\mbox{for } \gcd(r,s)=1
	$$
	and little bit simplifies the property (cf. \cite[Lemma 3.7]{Bouchard:2018pem})
	$$
	\cT_{p^l} =\cT_p\cT_{p^{l-1}}-p\, U_{p}\cT_{p^{l-2}},
	\qquad
	\mbox{for } l\ge 2 \mbox{ and } p \mbox{ prime}.
	$$}
\be
(\cT_r[\vph])(\tau,z)= 
\sum_{a,d>0 \atop ad=r}\(\frac{\sqrt{r}}{d}\)^{w+\bw}\sum_{b=0}^{d-1}\vph\(\frac{a\tau+b}{d},az\),
\label{HeckeTphi}
\ee
where $(w,\bw)$ is the weight of the Jacobi form.
To compare with \cite{Dabholkar:2012nd}, let us restrict to holomorphic Jacobi forms of weight $w$
and perform a Fourier expansion of $\vph(\tau,z)$ in both variables
\be
\vph(\tau,z)=\sum_{n,m}c(n,m)\, \q^n y^m ,
\qquad
\q=e^{2\pi \I \tau},
\quad
y=e^{2\pi \I z}.
\ee
Then, using the identity
\be
d^{-1}\sum_{b=0}^{d-1} e^{2\pi\I  sn/d}=\delta^{(1)}_{n/d},
\ee
one finds
\be
\begin{split}
	(\cT_r[\vph])(\tau,z)=&\,   
	r^{1-w/2}\sum_{a|r}a^{w-1}\sum_{n,m}\delta^{(1)}_{an/r}\, c(n,m)\, \q^{a^2n/r} y^{am}
	\\
	=&\,  r^{1-w/2}\sum_{a|r}a^{w-1}\sum_{n,m} c\(\frac{rn}{a}\, ,m\)\, \q^{an} y^{am}.
\end{split}
\label{HeckeTphi2}
\ee
As a result, we arrive at the following action on the Fourier coefficients
\be
\cT_r\ :\ c(n,m)\mapsto r^{1-w/2}\sum_{d|(n,m,r)} d^{w-1}\, c\(\frac{nr}{d^2}, \frac{m}{d}\),
\ee
which coincides up to the factor $r^{1-w/2}$ (see footnote \ref{foot-Hecke}) with
\cite[Eq.(4.37)]{Dabholkar:2012nd}. 
The action of \eqref{HeckeTphi} on the VV modular forms $\hi{m}_\mu$ follows from 
Definition 3.8 and Theorem 3.9 of \cite{Bouchard:2018pem} and coincides with 
the one given in \eqref{defHecke-exp}.\footnote{Note that $w$ in \eqref{defHecke-exp} 
	is equal to $w-1/2$ in terms of the weight appearing in \eqref{HeckeTphi}. 
	Hence the additional normalization factor discussed 
	in footnote \ref{foot-Hecke} now reads $r^{\frac34-\hf(w+\bw)}$.}

\section{Lattices, glue vectors and zero modes}
\label{ap-glue}

In this appendix we present various results about our main lattice $\bfLami{\bfr}$,
about the glue vectors appearing in decompositions of $\bfLami{\bfr}$ and the extended lattice $\bbLami{\bfr}$,
and an analysis of the set of zero modes.
Throughout the appendix we will use the following convenient notations: 
$\rdcr={\Nr}/{\Nr_0}$, $\rdcr_i={\Nr_i}/{\Nr_0}$
and $\Nr_{i_1\dots i_m}=\gcd(\Nr_{i_1},\dots,\Nr_{i_m})$.
If the indices are consecutive, i.e. $i_k=i_1+k-1$ and $i_m=j$,  
we will use a shorthand notation $\sss{i}{j}$ instead of $\Nr_{i_1\dots i_m}$.
Note that $\sss{1}{n}=\Nr_0$. 

\subsection{Lattice $\bfLami{\bfr}$}
\label{ap-mainlat}

The lattice $\bfLami{\bfr}$ is defined as 
\be
\bfLami{\bfr}=\left\{\bfk\in \IZ^n \ :\ \sum_{i=1}^n \Nr_i k_i=0 \right\}
\label{def-bfLam-ap}
\ee
and carries the bilinear form
\be
\bfx\cdot\bfy=\kappa\sum_{i=1}^n \Nr_i x_i y_i.
\label{bf-r-ap}
\ee
To describe its dual lattice, let us fix a $n$-tuple $\bfrho=(\rho_1,\dots,\rho_n)$ such that
$\sum_{i=1}^n \rho_i\Nr_i=\Nr_0$ and introduce
\be
\hmu^{(\bfr)}_i(\bfmu,\lambda)=\frac{\mu_i}{\kappa\Nr_i}-\frac{\sum_{i=1}^n\mu_i}{\kappa\Nr}
+\(\rho_i-\frac{\Nr_0}{\Nr}\)\lambda.
\label{defhmu-all}
\ee 
Note that it satisfies $\sum_{i=1}^n \Nr_i\hmu^{(\bfr)}_i=0$. Then the dual lattice is the following
subset of ${\Span}(\bfLami{\bfr})$ 
\be
(\bfLami{\bfr})^\star=\left\{\bfk\in\IZ^n+\bfhmu^{(\bfr)}(\bfmu,\lambda)\ :\  
\sum_{i=1}^N \Nr_i k_i=0, \ \mu_i\in \IZ_{\kappa\Nr_i},\ \lambda \in \IZ_{\rdcr} \right\}.
\label{duallat}
\ee
Since the overall shift $\mu_i\to \mu_i+\rdcr_i$ for all $i$ leaves $\bfhmu^{(\bfr)}$ invariant, 
the discriminant group is
\be 
\label{disc-grp} 
\bfDr=\IZ_{\rdcr} \otimes \frac{\prod_{i=1}^n \IZ_{\kappa\Nr_i}}{\IZ_{\kappa\Nr_0}}\,  
\ee 
and hence its order, which coincides with the determinant of the bilinear form \eqref{bf-r-ap}, is given by
\be 
|\bfDr| =\det\bfLami{\bfr}=\frac{\kappa^{n-1}\Nr}{\Nr_0^2}\,\prod_{i=1}^{n} \Nr_i.
\label{detbfLam}
\ee 
In the special case $\kappa=\Nr_1=\cdots=\Nr_n=1$, the lattice coincides with the standard $A_{n-1}$ root lattice
whose discriminant group is $\IZ_n$.

There are two decompositions of $\bfLami{\bfr}$ which play an important role in our story.
One is the decomposition \eqref{recover-L} on the sublattice generated by the vectors 
$\bfhv_{i,i+1}=\bfv_{i,i+1}/\Nr_{i,i+1}$, $i=1,\dots, n-1$.
In this case the glue vectors are given by linear combinations of other normalized vectors $\bfhv_{ij}$ with $j-i>1$.
The number of values taken by the coefficients of $\bfhv_{ij}$ that
generate independent glue vectors is given by
\be
\qquad N_{ij}\equiv\frac{\sss{i+1}{j-1}\sss{i}{j}}{\sss{i}{j-1}\sss{i+1}{j}}\, .
\label{valNij}
\ee 
The total number of glue vectors is 
\be 
\prod_{i=1}^{n-2}\prod_{j=i+2}^n N_{ij}=
\frac{\Nr_0\prod_{j=2}^{n-1} \Nr_i }{\prod_{i=1}^{n-1}\Nr_{i,i+1}}\, ,
\label{total-valNij}
\ee 
which agrees with the formula \eqref{Ng}.

The other important decomposition corresponds to a split of the set of charges into $m$ subsets
induced by a decomposition $n=\sum_{k=1}^m n_k$.
This gives rise to the following lattice factorization
\be 
\bfLami{\bfr}=\bigcup\limits_{\Csf=0}^{N_g-1} \[\(\bfLami{\bfs}+\bfgi{0}_\Csf \)
\oplus\(\mathop{\oplus}\limits_{k=1}^m \(\bfLami{\frr_k}+\bfgi{k}_\Csf\)\)\],
\label{lat-decomp-r}
\ee 
where we used notations from \eqref{split-rs}. To write the glue vectors, we introduce $s_0=\gcd(s_1,\dots,s_m)$,
$\hs_k=s_k/\gcd(\frr_k)$ and
\be 
\label{bfrho}
\begin{split}
\bfrho^{(0)}=&\, (\rho^{(0)[n_1]}_1,\dots,\rho^{(0)[n_m]}_m),
\qquad\qquad\qquad\
\sum_{k=1}^m s_k\rho^{(0)}_k=s_0,	
\\	
\bfrho^{(k)}=&\, (0^{[j_k]},\rho^{(k)}_1,\dots,\rho^{(k)}_{n_k},0^{[n-j_{k+1}]}),
\qquad 
\sum_{i=1}^{n_k} \Nr_{j_k+i}\rho^{(k)}_i=\gcd(\frr_k).
\end{split}
\ee 
Then the glue vectors are given by
\be 
\bfg_\Csf=\sum_{k=1}^m\csf_k \bfrho^{(k)}-\frac{1}{s_0}\(\sum_{\ell=1}^m \gcd(\frr_k)\, \csf_\ell\)\bfrho^{(0)}.
\ee
and are labelled by the following set of indices
\be 
\Csf=\biggl\{(\csf_1,\dots,\csf_m)\ :\ \csf_k\in\IZ_{\hs_k},\ \sum_{k=1}^m \gcd(\frr_k)\, \csf_k\in s_0\IZ  \biggr\}. 
\label{setC}
\ee 
Thus, the indices are not free, but should satisfy a constraint.
Taking this into account, the total number of glue vectors equals
\be  
N_g=\frac{\Nr_0}{s_0}\, \prod_{k=1}^m \hs_k.
\label{number-nus}
\ee 
It is useful to note that since $\bfgi{k}_\Csf\in\bfDi{\frr_k}$ and $\bfgi{0}_\Csf\in \bfDi{\bfs}$,
they should be particular cases of the vector defined in \eqref{defhmu-all}.
And indeed, it is easy to check that
\be 
\bfgi{k}_\Csf=\bfhmu^{(\frr_k)}(0,\csf_k),
\qquad 
\bfgi{0}_\Csf=\bfhmu^{(\bfs)}\bigl(\bfnu,-\tfrac{\sum_{k=1}^m\nu_k}{\kappa s_0}\bigr),
\quad \mbox{where} \quad 
\nu_k=\kappa\gcd(\frr_k)\,\csf_k.
\ee 

Let us now establish links to the sum over charges appearing in all anomaly equations of this paper. 
First, we note that the vector $\bfhmu$ defined in \eqref{thetadata} coincides with 
$\bfhmu^{(\bfr)}\bigl(\bfmu,\frac{\Delta\mu}{\kappa\Nr_0}\bigr)$ from \eqref{defhmu-all}
where the integer valuedness of the second argument is the usual condition imposed by $\delta^{(\kappa\Nr_0)}_{\Delta\mu}$
(see, e.g., Theorem \ref{thm-gensol}) and following from the condition  
$\sum_{i=1}^{n} q_i = \mu + \kappa \Nr/2$ in the sum over D2-brane charges (see, e.g., \eqref{Rirf-to-rmRrf}).
This makes obvious the fact that $(\mu,\bfmu)$, subject to the above condition and the identification
$(\mu,\bfmu)\simeq (\mu+\rdcr,\bfmu+\hat\bfr)$, labels elements of the discriminant 
group $\bfDr$.
To see that the whole sum over D2-brane charges corresponds to the sum over $\bfLami{\bfr}$, 
it is sufficient to identify 
\be
q_i=\kappa \Nr_i\( k_i+\frac{\mu}{\kappa\Nr}+\hf\)
=\kappa \Nr_i\( \eps_i+\frac{\rho_i\Delta\mu}{\kappa\Nr_0}\)+\mu_i+\hf\, \kappa \Nr_i,
\label{rel-charges}
\ee 
where $k_i=\eps_i+\hmu^{(\bfr)}_i\bigl(\bfmu,\frac{\Delta\mu}{\kappa\Nr_0}\bigr)$, with $\eps_i\in \IZ$,
is the component of the vector $\bfk$ belonging to the dual lattice \eqref{duallat}.
Due to the constraint on $\Delta\mu$, the expression in the round brackets in the first term is an integer,
so that $q_i \in \kappa \Nr_i \IZ + \mu_i + \kappa \Nr_i/2$, as required.
It also easy to check that the identification \eqref{rel-charges} maps the quadratic form $-Q_n(\bfhgam)$ \eqref{defQlr}
to the one in \eqref{bf-r-ap}.

Finally, we observe that the sum over $\bfnu$ in the anomaly equation for the anomalous coefficients 
or their refined version (see, e.g., \eqref{refexp-whgi}) is nothing but 
the sum over the glue vectors of the lattice factorization \eqref{lat-decomp-r}.
To see this explicitly, let us introduce another $m$-tuple
$(\hrho_1,\dots,\hrho_m)$ such that $\sum_{k=1}^m \gcd(\frr_k)\,\hrho_k=\Nr_0$.
It allows to take $\bfrho=\sum_{k=1}^m \hrho_k\bfrho^{(k)}$.
Then it is straightforward to verify that, identifying 
$\csf_k=\frac{\nu_k-\sum_{i=1}^{n_k}\mu_{j_k+i}}{\kappa\gcd(\frr_k)}-\frac{\hrho_k\Delta\mu}{\kappa\Nr_0}$, 
one gets the relation
\be  
\bfhmu+\bfg_\Csf
=\bfhmu^{(\bfs)}\Bigl(\bfnu,\tfrac{\mu-\sum_{k=1}^m\nu_k}{\kappa s_0}\Bigr)
+\sum_{k=1}^m\bfhmu^{(\frr_k)}\Bigl(\bfmu,\tfrac{\nu_k-\sum_{i=1}^{n_k}\mu_{j_k+i}}{\kappa\gcd(\frr_k)}\Bigr),
\label{relgluevec}
\ee 
where the two vectors on the r.h.s. belong to $\bfDi{\bfs}$ and $\bfDi{\frr_k}$, respectively. 
This demonstrates that the lattice factorization \eqref{lat-decomp-r} gives rise to factors labelled 
by $(\mu,\bfnu)$ and $(\nu_k,\frm_k)$,
as expected. The ranges of summations over $\nu_k$ also agree, while
the conditions restricting the values of $\nu_k$ originating from \eqref{setC} and the above identification,
namely
\be  
\mu-\sum_{k=1}^m\nu_k\in \kappa s_0\IZ,
\qquad
\nu_k-\sum_{i=1}^{n_k}\mu_{j_k+i}\in \kappa\gcd(\frr_k)\,\IZ,
\ee 
in equations like \eqref{refexp-whgi} arise from the conditions on the sum over charges defining each of the factors.

Let us apply what we have just shown to theta series.
To this end, let 
\be
\Phi(\bfx)=\Phi_0(\bfx^{(0)})\prod_{k=1}^m \Phi_k(\bfx^{(k)}),
\ee 
where the upper indices $^{(0)}$ and $^{(k)}$ on a vector denote its projections to $\bfLami{\bfs}$ and $\bfLami{\frr_k}$,
respectively. We also assume that theta series associated to lattices $\bfLami{\bfr}$
and labelled by $(\mu,\bfmu)$ implicitly contain the factor $\delta^{(\kappa\Nr_0)}_{\Delta\mu}$.
Then our result implies that
\be  
\vth_{\mu,\bfmu}(\tau,\bftet z;\bfLami{\bfr},\Phi,0) =
\(\prod_{k=1}^m\sum_{\nu_k=0}^{\kappa s_k-1} \)\vth_{\mu,\bfnu}(\tau,\bftet^{(0)} z;\bfLami{\bfs},\Phi_0,0) 
\prod_{k=1}^m \vth_{\nu_k,\frm_k}(\tau,\bftet^{(k)} z;\bfLami{\frr_k},\Phi_k,0). 
\label{factortheta-rs}
\ee

\subsection{Basis expansion of glue vectors}

Here we obtain an expansion of the glue vectors appearing in the lattice decompositions \eqref{lat-glue-our} 
and \eqref{recover-L} in the basis $\ebbm_\lambda =(\ebbm_0,\ebbm_{i,\alpha},\hvbbm_{i,i+1},\hubbm_{i,i+1})$.
To this end, we note that the dual basis, i.e. such that $\ebbm_\lambda \ast \ebbm^*_{\lambda'}=\delta_{\lambda\lambda'}$, 
is given by 
\be
\label{glue-projectors-n}
\begin{split}
	\ebbm^*_0=&\, -\frac{1}{d_{\Nr}}\, \ebbm_0,
	\hspace{2.9cm} 
	\ebbm^*_{i,\alpha}=\glueg_{i,\alpha}-\frac{\alpha}{\di{i}}\, \glueg_{i},
	\\
	\hvbbm^*_{i,i+1} =&\,  \frac{1}{\Nr_0\kappa_i}\sum_{k=1}^{i}\sum_{l=i+1}^n\Nr_{kl}\hvbbm_{kl},
		\qquad
	\hubbm^*_{i,i+1} = - \frac{1}{4^\eps\Nr_0\kappa_i}\sum_{k=1}^{i}\sum_{l=i+1}^n \Nr_{kl}\hubbm_{kl},
\end{split}
\ee
where we used the notation $\kappa_i$ introduced in \eqref{defkapij}.
Then evaluating the scalar products of the vectors appearing in \eqref{gluegA} 
with the dual basis, one finds the following expansions
\be
\begin{split} 
\glueg_0 =&\,\sum_{l=1}^{n-1}
\frac{\kappa\Nr_0}{\kappa_l}\(\sum_{j=1}^{l}\sum_{k=l+1}^{n}\rdcr_j\rdcr_k(\rho_j-\rho_k) \) 
\hubbm_{l,l+1}
+\frac{1}{\rdcr}\,\ebbm_{0},
\\
\glueg_{i,\asf} =&\, 
4^{-\eps}\asf
\[\sum_{l=i}^{n-1}\frac{1}{\kappa_l}\(\sum_{k=l+1}^{n}\rdcr_k\)\hubbm_{l,l+1}
-\sum_{l=1}^{i-1}\frac{1}{\kappa_l}\(\sum_{k=1}^{l}\rdcr_k\)\hubbm_{l,l+1}\]
+\frac{\asf}{d_\Nr}\,\ebbm_{0}
\\
&\,
- \frac{1}{d_{\Nr_i}}\Biggl[\sum_{\alpha=1}^{\asf}\(d_{\Nr_i}-\asf\) \alpha\, \ebbm_{i,\alpha} 
+\sum_{\alpha=\asf+1}^{d_{\Nr_i}-1} 
\(d_{\Nr_i}-\alpha\)\asf\, \ebbm_{i,\alpha}\Biggr],
\end{split} 
\label{exp-g}
\ee
while for $\hvbbm_{ij}$ and $\hubbm_{ij}$, one gets
\be 
\hvbbm_{ij} =\frac{\Nr_i \Nr_j}{\Nr_{ij}}\sum_{l=i}^{j-1} \frac{\kappa\rdcr}{\kappa_l}\,\hvbbm_{l,\,l+1},
\qquad
\hubbm_{ij} =\frac{\Nr_i \Nr_j}{\Nr_{ij}} \sum_{l=i}^{j-1}\frac{\kappa\rdcr}{\kappa_l}\,\hubbm_{l,l+1}.
\label{exp-vw}
\ee

\subsection{Zero modes for two charges} 
\label{subsec-glue2}

In the case $n=2$, \eqref{exp-g} takes the form
\bea
\label{glueexp2}
\glueg_0 &=&  \frac{1}{\rdcr}\((\rho_1-\rho_2)\hubbm_{12} + \ebbm_0\),
\\
\glueg_{i,\asf} &=& \frac{\asf}{d_\Nr}
\(\frac{(-1)^{i-1}}{\rdcr_i}\,\hubbm_{12} + \ebbm_0\)  
- \frac{1}{d_{\Nr_i}}\(\sum_{\alpha=1}^{\asf}\(d_{\Nr_i}-\asf\) \alpha\, \ebbm_{i,\alpha} 
+\sum_{\alpha=\asf+1}^{d_{\Nr_i}-1} \(d_{\Nr_i}-\alpha\)\asf\, \ebbm_{i,\alpha}\).
\nn
\eea
which leads to the dependence of the indices of theta series, $\nu$ and $\tnu$, 
on the glue vectors given in \eqref{nutnu}.

Let us find a manifest description of the set of glue vectors ensuring the existence of zero modes
which is defined as 
\be
\cAr_0(\mu_{12})=\{\Asf\; :\ 2^\eps\tnu(\Asf)-\mu_{12}\in 2\kappa_{12}\IZ\}.
\ee
Taking into account \eqref{nutnu}, explicitly the zero mode condition reads
\be
2^\eps\kappa\Nr_0 \rdcr_1\rdcr_2(\rho_1-\rho_2) \asf_0 + 2^{-\eps}\(\rdcr_2 \asf_{1}-\rdcr_1\asf_{2}\)-\mu_{12} 
\in 2\kappa_{12}\IZ.
\label{glui-zmod}
\ee 
In the following, to treat the case $\eps=0$, we will have to distinguish two cases whether $\rdcr$ is even or not.
Therefore, we define $\eps'=\eps\delta^{(2)}_{\rdcr}$ and $\eps'_i=\eps\delta^{(2)}_{\rdcr_i}$.
We then take $\asf_i \in [0,d_{\Nr_i}-1]$ such that
\be
\asf_i = (-1)^{i-1}2^{\eps+\eps'}\kappa \Nr_0 \rdcr_1 \rdcr_2 
\left\{\frac{\mu_{12}}{2^{\eps'} \kappa \Nr_0 \rdcr_1 \rdcr_2}\right\} \rho_{3-i} +\ell\rdcr_i 
+\eps'_i (1-\eps')\,2^{1+\eps'}\kappa\Nr_0\rdcr_1 \rdcr_2\veps
\mod d_{\Nr_i},
\label{zm-alphai}
\ee
where $\{x\}$ denotes the fractional part of $x$, $\ell=0,\dots,4^\eps\kappa\Nr_0-1$
and $\veps=0,2^{\eps}-1$. Note that the last term is non-vanishing only for 
$\kappa=1$, odd $\rdcr$ and even $\rdcr_i$. 
We also observe that the $\ell$-dependent term cancels in \eqref{glui-zmod} and 
the transformation $\asf_i \to \asf_i + d_{\Nr_i}$ induces a shift of 
$2^{-\eps}(\rdcr_2 \asf_1 - \rdcr_1 \asf_2)$ by $\pm 2^\eps\kappa \Nr_0 \rdcr_1 \rdcr_2$.
Therefore, $\asf_i$ specified in \eqref{zm-alphai} satisfy
\be
2^{-\eps}(\rdcr_2 \asf_1 - \rdcr_1 \asf_2) =\mu_{12}+ 2^{\eps'}\kappa \Nr_0 \rdcr_1 \rdcr_2 m, 
\qquad m\in \IZ.
\ee
Substituting this into \eqref{glui-zmod}, one reduces the condition to 
\be
2^{\eps-\eps'}(\rho_1-\rho_2) \asf_0- 2^{-\eps'}\rdcr n=m, 
\qquad 
n \in \IZ.
\label{cond-amn}
\ee 
Note that $2^{-\eps'}\rdcr$ is integer and, due to the possibility to rewrite the condition on $\rho_i$ 
as $(\rho_1-\rho_2)\rdcr_1+\rho_2\rdcr=1$, one has $\gcd(\rho_1-\rho_2,\rdcr)=1$, which also 
implies that $\gcd(2^{\eps-\eps'}(\rho_1-\rho_2),2^{-\eps'}\rdcr)=1$.
Thus, if we restrict $\asf_0$ to belong to the interval $[0,\rdcr-1]$, 
for $\eps'=0$ and given $m$, there is a unique pair $(\asf_0,n)$ satisfying \eqref{cond-amn},
while for $\eps'=1$ there are two such pairs $(\asf_0+\hf\rdcr\veps,n+(\rho_1-\rho_2)\veps)$, $\veps=\{0,2^\eps-1\}$.
Thus, the equations \eqref{zm-alphai} and \eqref{cond-amn}
encode $2^{3\eps} \kappa \Nr_0$ solutions to the condition \eqref{glui-zmod}
parametrized by $(\ell,\veps)$ and provide an explicit description of the set $\cAr_0(\mu_{12})$.

\subsection{Zero modes of maximal order} 
\label{subsec-zm}

For $n$ charges, the maximal order of zero modes is $n-1$ and their set is determined by the conditions
$\wbbm_{ij}\ast\kbbm=0$ for all $i,j$ and $\kbbm\in \bbLami{\bfr}+\bbmu$. 
In practice we work with the lattice decompositions \eqref{lat-glue-our} and \eqref{recover-L}.
Therefore, $\kbbm$ can be replaced by $\kbbm_{||} \in\bbLami{\bfr}_{||}+\glueg_\Asf^{||}+\bbmu$
which is expanded as in \eqref{expkparal} with coefficients whose fractional parts are determined by 
the glue vectors of the two decompositions, labelled by $\Asf=\{\asf_0,\asf_1,\dots,\asf_n\}$ and 
$\Bsf=\{\bsf_{ij},\tbsf_{ij}\}$, respectively.
What we are interested in is the set of glue vectors for which the space of solutions to  
the above conditions is non-empty.

To describe the resulting set, let us define $\Delta\bfhmu=\bfhmu-2^\eps\bfhtmu$
where $\bfhmu$ is the vector from \eqref{thetadata}, while $\bfhtmu$ is the same vector with $\kappa$ multiplied by $4^\eps$ 
and $\mu$, $\mu_i$ replaced by $\tmu$, $\tmu_i$ computed in \eqref{tildemus}. This latter vector can be seen
as the projection of $\glueg_\Asf$ on $\tbfLami{\bfr}$. Using these definitions, one finds
that the components of $\Delta\bfhmu$ have the following explicit expressions
\be  
\Delta\hmu_i
=\frac{\mu_i-2^{-\eps}\asf_i}{\kappa\Nr_i}-\frac{\mu-2^{-\eps}\sum_{j=1}^n\asf_j}{\kappa\Nr}
+\frac{2^\eps \Nr_0\asf_0}{\Nr}+\(\frac{\Delta\mu}{\kappa\Nr_0}-2^\eps\asf_0\)\rho_i.
\label{comp-dhmu}
\ee 
Then we summarize the main facts about the desired set in the following

\begin{proposition}\label{prop-zm}
The set of glue vectors of zero modes of maximal order can be characterized by two sets of conditions:
\begin{enumerate}
	\item 
	the first specifies a set $\cAr_0$ of $\Asf$ indices and is independent of $\Bsf$:
	\be 
	\Delta\hmu_i(\Asf) \in\IZ,
	\qquad
	i=1,\dots n;
	\label{cond-glueA}
	\ee 
	\item
	the second restricts $\Bsf$ indices: 
	\be
	\bsf_{ij}-2^\eps \tbsf_{ij}=-c_{ij}(\Asf) \mod N_{ij},
	\qquad
	i<j-1,
	\label{cond-glueB}
	\ee 
	where $N_{ij}$ is given in \eqref{valNij} and $c_{ij}(\Asf)$ are the coefficients in the expansion
	\be 
	\Delta\bfhmu(\Asf)=\sum_{i<j}c_{ij}(\Asf)\,\bfhv_{ij},
	\label{decomDhmu}
	\ee 
	which is unique provided $\Asf\in\cAr_0$ and $c_{ij}(\Asf)$ are required to be integers
	ranging from 0 to $N_{ij}-1$.
\end{enumerate} 
\end{proposition}
\begin{proof}
Let $\kbbm_{||}=\bblambda+\glueg_\Asf^{||}+\bbmu$ where $\bblambda=(\bflam,\bftlam)\in \bbLami{\bfr}_{||}$.
Given the form of the null vectors $\wbbm_{ij}$ \eqref{def-bfcvij}, the maximal order zero mode condition 
can be written as 
\be 
\bfv_{ij} \cdot(\bflam-2^\eps\bftlam +\bfhmu-2^\eps\bfhtmu)=0
\qquad \mbox{for all } i,j,
\label{maxzm-prof}
\ee 
where we used the bilinear form \eqref{bf-r} and took into account that the bilinear form on $\tbfLami{\bfr}$ differs
by the factor $-4^\eps$. It is clear that the condition \eqref{maxzm-prof} is equivalent to 
\be 
\bflam-2^\eps\bftlam=-\Delta\bfhmu,
\label{maxzm-prof2}
\ee
which in turn implies that $\Delta\bfhmu\in\bfLami{\bfr}$. Since the components \eqref{comp-dhmu}
automatically satisfy the condition $\sum_{i=1}^n\Nr_i\Delta\hmu_i=0$, 
the only remaining condition is that they must be integer.
This is precisely the condition \eqref{cond-glueA} which provides the definition of the set $\cAr_0$.

Next, each vector in \eqref{maxzm-prof2} has a unique decomposition according to the lattice decomposition in \eqref{recover-L}.
In particular, $\bflam=\bflam'+\sum_{i,j=1\atop j-i>1}^{n}\bsf_{ij} \bfhv_{ij}$ and 
$\bftlam=\bftlam'+\sum_{i,j=1\atop j-i>1}^{n}\tbsf_{ij} \bfhv_{ij}$ where
$\bflam',\bftlam'\in \Span\{\bfhv_{k,k+1}\}_{k=1}^{n-1}$, while $\Delta\bfhmu$ can be written as in \eqref{decomDhmu}
where $c_{ij}$ with $j-i>1$ play the role of the glue vector indices.
Then the second set of conditions \eqref{cond-glueB} claimed by the proposition is a direct consequence of 
\eqref{maxzm-prof2}.
\end{proof}

\begin{remark}\label{remark-zm}
In fact, due to $\sum_{i=1}^n\Nr_i\Delta\hmu_i=0$, a stronger form of the condition \eqref{cond-glueA} holds: 	
$\Delta\hmu_i \in \gcd(\{\rdcr_j\}_{j\ne i})\IZ$. 
In particular, this should be taken into account to reproduce the condition \eqref{glui-zmod}
in the $n=2$ case.
\end{remark}

\begin{corollary}\label{cor-sumB}
Each pair $(\bsf_{ij},\tbsf_{ij})$ contributes $N_{ij}$ solutions so that their total number is
given by \eqref{total-valNij} and is $\Asf$-independent. 
As a result, if $\Dzm(\Asf,\Bsf)$ is the Kronecker symbol 
implementing the maximal zero mode condition,
then 
\be 
\sum_B \Dzm(\Asf,\Bsf)=\frac{\Nr_0\prod_{j=2}^{n-1} \Nr_i }{\prod_{i=1}^{n-1}\Nr_{i,i+1}}\, \delta_{\Asf\in\cAr_0} .
\label{sumB}
\ee 
\end{corollary}

Although Proposition \ref{prop-zm} gives a formula \eqref{cond-glueB}
for the $\Bsf$ indices corresponding to the zero modes, it does not tell us how to compute
the coefficients $c_{ij}$ encoding the solution. Let us explain a simple recipe how this can be done.
We start with the first component of the expansion \eqref{decomDhmu} 
multiplied by the factor $\frac{\Nr_{1n}}{\Nr_{2::n}}$ so that, according to Remark \ref{remark-zm},
the result is still an integer. It is found to be
\be 
\frac{\Nr_{1n}}{\Nr_{2::n}}\, \Delta\hmu_1 =\sum_{j=2}^n \frac{\Nr_{1n}\Nr_j}{\Nr_{1j}\Nr_{2::n}} \, c_{1j}
=\frac{\Nr_n}{\Nr_{2::n}} \,c_{1n} \mod N_{1n},
\ee 
where we used the fact that all $\frac{\Nr_{1n}\Nr_j}{\Nr_{1j}\Nr_{2::n}}$ with $2\leq j\leq  n-1$ are divisible by $N_{1n}$, 
as follows from its expression in \eqref{valNij}. 
Furthermore, it is easy to see that $\gcd\bigl(\frac{\Nr_n}{\Nr_{2::n}},N_{1n}\bigr)=1$.
This implies that there exist $a_{1n},b_{1n}\in\IZ$ such that $a_{1n}\,\frac{\Nr_n}{\Nr_{2::n}}+b_{1n}N_{1n}=1$, i.e.
\be 
a_{1n}\, \frac{\Nr_{1n}}{\Nr_{2::n}}\,\Delta\hmu_1=c_{1n}\mod N_{1n}.
\ee 
Then one subtracts the already found part of the expansion \eqref{decomDhmu} from $\Delta\bfhmu$ and 
repeats the procedure, first by lowering the second index and then by raising the first index.
Thus, to determine $c_{ij}$, one chooses $a_{ij},b_{ij}\in\IZ$ such that $a_{ij}\,\frac{\Nr_j}{\Nr_{i+1::j}}+b_{ij}N_{ij}=1$
and uses the fact that
\be 
\begin{split} 
&\,
a_{ij}\,\frac{\Nr_{ij}}{\Nr_{i+1::j}}\(\Delta\hmu_i-\sum_{k=1}^{i-1}\sum_{l=k+1}^n c_{kl}\,(\bfhv_{kl})_i-\sum_{l=j+1}^n c_{il}\,(\bfhv_{kl})_i \)
\\
= &\, a_{ij}\sum_{l=i+1}^{j}\frac{\Nr_{ij}\Nr_l}{\Nr_{il}\Nr_{i+1::j}} \,  c_{il}
= a_{ij}\( \frac{\Nr_j}{\Nr_{i+1::j}} \, c_{ij} \!\!\!\! \mod N_{ij}\)
=c_{ij} \mod N_{ij}.
\end{split} 
\label{findcij}
\ee

\section{Computations for three charges}
\label{ap-3ch}

\subsection{Contributions with zero mode order equal 0 and 1}
\label{subsec-ch3-zm01}

For three charges, the kernel of the theta series $\vthls{\bfr}_{\bbmu,\Asf}$ defined in \eqref{defprojtheta} 
is given in \eqref{expr-Phi2} in terms of $\gamma_{ij}$. 
The scalar products $\xbbm_\bbbeta\ast\vbbm_{ij}$ appearing in the kernel are given by the relation \eqref{rel-scpr}.
Similarly, we define for $i<j$ 
\be 
\begin{split} 
\om_{ij}=&\, \frac{1}{\sqrt{2\tau_2}}\, \wbbm_{ij}\ast\xbbm_\bbbeta= \wbbm_{ij}\ast\kbbm\, , 
\\
\omb_{ij}=&\, \frac{1}{\sqrt{2\tau_2}}\, \wbbm_{ij}\ast\xbbm= \wbbm_{ij}\ast(\kbbm+\bbbeta)
\\
&\qquad \qquad\qquad
=\om_{ij}+\left\{\begin{array}{ll}
	2^\eps\kappa (\Nr_i+\Nr_j)\Nr_i\Nr_j\beta,&\quad (ij)=(12),(23),
	\\
	2^\eps\kappa (\Nr+\Nr_2)\Nr_1\Nr_3\beta,&\quad (ij)=(13).
	\end{array}\right.
\end{split}
\label{def-bij}
\ee
For $i>j$, $\om_{ij}$ and $\omb_{ij}$ are defined by applying the permutation $i\leftrightarrow j$ to \eqref{def-bij}.
More generally, permutations leave $\kbbm$ invariant and permute the indices of all vectors like $\vbbm_{ij}$ and $\wbbm_{ij}$.
Note that the component of $\zbbm$ along $\bfLami{\bfr}$ is proportional to $\bftet$ (see \eqref{thetadata}) 
defined in terms of the vectors $\bfv_{ij}$ and therefore it is affected by the permutations so as $\bbbeta$.
This explains why $\om_{ij}$ are anti-symmetric, whereas this is not true for $\omb_{ij}$ 
as follows from the last line in \eqref{def-bij}.
We also have $\om_{i+j,k}=\om_{ik}+\om_{jk}$ and similarly for $\omb_{1+2,3}$ and $\omb_{1,2+3}$.
Using these properties, one can write 
\be 
\begin{split}
\Sym\bigl\{\vthls{\bfr}_{\bbmu,\Asf}(\tau, z) \bigr\}=&\, \sum_{\kbbm_{||} \in\bbLami{\bfr}_{||}+\glueg_\Asf^{||}+\bbmu}
\q^{-\hf\kbbm_{||}^2}\Sym\Biggl\{  y^{\gamma_{1,2+3}+\gamma_{23}} 
\biggl[\frac13\, \delta_{\gamma_{1,2+3}}\delta_{\gamma_{1+2,3}} 
\\
&
+\Bigl(\sgn(\gamma_{1,2+3})-\sgn(\omb_{12})\Bigr)\Bigl(\sgn(\gamma_{1+2,3})-\sgn(\omb_{23})\Bigr) \biggr]
\\
&
+\delta_{\om_{12}}\phi^{(\kappa_{12})}(\tau,\Nr_{12} z) 
\Bigl[ y^{\gamma_{1+2,3}}
\Bigl(\sgn(\gamma_{1+2,3})-\sgn(\omb_{1+2,3})\Bigr)
\\
&\qquad
+y^{-\gamma_{1+2,3}}\Bigl(\sgn(\gamma_{3,1+2})-\sgn(\omb_{3,1+2})\Bigr)\Bigr]
\sum_{\eta=1}^{2^{1+\eps}\kappa_{12}}\delta_{\frac{\gamma_{12}}{r_{12}}-\eta}
\Biggr\},
\label{theta-par3}
\end{split}
\ee
where we applied the permutation $1\leftrightarrow 3$ to the last term in \eqref{expr-Phi2}.

Let us extract the contribution of vanishing zero mode order. It is given by the first two lines in \eqref{theta-par3}
where all $\om_{ij}$ are non-vanishing. Under this condition, for small $\beta$ one can replace 
$\sgn(\omb_{ij})$ by $\sgn(\om_{ij})$ so that this term becomes
\be
\Sym  \left\{ y^{\gamma_{1,2+3}+\gamma_{23}} \[\frac13\, \delta_{\gamma_{1,2+3}}\delta_{\gamma_{1+2,3}} 
+\Bigl(\sgn(\gamma_{1,2+3})-\sgn(\om_{12})\Bigr)\Bigl(\sgn(\gamma_{1+2,3})-\sgn(\om_{23})\Bigr) \]\right\}.
\ee
We expand the product and apply the sign identity 
\be
\sgn(x_1)\sgn(x_2) = 1-\delta_{x_1}\delta_{x_2} - \sgn(x_1-x_2)\(\sgn(x_1) - \sgn(x_2) \)
\label{sign-id}
\ee
to $\frac13$ of the terms with two $\gamma$'s and two $\om$'s. This results in 
\bea
&&
\Sym  \left\{y^{\gamma_{1,2+3}+\gamma_{23}} \bigg[
\frac23\, \sgn(\gamma_{1,2+3})\sgn(\gamma_{1+2,3}) 
-\frac13\, \sgn(\gamma_{1+3,2})\Bigl(\sgn(\gamma_{1,2+3})-\sgn(\gamma_{1+2,3})\Bigr)
\right.
\nn\\
&&\qquad
+\frac23\, \sgn(\om_{12})\sgn(\om_{23}) + \frac13\, \sgn(\om_{13})\Bigl(\sgn(\om_{12})+\sgn(\om_{23})\Bigr)
\\
&&\left. \qquad
-\sgn(\gamma_{1,2+3})\sgn(\om_{23})-\sgn(\om_{12})\sgn(\gamma_{1+2,3})\bigg]\right\},
\nn
\eea
where we used $\gamma_{1,2+3}-\gamma_{1+2,3}=\gamma_{1+3,2}$ and $\Nr_3\om_{12}+\Nr_1\om_{23}=\Nr_2\om_{13}$
following from \eqref{rel-vw}. Then we use the symmetrization 
to bring the products of sign functions in each line to a single expression. This gives
\be
\begin{split}
	&\, \Sym  \biggl\{\(
	\frac13\, \sgn(\gamma_{1,2+3})\sgn(\gamma_{1+2,3})+\frac13\, \sgn(\om_{12})\sgn(\om_{23}) 
	-\hf\, \sgn(\gamma_{1,2+3})\sgn(\om_{23})\)
	\\
	& \qquad
	\times\Bigl(y^{\gamma_{1,2+3}}-y^{-\gamma_{1,2+3}}\Bigr)
	\Bigl(y^{\gamma_{23}}-y^{-\gamma_{23}}\Bigr)\biggr\}.
\end{split}
\label{nzmod-n3}
\ee
This result makes it manifest that the contribution of vanishing zero mode order has a zero of second order at $z=0$.

Next, we consider zero mode contributions of order 1. They are characterized by vanishing of only one of $\om_{ij}$'s.
This will be indicated by insertion of the corresponding Kronecker symbol $\delta_{\om_{ij}}$, but we will omit 
the factors $1-\delta_{\om_{jk}}$ ensuring that other variables are non-vanishing. 
Thus, the relevant contribution to \eqref{theta-par3} reads
\be
\begin{split}
	&\Sym \Biggl\{ y^{\gamma_{1,2+3}+\gamma_{23}} \[
	\frac13\, \delta_{\gamma_{1,2+3}}\delta_{\gamma_{1+2,3}}\(\delta_{\om_{12}}+\delta_{\om_{23}}+\delta_{\om_{13}}\)
	\right.
	\\
	&\qquad\qquad
	+\Bigl(\sgn(\gamma_{1,2+3})-\sgn(\beta)\Bigr)\Bigl(\sgn(\gamma_{1+2,3})-\sgn( \om_{23})\Bigr) \delta_{\om_{12}}
	\\
	&\qquad\qquad
	+\Bigl(\sgn(\gamma_{1,2+3})-\sgn( \om_{12})\Bigr)\Bigl(\sgn(\gamma_{1+2,3})-\sgn(\beta)\Bigr) \delta_{\om_{23}}
	\\
	&\qquad\qquad
	+\Bigl(\sgn(\gamma_{1,2+3})+\sgn( \om_{23})\Bigr)\Bigl(\sgn(\gamma_{1+2,3})-\sgn( \om_{23})\Bigr) \delta_{\om_{13}}
	\biggr]
	\\
	&\quad
	+\delta_{\om_{12}}
	\Bigl(y^{\gamma_{1+2,3}}-y^{-\gamma_{1+2,3}}\Bigr)
	\Bigl(\sgn(\gamma_{1+2,3})-\sgn(\om_{1+2,3})\Bigr)\phi^{(\kappa_{12})}
	\sum_{\eta=1}^{2^{1+\eps}\kappa_{12}}\delta_{\frac{\gamma_{12}}{r_{12}}-\eta}\Biggr\},
\end{split}
\label{zmod1-n3}
\ee
where in the forth line we used that $\sgn (\om_{12}) = -\sgn (\om_{23})$ provided $\om_{13}=0$. 
Then we use the symmetrization to transform all $\delta_{\om_{ij}}$ to $\delta_{\om_{12}}$, 
and add and subtract a term proportional to $\sgn(\gamma_{12})$, which is recombined with the terms 
coming from the second and third lines of \eqref{zmod1-n3}.
Taking also into account that $\sgn(\om_{23})=\sgn(\om_{1+2,3})$ provided $\om_{12}=0$, 
we end up with the following result
\bea
&&\Sym \Biggl\{\delta_{\om_{12}}\Biggl[
\frac13 \(y^{\gamma_{1,2+3}+\gamma_{23}}+y^{-\gamma_{1,2+3}-\gamma_{23}}+y^{\gamma_{1,2+3}-\gamma_{23}}\)
\delta_{\gamma_{1,2+3}}\delta_{\gamma_{1+2,3}}
\nn\\
&&\qquad 
+\(y^{\gamma_{1,2+3}+\gamma_{23}}+y^{-\gamma_{1,2+3}-\gamma_{23}}\)
\Bigl(\sgn(\gamma_{1+2,3})-\sgn( \om_{23})\Bigr)\Bigl(\sgn(\gamma_{1,2+3})-\sgn(\gamma_{12})\Bigr) 
\nn\\
&&\qquad 
+y^{\gamma_{1,2+3}-\gamma_{23}}
\Bigl(\sgn(\gamma_{1,2+3})-\sgn( \om_{23})\Bigr)\Bigl(\sgn(\gamma_{1+3,2})+\sgn( \om_{23})\Bigr) 
\label{zmod-lev1-n=3} \\
&&\qquad
+\Bigl(\sgn(\gamma_{1+2,3})-\sgn(\om_{1+2,3})\Bigr)\Biggl(
\Bigl(y^{\gamma_{1+2,3}}-y^{-\gamma_{1+2,3}}\Bigr)\phi^{(\kappa_{12})}
\sum_{\eta=1}^{2^{1+\eps}\kappa_{12}}\delta_{\frac{\gamma_{12}}{r_{12}}-\eta}
\nn\\
&&\qquad\quad
+y^{\gamma_{1+2,3}+\gamma_{12}}\Bigl(\sgn(\gamma_{12})-\sgn(\beta)\Bigr)
+y^{-\gamma_{1+2,3}-\gamma_{12}}\Bigl(\sgn(\gamma_{12})+\sgn(\beta)\Bigr)
\Biggr)\Biggr]\Biggr\}.
\nn
\eea
Expanding the product in the third line and applying the identity \eqref{sign-id} to 
the product of two $\sgn(\gamma_{i+j,k})$, one finds that it is equal to
\be
\Sym \Bigl\{\delta_{\om_{12}}\, y^{\gamma_{1,2+3}-\gamma_{23}}\Bigl[
\Bigl(\sgn(\gamma_{1+2,3})-\sgn( \om_{23})\Bigr)\Bigl(\sgn(\gamma_{1+3,2})-\sgn(\gamma_{1,2+3})\Bigr) 
-\delta_{\gamma_{1,2+3}}\delta_{\gamma_{1+3,2}} 
\Bigr] \Bigr\}.
\ee 
The last term cancels the first line in \eqref{zmod-lev1-n=3} (to see this, one should use the vanishing 
condition imposed by the Kronecker symbols and the symmetrization $1\leftrightarrow 2$),
while the first term can be combined with the second line in that equation and,
after the same symmetrization and the replacement of $\sgn(\om_{23})$ by $\sgn(\om_{1+2,3})$, gives
\be
\begin{split}
	& \Sym\biggl\{ \delta_{\om_{12}}
	\Bigl(y^{\gamma_{1,2+3}}-y^{-\gamma_{1,2+3}}\Bigr) 
	\Bigl(y^{\gamma_{23}}-y^{-\gamma_{23}}\Bigr)
	\\
	&\qquad\times
	\Bigl(\sgn(\gamma_{1+2,3})-\sgn( \om_{1+2,3})\Bigr) \Bigl(\sgn(\gamma_{1,2+3})-\sgn(\gamma_{12})\Bigr)
	\biggr\}.
\end{split}
\label{doublez-contr1}
\ee
Thus, this contribution also has a zero of second order at $z=0$.

It remains to compute the contribution of the last two lines in \eqref{zmod-lev1-n=3}, which after symmetrization 
takes the form
\be 
\begin{split}
&\Sym \Biggl\{\delta_{\om_{12}}\Bigl(y^{\gamma_{1+2,3}}-y^{-\gamma_{1+2,3}}\Bigr)
\Bigl(\sgn(\gamma_{1+2,3})-\sgn(\om_{1+2,3})\Bigr)
\\
&\qquad\times\Biggl(
\phi^{(\kappa_{12})}\sum_{\eta=1}^{2^{1+\eps}\kappa_{12}}\delta_{\frac{\gamma_{12}}{r_{12}}-\eta}
+y^{\gamma_{12}}\Bigl(\sgn(\gamma_{12})-\sgn(\beta)\Bigr)
\Biggr)\Biggr\}.
\end{split}
\label{contrib-pole1zm} 
\ee 
We observe that it factorizes into two parts: the second line depends only on $\gamma_{12}$,
while the first line together with the factor $\q^{-\hf\kbbm_{||}^2}$ can be shown 
to depend on $\gamma_{12}$ only $\mod 2^{1+\eps}\Nr_{12} \kappa_{12}$.
Indeed, it is easy to verify that, provided $\om_{12}=0$, the transformation 
$\kbbm_{||}\mapsto \kbbm_{||}+\hwbbm_{12}$ affects neither $\kbbm_{||}^2$, nor the quantities   
appearing in the first line: $\om_{12}$, $\om_{1+2,3}$, $\gamma_{1+2,3}$.
On the other hand, it maps $\gamma_{12} \to \gamma_{12} + 2^{1+\eps}\Nr_{12} \kappa_{12}$,
which proves the claim.
This means that, substituted into the theta series, the contribution \eqref{contrib-pole1zm} is equivalent to
\be
\begin{split}
    & \Sym \Biggl\{ \delta_{\om_{12}}\Bigl(y^{\gamma_{1+2,3}}-y^{-\gamma_{1+2,3}}\Bigr)
    	\Bigl(\sgn(\gamma_{1+2,3})-\sgn(\om_{1+2,3})\Bigr)
    \\
    & \qquad\times 
    \sum_{\eta=1}^{2^{1+\eps}\kappa_{12}}\delta_{\frac{\gamma_{12}}{r_{12}}-\eta}	
	\(\phi^{(\kappa_{12})}+\sum_{\tell \in \IZ} 
	\Bigl(\sgn(\eta+2^{1+\eps}\kappa_{12}\tell)-\sgn(\beta)\Bigr) 
	\,y^{\Nr_{12}(\eta+2^{1+\eps}\kappa_{12}\tell)}\) 
	\Biggr\}.
\end{split}
\label{contrib-pole1zm-2} 
\ee
The sum over $\tell$ is identical to the one in \eqref{geompr} (where the sum over $\sigma$ is equivalent to 
the symmetrization with respect to $1\leftrightarrow 2$ in \eqref{contrib-pole1zm-2}).
Thus, we can borrow the result \eqref{zmod-sigmap1}, which gives
\be
\begin{split}
&
\Sym\Biggl\{ \delta_{\om_{12}}\Bigl(y^{\gamma_{1+2,3}}-y^{-\gamma_{1+2,3}}\Bigr)
\Bigl(\sgn(\gamma_{1+2,3})-\sgn(\om_{1+2,3})\Bigr)
\\
& \qquad\times 
\sum_{\eta=1}^{2^{1+\eps}\kappa_{12}}\delta_{\frac{\gamma_{12}}{r_{12}}-\eta}	
\(\phi^{(\kappa_{12})}
-\frac{y^{2^{1+\eps}\Nr_{12} \kappa_{12}\lambda_{12}}+y^{-2^{1+\eps}\Nr_{12} \kappa_{12}\lambda_{12}}}
{y^{2^{\eps}\Nr_{12} \kappa_{12}}-y^{-2^{\eps}\Nr_{12} \kappa_{12}}} \) 
\Biggr\},
\end{split}
\label{doublez-contr2}
\ee
where $\lambda_{12}$ is given by \eqref{def-lam0} with $\tnu$ replaced by $2^{-\eps}\eta$. 
Since $\phi^{(\kappa_{12})}$ was chosen precisely to cancel the pole coming from the geometric progression,
the last bracket behaves as $\cO(z)$ and the contribution \eqref{doublez-contr2} has a zero of second order at $z=0$.
As a result, the same conclusion applies to the total zero mode contribution of order 1 to the summand in \eqref{theta-par3},
given by the sum of \eqref{doublez-contr1} and \eqref{doublez-contr2}.

\subsection{The unrefined limit}
\label{subsec-unrefn=3}

We define the theta series appearing in \eqref{res-unref3} as 
\be 
\vthpzi{\bfr}{k}_{\bbmu,\Asf}(\tau)=\lim_{z\rightarrow 0} 
\frac{\bigl[\Sym\bigl\{\vthls{\bfr}_{\bbmu,\Asf}(\tau, z)\bigr\}\bigr]_k}{(y-y^{-1})^2},
\label{deftheta-unref}
\ee 
where $[\;\cdot\;]_k$ denotes the contribution of zero mode order equal $k$.
In particular, the kernel of the contribution of vanishing order is given by \eqref{nzmod-n3}, 
and the one of order 1 is given by the sum of \eqref{doublez-contr1} and \eqref{doublez-contr2}.

First, we represent the theta series \eqref{deftheta-unref} as a sum over the lattice $\bbLami{\bfr}_{||}$ 
with the kernel expressed through $\gamma_{ij}$ and $\omega_{ij}$ as in \eqref{theta-par3}.
One easily finds
\bea
\vthpzi{\bfr}{0}_{\bbmu,\Asf}(\tau)&=& 
\sum_{\kbbm_{||} \in\bbLami{\bfr}_{||}+\glueg_\Asf^{||}+\bbmu}\nDzm\Sym\biggl\{
\gamma_{23}\gamma_{1,2+3} \(\frac13\, \sgn(\gamma_{1,2+3})\sgn(\gamma_{1+2,3})
+\frac13\, \sgn(\om_{12})\sgn(\om_{23}) 
\right.\nn\\
&&\left.\qquad
-\hf\, \sgn(\gamma_{1,2+3})\sgn(\om_{23})\)\biggr\}\,\q^{-\hf\kbbm_{||}^2} ,
\label{vth-zmk}\\
\vthpzi{\bfr}{1}_{\bbmu,\Asf}(\tau)&=&
\!\!\!\!\sum_{\kbbm_{||} \in\bbLami{\bfr}_{||}+\glueg_\Asf^{||}+\bbmu} \!\!\!\!
\Sym\biggl\{
\delta_{\om_{12}}(1-\delta_{\om_{23}})\gamma_{1,2+3} \Bigl(\sgn(\gamma_{1+2,3})-\sgn( \om_{1+2,3})\Bigr)
\nn\\
&&\times \biggl[ \gamma_{23}
\Bigl(\sgn(\gamma_{1,2+3})-\sgn(\gamma_{12})\Bigr)
+\frac{2^\eps\Nr_{12}}{12}\(1- 12 \lambda_{12}^2 - \frac{E_2(\tau)}{2^{1+2\eps} \kappa_{12}}\) 
\sum_{\eta=1}^{2^{1+\eps}\kappa_{12}}\delta_{\frac{\gamma_{12}}{r_{12}}-\eta}	
\biggr]\biggr\}\,
\q^{-\hf\kbbm_{||}^2},
\nn 
\eea
where the factor $\nDzm=(1-\delta_{\om_{12}})(1-\delta_{\om_{23}})(1-\delta_{\om_{13}})$ implements the condition 
of vanishing zero mode order. 
Note that in the first theta series the symmetrization over charges must be implemented {\it before} performing the sum. 
The reason is that in the process of getting the kernel  
that makes manifest the existence of a second order zero and allows to compute the unrefined limit, 
the conditions of convergence ceased to be satisfied. This happened due to the use of various permutations
which spoiled the original structure of the kernel satisfying Theorem \ref{th-conv}. 
Of course, one can do them backwards to recover 
the manifest convergence. We omit the details of manipulations, which are similar to the ones in the previous subsection,
and present just the final result
\be
\vthpzi{\bfr}{0}_{\bbmu,\Asf}(\tau)=
\frac16
\!\!\sum_{\kbbm_{||} \in\bbLami{\bfr}_{||}+\glueg_\Asf^{||}+\bbmu}\!\!\!\!
\nDzm (1+\sigma_{13})\biggl[
\gamma_{23}\gamma_{1,2+3} \Bigl(\sgn(\gamma_{1,2+3})-\sgn(\om_{13})\Bigr)
\Bigl(\sgn(\gamma_{2,1+3})-\sgn(\om_{23})\Bigr)\biggr]\,\q^{-\hf\kbbm_{||}^2} ,
\ee
where $\sigma_{ij}$ denotes the permutation $i\leftrightarrow j$.
Note that although the symmetry of this expression under all permutations of charges is not manifest,
it is actually symmetric being obtained from a symmetrized expression.
For the second theta series in \eqref{vth-zmk} the issue of convergence is absent since it is ensured by the kernel 
so that the symmetrization can be put outside the sum.

For computer evaluation it might be useful also to rewrite the above expressions as a sum over an unconstrained lattice.
To this end, we can substitute the expansion \eqref{expkparal} of the lattice vector 
and rewrite the sum over the lattice $\bbLami{\bfr}_{||}$ as a sum over the coefficients $\ell_i$ and $\tell_i$ 
at the price of introducing new glue vectors $\Bsf=\{\bsf_{13},\tbsf_{13}\}$.
This leads to more explicit, although more cumbersome expressions:
\bea
\vthpzi{\bfr}{0}_{\bbmu,\Asf}(\tau)&=& \frac{\Nr_0}{6}\,(1+\sigma_{13})
\sum_\Bsf\sum_{\ell_i\in\IZ+\frac{\nu_i}{\kappa_i}}
\sum_{\tell_i\in\IZ+\frac{\tnu_i}{\kappa_i}} \nDzm 
\Bigl[\Nr_{23} \kappa_1 \ell_1 (2\kappa_{23}\ell_2-\kappa_{123}\ell_1)
\Bigl( \sgn(\ell_1)-\sgn\Bigl(\tfrac{\Delta\ell_1}{\Nr_{12}}+\tfrac{\Delta\ell_2}{\Nr_{23}}\Bigr)\Bigr)
\nn\\
&&
\times \Bigl(\sgn\Bigl(\tfrac{\Nr_3}{\Nr_{23}}\,\ell_2-\tfrac{\Nr_1}{\Nr_{12}}\,\ell_1\Bigr)
-\sgn\Bigl(\tfrac{\Nr_2+\Nr_3}{\Nr_{23}}\,\Delta\ell_2-\tfrac{\Nr_1}{\Nr_{12}}\,\Delta\ell_1\Bigr)\Bigr)
\Bigr]
\q^{\kappa_{12}(4^\eps \tell_1^2-\ell_1^2)+\kappa_{23}(4^\eps \tell_2^2-\ell_2^2)-\kappa_{123}(4^\eps\tell_1\tell_2-\ell_1\ell_2)},
\nn\\
\vthpzi{\bfr}{1}_{\bbmu,\Asf}(\tau)&=& \Nr_0\Sym\Biggl\{ \kappa_{1} \sum_\Bsf 
\sum_{\ell_i\in\IZ+\frac{\nu_i}{\kappa_i}}\sum_{\tell_i\in\IZ+\frac{\tnu_i}{\kappa_i}} 
\delta_{\frac{\Nr_1+\Nr_2}{\Nr_{12}}\,\Delta\ell_1-\frac{\Nr_3}{\Nr_{23}}\,\Delta\ell_2}
\(1-\delta_{\Delta\ell_1}\)\ell_1 \Bigl(\sgn(\ell_2)-\sgn(\Delta\ell_2)\Bigr)
\label{vth-zm01}\\
&&
\times \biggl[
\Nr_{23}(2\kappa_{23}\ell_2-\kappa_{123}\ell_1)
\Bigl(\sgn(\ell_1)-\sgn\Bigl(\tfrac{\Nr_1+\Nr_2}{\Nr_{12}}\,\ell_1-\tfrac{\Nr_3}{\Nr_{23}}\,\ell_2\Bigr)\Bigr)
\nn\\
&&\quad 
+\frac{2^\eps\Nr_{12}}{12}\(1- 12 \lambda_{12}^2 - \frac{E_2(\tau)}{2^{1+2\eps} \kappa_{12}}\) 
\!\!\sum_{\eta=1}^{2^{1+\eps}\kappa_{12}}\!\! \delta_{2\kappa_{12}\ell_1-\kappa_{123}\ell_2-\eta}	
\biggr]
\q^{\frac{\Nr_0\Nr_2\kappa_2}{2(\Nr_1+\Nr_2)}\,(4^\eps\tell_2^2-\ell_2^2)}\Biggr\},
\nn
\eea
where $\Delta\ell_i=\ell_i-2^\eps\tell_i$.
The variables $\nu_i$ and $\tnu_i$ determining the rational parts are given in \eqref{nu-n3}, 
and the factor implementing the condition of vanishing zero mode order reads as
$\nDzm=\(1-\delta_{\frac{\Nr_1+\Nr_2}{\Nr_{12}}\,\Delta\ell_1-\frac{\Nr_3}{\Nr_{23}}\,\Delta\ell_2}\)
\(1-\delta_{\frac{\Nr_2+\Nr_3}{\Nr_{23}}\,\Delta\ell_2-\frac{\Nr_1}{\Nr_{12}}\,\Delta\ell_1}\)
\(1-\delta_{\frac{\Delta\ell_1}{\Nr_{12}}+\frac{\Delta\ell_2}{\Nr_{23}}}\)$.

\section{Consistency of different solutions}
\label{ap-consist}

\subsection{Two charges}
\label{subsec-Jacobi2}

The solution \eqref{res-unref2} obtained using indefinite theta series
can be compared with the one constructed in section \ref{subsec-sol2ch} using Hecke-like operators.
The consistency requires that their difference is a VV modular form.
To verify whether this is the case, let us write the function \eqref{res-unref2}
as in \eqref{sol-n=2} 
\be 
\label{chG-n=2}
\gi{\Nr_1,\Nr_2}_{\mu, \mu_1, \mu_2}=\Nr_0 \delta^{(\kappa\Nr_0)}_{\Delta\mu} \chG^{(\kappa;\Nr_1,\Nr_2)}_{\mu_{12}}
\ee
and define 
\be
\vph^{(\kappa)}_{\Nr_1,\Nr_2}(\tau,z)=\sum_{\mu=0}^{2\kappa_{12}-1}
\Bigl(\chG^{(\kappa;\Nr_1,\Nr_2)}_{\mu}(\tau)-G^{(\kappa_{12})}_{\mu}(\tau)\Bigr)\ths{\kappa_{12}}_\mu(\tau,z),
\ee
where $G^{(\kappa)}$ is given by \eqref{DMZ} and $\ths{\kappa}_\mu(\tau,z)$ is the theta series \eqref{deftheta}. 
The modularity condition is equivalent to the requirement that $\vph^{(\kappa)}_{\Nr_1,\Nr_2}$ 
is a Jacobi form of weight 2, index $\kappa_{12}$ and trivial multiplier system.
Using Mathematica, for the choice a) in \eqref{choicetr} of the vectors $\frt^{(\Nr)}$,
we have found the following results for several low values of charges and $\kappa$:
\begin{subequations}
\bea
\vph^{(1)}_{1,1}(\tau,z)&=&\frac{7}{5971968} \, \frac{E_4(\tau) E_6(\tau)}{\Delta(\tau)}\,
\Bigl(E_4(\tau)\,\varphi_{0,1}+E_6(\tau)\,\varphi_{-2,1}\Bigr)-\frac{11}{3456} \, E_4(\tau)\,\varphi_{-2,1},
\\
\vph^{(2)}_{1,1}(\tau,z)&=&-\frac{E_4(\tau)}{2304} \,\varphi_{-2,1}\varphi_{0,1},
\\
\vph^{(3)}_{1,1}(\tau,z)&=& \frac{1}{9216}\,\Bigl(E_4(\tau)\,\varphi_{-2,1}\varphi_{0,1}^2
+2E_6(\tau)\,\varphi_{-2,1}^2 \varphi_{0,1} 
+E_4^2(\tau)\,\varphi_{-2,1}^3 \Bigr),
\\
\vph^{(4)}_{1,1}(\tau,z)&=& \vph^{(2)}_{2,2}(\tau,z)= 
\frac{1}{497664}\,\Bigl(-5 E_4(\tau)\,\varphi_{-2,1}\varphi_{0,1}^3
+14 E_6(\tau)\,\varphi_{-2,1}^2 \varphi_{0,1}^2 
\nn\\
&& 
+7 E_4^2(\tau)\,\varphi_{-2,1}^3\varphi_{0,1}
+2 E_4(\tau) E_6(\tau)\,\varphi_{-2,1}^4\Bigr) ,
\\
\vph^{(5)}_{1,1}(\tau,z)&=& \frac{1}{19906560}\,\Bigl(29 E_4(\tau)\,\varphi_{-2,1}\varphi_{0,1}^4
+130 E_6(\tau)\,\varphi_{-2,1}^2 \varphi_{0,1}^3 
+ 195 E_4^2(\tau)\,\varphi_{-2,1}^3\varphi_{0,1}^2
\nn\\
&& +116 E_4(\tau) E_6(\tau)\,\varphi_{-2,1}^4\varphi_{0,1}
+\Bigl(112 E_6^2(\tau)-90 E_4^3(\tau)\Bigr)\varphi_{-2,1}^5\Bigr) ,
\\
\vph^{(6)}_{1,1}(\tau,z)&=&\frac{1}{597196800}\,\Bigl(-125 E_4(\tau)\,\varphi_{-2,1}\varphi_{0,1}^5
+300 E_6(\tau)\,\varphi_{-2,1}^2 \varphi_{0,1}^4 
- \frac{3}{2}\, E_4^2(\tau)\,\varphi_{-2,1}^3\varphi_{0,1}^3
\nn\\
&& +149 E_4(\tau) E_6(\tau)\,\varphi_{-2,1}^4\varphi_{0,1}^2
-\( \frac{2529}{8}\,E_4^3(\tau)- 264 E_6^2(\tau)\)\varphi_{-2,1}^5\varphi_{0,1}
\nn\\
&& +\frac{63}{2}\, E_4^2(\tau)E_6(\tau)\,\varphi_{-2,1}^5\Bigr) ,
\\
\vph^{(2)}_{1,2}(\tau,z)&=& -\frac{1}{71663616}\,\Bigl(19 E_4(\tau)\,\varphi_{-2,1}\varphi_{0,1}^5
+30 E_6(\tau)\,\varphi_{-2,1}^2 \varphi_{0,1}^4 
+ 74 E_4^2(\tau)\,\varphi_{-2,1}^3\varphi_{0,1}^3
\nn\\
&& +148 E_4(\tau) E_6(\tau)\,\varphi_{-2,1}^4\varphi_{0,1}^2
+99 E_4^3(\tau)\,\varphi_{-2,1}^5\varphi_{0,1}+14 E_4^2(\tau)E_6(\tau)\,\varphi_{-2,1}^5\Bigr) ,
\eea
\label{res-diff-a}
\end{subequations}
while for the choice b) in \eqref{choicetr}, we have
\bea
\vph^{(6)}_{1,1}(\tau,z)&=&\frac{1}{1194393600}\,\Bigl(-1250 E_4(\tau)\,\varphi_{-2,1}\varphi_{0,1}^5
+12600 E_6(\tau)\,\varphi_{-2,1}^2 \varphi_{0,1}^4 
- 19487 E_4^2(\tau)\,\varphi_{-2,1}^3\varphi_{0,1}^3
\nn\\
&& +259558 E_4(\tau) E_6(\tau)\,\varphi_{-2,1}^4\varphi_{0,1}^2
-\( \frac{23409}{4}\,E_4^3(\tau)+851952 E_6^2(\tau)\)\varphi_{-2,1}^5\varphi_{0,1}
\nn\\
&& +9457 E_4^2(\tau)E_6(\tau)\,\varphi_{-2,1}^5\Bigr) .
\label{res-diff-b}
\eea
Here we used the standard Jacobi forms defined in \cite{Dabholkar:2012nd}
\be
\label{DMZ-phis}
\begin{split}
	\Delta(\tau)&=\eta^{24}(\tau)=\frac{E_4^3(\tau)-E_6^2(\tau)}{1728}\, ,
	\\
	\varphi_{-2,1}\(\tau,z\) &= \frac{\theta_1^2(\tau,z)}{\eta^6(\tau)},
	\\
	\varphi_{0,1}\(\tau,z\) &= 4\(\frac{\theta_2(\tau,z)^2}{\theta_2(\tau)^2}
	+\frac{\theta_3(\tau,z)^2}{\theta_3(\tau)^2}+\frac{\theta_4(\tau,z)^2}{\theta_4(\tau)^2} \),
\end{split}
\ee
and $\theta_2,\, \theta_3, \, \theta_4$ are the Jacobi theta functions.
All functions in \eqref{res-diff-a} and \eqref{res-diff-b} have the right weight and index, so that
the two solutions are indeed mutually consistent.

\subsection{Three unit charges}
\label{subsec-Jacobi3}

Here we compare the solution \eqref{res-unref3} for three charges specified to the case $\Nr_i=\kappa=1$
with the one given by the normalized generating function of $SU(3)$ VW invariants on $\IP^2$ \eqref{rel-tggP2-n}.
An explicit expression for this generating function can be found, e.g., in \cite[\S A]{Manschot:2017xcr}.
For convenience of the reader we copy it here. 
Let $T_3(\tau)=\vthA{3}_0(\tau,0)$ be the theta series associated with the $A_2$ lattice
and
\be
\label{VW3-summands}
\begin{split}
	S_{1,\mu}(k;\q) =&\,\frac{(-1+E_2(\tau))(k-\mu+1)}{2(1-\q^{3k-\mu})} 
	+ \frac{9(k-\mu)^2+33(k-\mu)+31-E_2(\tau)}{2(1-\q^{3k-\mu})^2}
	\\
	&\,
	-\frac{15(k-\mu)+34}{(1-\q^{3k-\mu})^3}+\frac{19}{(1-\q^{3k-\mu})^4}\, ,
	\\
	S_2(A,B;\q) =&\, \frac{4 \q^B}{(1-\q^A)(1-\q^B)^3}+\frac{4\q^A}{(1-\q^A)^3(1-\q^B)}+ \frac{4}{(1-\q^A)(1-\q^B)^2} 
	\\
	&\,
	-\frac{2(A+B+1)\q^B}{(1-\q^A)(1-\q^B)^2}-\frac{2(A+B+1)\q^A}{(1-\q^A)^2(1-\q^B)}+\frac{(A+B-2)^2-8}{2(1-\q^A)(1-\q^B)}.
\end{split}
\ee
Then the normalized generating function has the following components
\be
\label{VW-n=3}
\begin{split}
	\vwgi{3,0} =&\, \frac{1}{T_3(\tau)}\Biggl[ \frac{13}{240} + \frac{1}{24}\, E_2\(\tau\) 
	+ \frac{1}{72}\,E_2\(\tau\)^2 + \frac{1}{720}\,E_4\(\tau\) -\frac92 \sum_{k\in \IZ} k^2 \q^{3k^2}
	\\
	&\quad
	+ \frac16 \sum_{k_1,k_2\in \IZ} (k_1+2k_2)^2 \q^{k_1^2+k_2^2+k_1 k_2} + \sum_{\substack{k\in \IZ \\ k\neq 0}} S_{1,0}(k;\q) \q^{3k^2}  
	\\
	&\quad
	+ \!\!\!\!\!\!\!\!\!\sum_{k_1,k_2\in \IZ\atop 2k_1+k_2 \neq  0,\; k_2\neq k_1} \!\!\!\!\!\!\!\!\!
	S_2(2k_1+k_2,\,k_2-k_1; \q)
	\q^{k_1^2+k_2^2+k_1 k_2 + 2k_1 + k_2} 
	\Biggr],
	\\
	\vwgi{3,\pm 1} =&\, \frac{1}{T_3(\tau)}\Biggl[ \sum_{k\in\IZ} S_{1,1}(k;\q)\q^{3k^2-\frac13}
	+ \!\!\!\!\!\!\!\!\!\!\! \sum_{k_1,k_2\in \IZ\atop 2k_1+k_2\neq 1,\; k_2\neq k_1} \!\!\!\!\!\!\!\!\!
	S_2(2k_1+k_2-1,k_2-k_1;\q) \q^{k_1^2+k_2^2+k_1 k_2 - \frac13 } \Biggr].
\end{split}
\ee

In contrast to the case of two charges, the difference of two solutions for three charges should not be modular.
Instead, taking the difference of two anomaly equations, one finds that  
\be  
\delta\whgi{1,1,1}_{\mu}=\delta\gi{1,1,1}_{\mu}+2\sum_\nu \trmRi{1,2}_{\mu,\nu}\delta \gi{1,1}_{\nu},
\ee 
where $\delta \gi{\cdots}_{\mu}$ denotes the difference of two solutions. 
The non-holomorphic function $\trmRi{1,2}_{\mu,\nu}$ has a modular anomaly opposite to the one 
of the anomalous coefficient $\gi{1,2}_{\mu,\nu}$. This implies that
the function
\be
\vph^{(1)}_{1,1,1}(\tau,z;\frt)=
\sum_{\mu=0}^2\(\delta\gi{1,1,1}_{\mu}(\tau)-2\sum_{\nu=0}^1 \gi{1,2}_{\mu,\nu}(\tau)\delta \gi{1,1}_{\nu}(\tau)\)
\vthA{3}_\mu(\tau,z;\frt),
\ee 
where $\vthA{3}_\mu$ is the theta series for the $A_2$ lattice defined in \eqref{3thetas}, 
must be a Jacobi form of weight 4, index $\frt^2/2$ and trivial multiplier system, for any vector $\frt$.
In our case $\delta\gi{1,1,1}_\mu$ is the difference between  \eqref{res-unref3} and $\frac19\vwgi{3,\mu}$,
$\delta \gi{1,1}_\mu$ is the difference of \eqref{res-unref2} and $\frac13\, \vwgi{2,\mu}=H_\mu$,
and $\gi{1,2}_{\mu,\nu}$ is the solution \eqref{res-unref2} specialized to $(\Nr_1,\Nr_2)=(1,2)$ and $\kappa=1$. 
Using Mathematica, for the choice a) in \eqref{choicetr} of the vectors $\frt^{(\Nr)}$ and
for $\frt$ equal to any of the three vectors: $(1,-1,0)$, $(0,1,-1)$ or $(1,0,-1)$,
we have found that
\bea
\label{vph111}
&&\vph^{(1)}_{1,1,1}(\tau,z;\frt) =\frac{1}{156031757844480\Delta(\tau)^2}\Bigl[
\Bigl(34638 E_4(\tau)^6 E_6(\tau)
-82938 E_4(\tau)^3 E_6(\tau)^3
\\
&&\quad 
+21840 E_6(\tau)^5 \Bigr) \vph_{-2,1}
+\Bigl(22559 E_4(\tau)^7
-99247E_4(\tau)^4 E_6(\tau)^2
+40428 E_4(\tau) E_6(\tau)^4\Bigr) \vph_{0,1}\Bigr],
\nn
\eea
which is indeed a Jacobi form of the right weight and index.
We have checked that the same conclusion holds also for $\frt=(1,1,-2)$ and $\frt=(2,0,-2)$. 
These checks confirm the consistency of the two solutions.

\section{Expansion of anomalous coefficients for small charges}
\label{ap-expand}

In this appendix we provide explicit $\q$-series of the anomalous coefficients for a few sets of small charges.
For two charges, the result is presented in terms of the VV function $\Nr_0\chG^{(\kappa;\,\Nr_1,\Nr_2)}_{\mu_{12}}$
introduced in \eqref{chG-n=2} since it efficiently encodes independent components of the anomalous coefficients.
The index $\mu_{12}$ is related to the indices of $\gi{\Nr_1,\Nr_2}_{\mu,\mu_1,\mu_2}$ via \eqref{defmu0}
and takes $2\kappa_{12}$ values. However, the symmetry $\mu_{12}\to -\mu_{12}$ 
reduces the number of independent components to $\kappa_{12}+1$.
For three charges, we present directly the non-vanishing components of 
$\gi{\kappa;\,\Nr_1,\Nr_2,\Nr_3}_{\mu,\mu_1,\mu_2,\mu_3}$
where we added $\kappa$ to the set of charges to distinguish different cases.
\be
\begin{split}
	\chG^{(1;1,1)}_{0} =&\, \q^{-1}\(
	\tfrac{7}{497664}
	-\tfrac{7573}{82944}\,\q
	-\tfrac{11993 }{3456}\,\q^2
	-\tfrac{6147187}{15552} \,\q^3 
	-\tfrac{417892013 }{20736}\,\q^4
	-\tfrac{2669990303 }{4608}\,\q^5
	\right.	\\
	&\left.
	-\tfrac{3236466331}{288}\,\q^6
	-\tfrac{141840373163 }{864}\,\q^7
	-\tfrac{106915932005927}{55296} \,\q^8
	+\cdots\),
	\\
	\chG^{(1;1,1)}_{1} =&\, \q^{-1/4}\(
	\tfrac{247}{62208}\,
	+\tfrac{2441 }{2592}\,\q
	-\tfrac{685847 }{6912}\,\q^{2}
	-\tfrac{60354863}{7776}\,\q^{3}
	-\tfrac{1794183169 }{6912}\,\q^{4} 
	-\tfrac{4761308023 }{864}\,\q^{5}
	\right.
	\\&\left.
	-\tfrac{890009700749}{10368}\,\q^{6}
	-\tfrac{688179765559 }{648}\,\q^{7}
	-\tfrac{25273195362785  }{2304}\,\q^{8}
	+\cdots\),
\end{split}
\ee 
\bea
\chG^{(1;2,1)}_0 &=& \q^{-1}\(
\tfrac{23}{248832}
-\tfrac{588457}{2903040}\,\q
-\tfrac{4374197}{181440}\,\q^2
-\tfrac{1405439843}{435456}\,\q^3
-\tfrac{45080954077 }{241920}\,\q^4
-\tfrac{18356399752001}{2903040}\,\q^5
\right.\nn\\&&\left.
-\tfrac{326985343617523 }{2177280}\,\q^6
-\tfrac{1998756778032397}{725760}\,\q^7
-\tfrac{13340225185984769 }{322560}\,\q^8
+\cdots\),
\nn\\
\chG^{(1;2,1)}_1 &=& \q^{-13/12}\(
\tfrac{7}{559872}
+\tfrac{254521}{19595520}\,\q
+\tfrac{21208469}{19595520}\,\q^{2}
-\tfrac{2141274337 }{3265920}\,\q^3
-\tfrac{436455080941}{9797760}\,\q^4
-\tfrac{8043883306813 }{6531840}\,\q^5
\right.\nn\\&&\left.
-\tfrac{70596332554177}{3919104}\,\q^6
-\tfrac{289466108764979}{2799360}\,\q^7
+\tfrac{307669392722141}{181440}\,\q^8
+\cdots\),
\\
\chG^{(1;2,1)}_2 &=& \q^{-4/3}\(
\tfrac{5}{4478976}
+\tfrac{12343}{2612736}\,\q
+\tfrac{20174263}{39191040}\,\q^{2}
-\tfrac{22053267949 }{78382080}\,\q^{3}
-\tfrac{578517733229}{39191040}\,\q^{4}
-\tfrac{678670136857 }{2449440}\,\q^{5}
\right.\nn\\&&\left.
-\tfrac{223277357377}{1632960}\,\q^{6}
+\tfrac{754932817919599 }{6531840}\,\q^{7}
+\tfrac{2198602487769755}{653184}\,\q^{8}
+\cdots\),
\nn\\
\chG^{(1;2,1)}_3 &=& \q^{-3/4}\(
\tfrac{221}{725760}
+\tfrac{12463}{72576}\,\q
-\tfrac{37963483}{2177280}\,\q^{2}
-\tfrac{863218933 }{362880}\,\q^{3}
-\tfrac{3226760741}{20160}\,\q^{4}
-\tfrac{6663319581967 }{1088640}\,\q^{5}
\right.\nn\\&&\left.
-\tfrac{38738609259101}{241920}\,\q^{6}
-\tfrac{115345074567433 }{36288}\,\q^{7}
-\tfrac{110739444675550361}{2177280}\,\q^{8}
+\cdots\),\nn
\eea
\be
\begin{split}
	2\,\chG^{(1;2,2)}_0 =& \q^{-2}\(
	\tfrac{5}{60466176}
	+\tfrac{761731}{9876142080}\,\q
	-\tfrac{5964113}{20995200}\,\q^2
	-\tfrac{2481864301319}{74071065600}\,\q^3
	-\tfrac{2872361748037}{411505920}\,\q^4
	\right.\\&
	-\tfrac{2209512975931771}{2743372800}\,\q^5
	-\tfrac{41509779007262023 }{617258880}\,\q^6
	-\tfrac{49227475374711403543}{12345177600}\,\q^7
	\\&\left.
	-\tfrac{19464785233233206843 }{114307200}\,\q^8
	+\cdots\),
	\\
	2\,\chG^{(1;2,2)}_1 =& \q^{-9/8}\(
	\tfrac{277841}{8230118400}
	+\tfrac{1059416423}{74071065600}\,\q
	+\tfrac{74795669047}{24690355200}\,\q^2
	-\tfrac{1801948967597 }{1371686400}\,\q^3
	\right.\\&
	-\tfrac{16603156407869}{55987200}\,\q^4
	-\tfrac{416033291771649169 }{12345177600}\,\q^5
	-\tfrac{2105420048303835229}{914457600}\,\q^6
	\\&\left.
	-\tfrac{313599714564738170867}{2962842624}\,\q^7
	-\tfrac{88220836719674177597447}{24690355200}\,\q^8
	+\cdots\),
	\\
	2\,\chG^{(1;2,2)}_2 =& \q^{-1/2}\(
	\tfrac{39977803}{14814213120}
	+\tfrac{3589837921}{6172588800}\,\q
	+\tfrac{129752549}{5443200}\,\q^2
	-\tfrac{22102735798051 }{1543147200}\,\q^3
	\right.\\&
	-\tfrac{114275375877392837}{24690355200}\,\q^4
	-\tfrac{936409392998996809 }{2057529600}\,\q^5
	-\tfrac{4796069474036495173}{188956800}\,\q^6
	\\&\left.
	-\tfrac{1501044707018652100451 }{1543147200}\,\q^7
	-\tfrac{77342357414981744924971}{2743372800}\,\q^8
	+\cdots\),
\end{split}
\ee
\bea
\chG^{(1;3,1)}_0 &=& \q^{-2}\(
\tfrac{47}{318504960}
+\tfrac{4247129}{45984153600}\,\q
-\tfrac{68036588023}{183936614400}\,\q^2
-\tfrac{55422609269}{696729600}\,\q^3
-\tfrac{337075127071841}{26276659200}\, \q^4
\right.\nn\\&&\left.
-\tfrac{5507214757771}{4490640}\,\q^5
-\tfrac{7981897203311365099}{91968307200}\,\q^6
-\tfrac{1110685079693008997}{239500800}\,\q^7
-\tfrac{17226825559762052679301}{91968307200}\,\q^8
+\cdots\),
\nn\\
\chG^{(1;3,1)}_1 &=& q^{-25/24}\(
\tfrac{3981871}{55180984320}
+\tfrac{2648895629}{137952460800}\,\q
+\tfrac{41222942111 }{13138329600}\,\q^{2}
-\tfrac{285856288266173}{91968307200}\,\q^{3}
-\tfrac{2814997287076607}{5109350400}\,\q^{4}
\right.\nn\\&&
-\tfrac{1347624999852051079}{25082265600}\,\q^{5}
-\tfrac{935913250177725865183}{275904921600}\,\q^{6}
-\tfrac{20586251240400924448327}{137952460800}\,\q^{7}
\nn\\&&\left.
-\tfrac{270196785882281507783201}{55180984320}\,\q^{8}
+\cdots\),
\nn\\
\chG^{(1;3,1)}_2 &=& \q^{-7/6}\(
\tfrac{2971}{107775360}
+\tfrac{5544628433}{551809843200}\,\q
+\tfrac{115442040149 }{13795246080}\,\q^{2}
+\tfrac{23452076233427}{68976230400}\,\q^{3}
-\tfrac{61420667950261}{638668800}\,\q^{4}
\right.\nn\\&&
-\tfrac{5663089816965736121}{275904921600}\,\q^{5}
-\tfrac{9916439642198445779}{5748019200}\,\q^{6}
-\tfrac{2995144716151528141301}{34488115200}\,\q^{7}
\nn\\&&\left.
-\tfrac{37830819493668472285}{12317184}\,\q^{8}
+\cdots\),
\nn\\
\chG^{(1;3,1)}_3 &=& \q^{-11/8}\(
\tfrac{19}{2949120}
+\tfrac{738683593}{91968307200}\,\q
-\tfrac{2206279697}{2874009600}\,\q^2
-\tfrac{31913899958429}{91968307200}\,\q^3
-\tfrac{1052975439005939}{18393661440}\,\q^4
\right.\nn\\&&
-\tfrac{870826166392452223}{91968307200}\,\q^5
-\tfrac{24710897250831165317}{30656102400}\,\q^6
-\tfrac{3872950410038690100437}{91968307200}\,\q^7
\nn\\&&\left.
-\tfrac{17796943645723753462087}{11496038400}\,\q^8
+\cdots\),
\\
\chG^{(1;3,1)}_4 &=& \q^{-2/3}\(
\tfrac{398074393}{220723937280}
+\tfrac{7709928847}{19707494400}\,\q
-\tfrac{1124847101051}{110361968640}\,\q^2
-\tfrac{56937410997937}{3284582400}\,\q^3
-\tfrac{1198914724901166557}{367873228800}\,\q^4
\right.\nn\\&&
-\tfrac{3539458767899738927}{12541132800}\,\q^5
-\tfrac{8440035595252062897707}{551809843200}\,\q^6
-\tfrac{4516265260602369778967}{7664025600}\,\q^7
\nn\\&&\left.
-\tfrac{16911256238023424894329}{973209600}\,\q^8
+\cdots\),
\nn\\
\chG^{(1;3,1)}_5 &=& \q^{-25/24}\(
\tfrac{1507663}{55180984320 }
+\tfrac{8390729309}{137952460800}\,\q
+\tfrac{138045501791 }{13138329600}\,\q^2
-\tfrac{108275799214493}{91968307200}\,\q^3
-\tfrac{6633302879487181}{15328051200}\,\q^4
\right.\nn\\&&
-\tfrac{1253534988377466919}{25082265600}\,\q^5
-\tfrac{914235335203953688063}{275904921600}\,\q^6 
-\tfrac{20417088939750658199287}{137952460800}\,\q^7 
\nn\\&&   \left.
-\tfrac{269351350642042364835521}{55180984320}\,\q^8   
+\cdots\),
\nn\\
\chG^{(1;3,1)}_6 &=& \q^{-3/2}\(
\tfrac{11}{159252480}
+\tfrac{47213977}{22992076800 }\,\q
+\tfrac{24794030699\q^{1/2}}{22992076800}\,\q^2
-\tfrac{107364499601}{821145600}\,\q^3
-\tfrac{3501922956080341}{183936614400}\,\q^4
\right.\nn\\&&
-\tfrac{100084012971526063}{22992076800}\,\q^5
-\tfrac{286767755453058307}{656916480}\,\q^6
-\tfrac{52241700229979368001}{2090188800}\,\q^7
\nn\\&&\left.
-\tfrac{2134111773571338037187}{2189721600}\,\q^8
+\cdots\),
\nn
\eea
\be
\begin{split}
	\chG^{(2;1,1)}_0 =&\, 
	-\tfrac{17}{128}
	-\tfrac{113}{72}\,\q
	-\tfrac{4139 }{64}\,\q^2
	-\tfrac{4249 }{4}\,\q^3
	-\tfrac{3279635}{288}\,\q^4
	-\tfrac{370057 }{4}\,\q^5
	-\tfrac{9852609 }{16}\,\q^6
	\\&
	-\tfrac{63375745}{18}\,\q^7
	-\tfrac{1142163379 }{64}\,\q^8
	+\cdots,
	\\
	\chG^{(2;1,1)}_1 =&\, \q^{-1/8}\(
	\tfrac{1}{288}
	+\tfrac{47 }{32}\,\q
	+\tfrac{1369 }{32}\,\q^{2}
	+\tfrac{109447}{144}\,\q^{3}
	+\tfrac{274299 }{32}\,\q^{4}
	+\tfrac{575603 }{8}\,\q^{5}
	+\tfrac{15679485}{32}\,\q^{6}
	\right.\\&\left.
	+\tfrac{91227113 }{32}\,\q^{7}
	+\tfrac{234427805}{16}\,\q^{8}
	+\cdots\),
	\\
	\chG^{(2;1,1)}_2 =&\,\q^{-1/2}\(
	\tfrac{1}{2304}
	+\tfrac{5}{16}\,\q
	-\tfrac{281}{32}\,\q^{2}
	-\tfrac{9563}{36}\,\q^{3}
	-\tfrac{913121 }{256}\,\q^{4}
	-\tfrac{531815 }{16}\,\q^{5}
	-\tfrac{70096291}{288}\,\q^{6}
	\right.\\&\left.
	-\tfrac{5984707 }{4}\,\q^{7}
	-\tfrac{2056205767}{256}\,\q^{8}
	+\cdots\),
\end{split}
\ee
\be
\begin{split}
	\chG^{(2;2,1)}_{0} =&\, 
	-\tfrac{6775}{20736}
	-\tfrac{7157}{576}\,\q
	-\tfrac{361297}{768}\,\q^2
	-\tfrac{46696309}{5184}\,\q^3
	-\tfrac{600613969}{5184} \,\q^4
	-\tfrac{329780519}{288} \,\q^5
	\\
	&
	-\tfrac{12090114491}{1296} \,\q^6
	-\tfrac{169489988459}{2592} \,\q^7
	-\tfrac{467398870843}{1152} \,\q^8
	+\cdots,
	\\
	\chG^{(2;2,1)}_1 =&\, \q^{-1/24}\(
	\tfrac{65}{5184}
	+\tfrac{6353}{864} \,\q
	+\tfrac{1386517}{5184}\,\q^2
	+\tfrac{27324275}{5184}\,\q^3
	+\tfrac{87346705}{1296}\,\q^4
	+\tfrac{3384145273}{5184}\,\q^5
	\right.\\&\left.
	+\tfrac{26915024329}{5184}\,\q^6
	+\tfrac{46072770397 }{1296}\,\q^7
	+\tfrac{559857514087}{2592}\,\q^8
	+\cdots\),
	\\
	\chG^{(2;2,1)}_2 =&\, \q^{-1/6}\(
	\tfrac{7}{2592}
	-\tfrac{319}{2304} \,\q
	-\tfrac{178241}{5184}\,\q^2
	-\tfrac{793715}{1728} \,\q^3
	-\tfrac{470705}{2592}\,\q^4
	+\tfrac{5047214}{81}\,\q^5
	\right.\\&\left.
	+\tfrac{2610260809}{2592}\,\q^6
	+\tfrac{13190202073}{1296} \,\q^7
	+\tfrac{718155005}{9}\,\q^8
	+\cdots\),
	\\
	\chG^{(2;2,1)}_3 =&\, \q^{-3/8}\(
	\tfrac{11}{5184}\,
	+\tfrac{8129}{5184}\,\q
	+\tfrac{3797}{144}\,\q^2
	-\tfrac{84545}{5184}\,\q^3
	-\tfrac{44772229}{5184}\,\q^4
	-\tfrac{98195357}{576}\,\q^5
	\right.\\&\left.
	-\tfrac{10553330059}{5184}\,\q^6
	-\tfrac{94932443735}{5184} \,\q^7
	-\tfrac{78048945335}{576}\,\q^8
	+\cdots\),
	\\
	\chG^{(2;2,1)}_4 =&\, \q^{-2/3}\(
	\tfrac{7}{41472}
	+\tfrac{3869 }{10368}\,\q
	-\tfrac{77077}{3456}\,\q^2
	-\tfrac{652597}{1296} \,\q^3
	-\tfrac{2914033}{648}\,\q^4
	-\tfrac{81599053}{5184} \,\q^5
	\right.\\&\left.
	+\tfrac{2596870511}{20736}	\,\q^6
	+\tfrac{4497526727}{1728} \,\q^7
	+\tfrac{180081086629}{6912}\,\q^8
	+\cdots\),
	\\
	\chG^{(2;2,1)}_5 =&\, \q^{-1/24}\(
	\tfrac{281}{5184 }
	+\tfrac{5093}{864}\,\q
	+\tfrac{1358437}{5184}\,\q^2
	+\tfrac{27265307}{5184} \,\q^3
	+\tfrac{87314575}{1296}\,\q^4
	+\tfrac{3383929057}{5184} \,\q^5
	\right.
	\\&\left.
	+\tfrac{26914617169}{5184} \,\q^6
	+\tfrac{46072612447}{1296}\,\q^7
	+\tfrac{559856986399}{2592}\,\q^8
	+\cdots\),
	\\
	\chG^{(2;2,1)}_6 =&\, \q^{-1/2}\(
	\tfrac{55}{20736}
	+\tfrac{649}{576}\,\q
	-\tfrac{228775}{5184}\,\q^2
	-\tfrac{2115335 }{1296}\,\q^3
	-\tfrac{65662657}{2304}	\,\q^4
	-\tfrac{1760666641}{5184}\,\q^5
	\right.\\&\left.
	-\tfrac{16293026827}{5184}\,\q^6
	-\tfrac{3481650799}{144}\,\q^7
	-\tfrac{1113729358001}{6912}\,\q^8
	+\cdots\),
\end{split}
\ee
\bea
\chG^{(3;1,1)}_0 &=& 
-\tfrac{13}{96}
+\tfrac{71}{48}\,\q
+\tfrac{103 }{16}\,\q^2
+\tfrac{139 }{8} \,\q^3
+\tfrac{2243}{48}\,\q^4
+104 \q^5        
+\tfrac{3693 }{16}\,\q^6
+\tfrac{3667 }{8}\,\q^7
+\tfrac{14403}{16}\,\q^8
+\cdots,
\\
\chG^{(3;1,1)}_1 &=& \,\q^{-1/12}\(
-\tfrac{1}{64}
+\tfrac{27 }{32}\,\q
+\tfrac{211 }{64}\,\q^{2}
+\tfrac{139}{16}\,\q^{3}
+\tfrac{179 }{8}\,\q^{4}
+\tfrac{1585 }{32}\,\q^{5}
+\tfrac{6979}{64}\,\q^{6}
+\tfrac{6993 }{32}\,\q^{7}
+\tfrac{27263}{64}\,\q^{8}
+\cdots\),
\nn\\
\chG^{(3;1,1)}_2 &=& \,\q^{2/3}
\(-\tfrac{1}{32} 
-\tfrac{13}{16} \,\q
-\tfrac{157 }{32}\,\q^{2}
-\tfrac{233}{16} \,\q^{3}
-\tfrac{309}{8} \,\q^{4}
-\tfrac{677}{8}  \,\q^{5}
-\tfrac{2925}{16}\,\q^{6}
-355 \q^{7} 
-\tfrac{10867}{16}\,\q^{8}
+\cdots\),
\nn\\
\chG^{(3;1,1)}_3 &=& \,\q^{1/4}\(
-\tfrac{5 }{96}
-\tfrac{37 }{32}\,\q
-\tfrac{611}{96}\,\q^2
-\tfrac{641}{32}\,\q^3
-\tfrac{869}{16}\,\q^4
-\tfrac{3991}{32}\,\q^5
-\tfrac{25955}{96}\,\q^6
-\tfrac{4341}{8}\,\q^7
-\tfrac{16719}{16}\,\q^8
+\cdots\),\nn
\eea
\be
\begin{split}
	\chG^{(4;1,1)}_0 =& 
	-\tfrac{25891}{124416}
	-\tfrac{45463 }{5184}\,\q
	-\tfrac{852205 }{1296}          \,\q^2
	-\tfrac{32943085}{1296}         \,\q^3
	-\tfrac{13379780347 }{20736}    \,\q^4
	-\tfrac{46060126277}{3888}      \,\q^5
	\\&
	-\tfrac{12147774365 }{72}       \,\q^6
	-\tfrac{158959400096}{81}       \,\q^7
	-\tfrac{66913439964953 }{3456}  \,\q^8
	+\dots,
	\\
	\chG^{(4;1,1)}_1 =& \,\q^{-1/16}
	\(-\tfrac{37}{7776}
	+\tfrac{743 }{144}\,\q
	+\tfrac{186041}{432}\,\q^2
	+\tfrac{49307035 }{2592}\,\q^3
	+\tfrac{250278505}{486}\,\q^4
	+\tfrac{25390442633 }{2592}\,\q^5
	\right.\\&\left.
	+\tfrac{1109071859713}{7776}\,\q^6
	+\tfrac{2183125435445 }{1296}\,\q^7
	+\tfrac{21784639784605}{1296}\,\q^8
	+\dots\),
	\\
	\chG^{(4;1,1)}_2 =& \,\q^{-1/4}\(
	\tfrac{247}{62208}
	+\tfrac{2441 }{2592}\,\q
	-\tfrac{685847}{6912}\,\q^2
	-\tfrac{60354863 }{7776}\,\q^3
	-\tfrac{1794183169}{6912}\,\q^4
	-\tfrac{4761308023 }{864}\,\q^5
	\right.\\&\left.
	-\tfrac{890009700749}{10368}\,\q^6
	-\tfrac{688179765559 }{648}\,\q^7
	-\tfrac{25273195362785}{2304}\,\q^8
	+\dots\),
	\\
	\chG^{(4;1,1)}_3 =& \,\q^{-9/16}\(
	\tfrac{5}{7776 }
	+\tfrac{667}{2592}\,\q
	+\tfrac{9853}{7776}\,\q^2
	+\tfrac{4114013 }{2592}\,\q^3
	+\tfrac{208989733}{2592}\,\q^4
	+\tfrac{16164572585 }{7776}\,\q^5
	\right.\\&\left.
	+\tfrac{31419487859}{864}\,\q^6
	+\tfrac{78893525585 }{162}\,\q^7
	+\tfrac{41486136304309}{7776}\,\q^8
	+\dots\),
	\\
	\chG^{(4;1,1)}_4 =& \,\q^{-1}\(
	\tfrac{7}{248832}
	+\tfrac{793}{31104}\,\q
	+\tfrac{2371}{1296}\,\q^2
	-\tfrac{1033957}{7776}\,\q^3
	-\tfrac{51449111 }{3456}\,\q^4
	-\tfrac{2662533095}{5184}\,\q^5
	\right.\\&\left.
	-\tfrac{10331116165}{972}\,\q^6
	-\tfrac{68953726973 }{432}\,\q^7
	-\tfrac{157973370319477}{82944}\,\q^8
	+\dots\),
\end{split}
\ee
\bea
\chG^{(6;1,1)}_0 &=& 
-\tfrac{465535}{1769472}
-\tfrac{525963}{25600}\,\q
-\tfrac{4385292809}{1638400}\,\q^2
-\tfrac{43459913743}{230400}\,\q^3
-\tfrac{95409288460249}{11059200}\,\q^4
-\tfrac{716734134057}{2560}\,\q^5
\nn\\&&
-\tfrac{19031436950174249}{2764800}\,\q^6
-\tfrac{46566304858176491}{345600}\,\q^7
-\tfrac{357424842056971557}{163840}\,\q^8
+\dots
\nn\\
\chG^{(6;1,1)}_1 &=& \q^{-1/24}\(
-\tfrac{128371}{5529600}
+\tfrac{40650107}{3686400}\,\q
+\tfrac{2207272423}{1228800}\,\q^2
+\tfrac{799482612769}{5529600}\,\q^3
+\tfrac{15579841901093}{2211840}\,\q^4
\nn\right.\\&&\left.
+\tfrac{3410305584319}{14400}\,\q^5
+\tfrac{16460186000332049}{2764800}\,\q^6
+\tfrac{130847802327022001}{1105920}\,\q^7
+\tfrac{793328412584550589}{409600}\,\q^8
+\dots \)\nn
\\
\chG^{(6;1,1)}_2 &=& \q^{-1/6}\(
\tfrac{1052509}{176947200}
+\tfrac{209941}{138240}\,\q
-\tfrac{5204324219 }{11059200}\,\q^2
-\tfrac{2193309727}{34560}\,\q^3
-\tfrac{67263775905863}{17694720}\,\q^4
-\tfrac{49340829606167}{345600}\,\q^5
\nn\right.\\&&\left.
-\tfrac{1416482503757447}{368640}\,\q^6
-\tfrac{6909229953035321}{86400}\,\q^7
-\tfrac{239502419541351089947}{176947200}\,\q^8
+\dots \)
\nn\\
\chG^{(6;1,1)}_3 &=& \q^{-3/8}\(
\tfrac{3059}{1228800}
+\tfrac{1240819 }{2211840}\,\q
+\tfrac{6103987}{409600}\,\q^2
+\tfrac{40005301093}{2764800}\,\q^3
+\tfrac{14584022769029}{11059200}\,\q^4
+\tfrac{12408048600401}{204800}\,\q^5
\nn\right.\\&&\left.
+\tfrac{2031688897405031}{1105920}\,\q^6
+\tfrac{14301651369906251}{345600}\,\q^7
+\tfrac{30365831720828173}{40960}\,\q^8
+\dots\)
\\
\chG^{(6;1,1)}_4 &=& \q^{-2/3}\(
\tfrac{13157}{29491200 }
+\tfrac{202109}{1382400}\,\q
+\tfrac{117034913 }{22118400}\,\q^2
-\tfrac{13051427}{10800}\,\q^3
-\tfrac{51184196269 }{184320}\,\q^4
-\tfrac{4080532598711}{230400}\,\q^5
\nn\right.\\&&\left.
-\tfrac{28426573177934497}{44236800}\,\q^6
-\tfrac{748741181058329}{46080}\,\q^7
-\tfrac{13975691481395331721}{44236800}\,\q^8
+\dots\)
\nn\\
\chG^{(6;1,1)}_5 &=& \q^{-25/24}\(
\tfrac{19}{442368}
+\tfrac{204533}{11059200}\,\q
+\tfrac{12303371 }{11059200}\,\q^2
-\tfrac{176760667}{2764800}\,\q^3
+\tfrac{115535286821}{3686400}\,\q^4
+\tfrac{19874511607}{5760}\,\q^5
\nn\right.\\&&\left.
+\tfrac{1786765849460017}{11059200}\,\q^6
+\tfrac{52850388569884871}{11059200}\,\q^7
+\tfrac{76411388946994309}{737280}\,\q^8
+\dots\)
\nn\\
\chG^{(6;1,1)}_6 &=& \q^{-4/2}\(
\tfrac{11}{7077888}
+\tfrac{877}{691200}\,\q
+\tfrac{108537}{409600}\,\q^2
+\tfrac{1579361 }{172800}\,\q^3
-\tfrac{42904604779}{35389440}\,\q^4
-\tfrac{10934623253}{25600}\,\q^5
\nn\right.\\&&\left.
-\tfrac{315649065581971}{11059200}\,\q^6
-\tfrac{178841598736999}{172800}\,\q^7
-\tfrac{8452543304385033}{327680}\,\q^8
+\dots\),
\nn
\eea
\bea
\gi{1;1,1,1}_{0} &=&\, \q^{-1}\(
-\tfrac{307}{53084160}
+\tfrac{248189}{10616832}\,\q
+\tfrac{19861423}{8847360}\,\q^2
+\tfrac{53650187}{737280} \,\q^3
+\tfrac{1952990383}{8847360}\,\q^4
-\tfrac{1384556270653}{17694720} \,\q^5
\right.\nn\\
&&\left.
-\tfrac{5992150958009}{1769472}\,\q^6
-\tfrac{36863885478071 }{442368}\,\q^7
-\tfrac{428314732091503585}{286654464}\,\q^8
+\cdots\),
\\
\gi{1;1,1,1}_{1} &=&\, \q^{-4/3}\(
-\tfrac{67}{8599633920}
+\tfrac{621023}{4299816960 }\,\q
+\tfrac{145699391}{1074954240}\,\q^2
-\tfrac{9013190993}{859963392} \,\q^3
-\tfrac{4987262966201}{4299816960}\,\q^4
\right.\nn\\&&\left.
-\tfrac{29007521629433}{537477120}\,\q^5
-\tfrac{779386701680657}{537477120}\,\q^6
-\tfrac{56215493359702901}{2149908480}\,\q^7
-\tfrac{49271213570980177}{143327232}\,\q^8
+\cdots\),\nn
\eea
\bea 
\!\!\!\!\!\!\!\!\!\!\!\!
\gi{1;2,1,1}_{0,0} &=& \q^{-2}\(
-\tfrac{17}{6449725440}
-\tfrac{767269}{13544423424}\,\q
+\tfrac{112089181661}{1015831756800}\,\q^2
+\tfrac{1043317980281}{56435097600}\,\q^3
+\tfrac{27231789142567 }{48372940800}\,\q^4
\right.\nn\\&&\left.
-\tfrac{9482261170891493 }{253957939200}\,\q^5
-\tfrac{284324365489757}{59719680}\,\q^6
-\tfrac{11734312010643938927 }{42326323200}\,\q^7
-\tfrac{77318861918278372661}{6772211712}\,\q^8
+\dots\)
\nn\\
\!\!\!\!\!\!\!\!\!\!\!\!
\gi{1;2,1,1}_{0,1} &=& \q^{-3/4}\(
-\tfrac{131899}{9069926400}
+\tfrac{9311573}{220449600}\,\q
-\tfrac{3649155569 }{21163161600}\,\q^2
-\tfrac{1965367009399}{661348800}\,\q^3
-\tfrac{16199438866333 }{55112400}\,\q^4
\nn\right.\\&&\left.
-\tfrac{32549313903053 }{1889568}\,\q^5
-\tfrac{5459536184966606299}{7054387200}\,\q^6
-\tfrac{2077744772341594457 }{73483200}\,\q^7
-\tfrac{3598424331353910127481}{4232632320}\,\q^8
+\dots\)
\nn\\
\!\!\!\!\!\!\!\!\!\!\!\!
\gi{1;2,1,1}_{1,0} &=& \q^{-11/8}\(
-\tfrac{631}{1074954240}
+\tfrac{1410085567}{1015831756800}\,\q
+\tfrac{16470146659}{28217548800}\,\q^2
-\tfrac{19337664631871}{338610585600}\,\q^3
-\tfrac{8998136290013387 }{1015831756800}\,\q^4
\nn\right.\\&&\left.
-\tfrac{25781708703370891}{37623398400}\,\q^5
-\tfrac{2838148975641283621}{67722117120}\,\q^6
-\tfrac{2131628890388703071609}{1015831756800}\,\q^7
-\tfrac{4631806887062853201259}{56435097600}\,\q^8
+\dots \)
\nn\\
\!\!\!\!\!\!\!\!\!\!\!\!
\gi{1;2,1,1}_{1,1} &=& \q^{-9/8}\(
\tfrac{352537}{507915878400}
+\tfrac{113751559}{18811699200}\,\q
+\tfrac{34152804511}{24186470400}\,\q^2
-\tfrac{17760737685259}{42326323200}\,\q^3
-\tfrac{2380349828670577}{56435097600}\,\q^4
\right.\\&&\left.
-\tfrac{14512111029604723}{5643509760}\,\q^5
-\tfrac{67386472835220554599}{507915878400}\,\q^6
-\tfrac{327553586631623048293}{56435097600}\,\q^7
-\tfrac{183797329597763776319}{895795200}\,\q^8
+\dots \)
\nn\\
\!\!\!\!\!\!\!\!\!\!\!\!
\gi{1;2,1,1}_{2,0} &=& \q^{-3/2}\(
-\tfrac{253}{29023764480}
+\tfrac{4062503}{18811699200}\,\q
+\tfrac{197690483}{846526464}\,\q^2
-\tfrac{1374141989333}{31744742400}\,\q^3
-\tfrac{153753168177467 }{37623398400}\,\q^4
\right.\nn\\&&\left.
+\tfrac{566837660926451}{2687385600}\,\q^5
-\tfrac{846083215871836309}{63489484800}\,\q^6
-\tfrac{147305259874126081}{167961600}\,\q^7
-\tfrac{74177388668796200399}{1763596800}\,\q^8
+\dots \) 
\nn\\
\!\!\!\!\!\!\!\!\!\!\!\!
\gi{1;2,1,1}_{2,1} &=& \q^{-5/4}\(
-\tfrac{1}{1866240}
+\tfrac{24730109}{10581580800}\,\q
+\tfrac{2530558951 }{15872371200}\,\q^2
-\tfrac{510086582701	}{3527193600}\,\q^3
-\tfrac{160843599445 }{70543872}\,\q^4
\nn\right.\\&&\left.
+\tfrac{634132172310379}{10581580800}\,\q^5
-\tfrac{959317534461841}{58786560}\,\q^6
-\tfrac{1229196315272186183}{661348800}\,\q^7
-\tfrac{759174953988648536803}{7936185600}\,\q^8
+\dots\)
\nn
\eea
\be
\begin{split} 
\gi{2;1,1,1}_{0,0,0,0} =&\,
\tfrac{62959}{1658880}
+\tfrac{226289}{207360}\,\q
+\tfrac{9809837}{552960}\,\q^2
+\tfrac{13343}{69120}             \,\q^3
-\tfrac{2629441427 }{829440}\,\q^4
\\&
-\tfrac{20201197 }{360}\,\q^5
-\tfrac{85538444479}{138240}\,\q^6
-\tfrac{36356182001 }{6912}\,\q^7
-\tfrac{2585001140599}{69120}\,\q^8
+\cdots
\\
\gi{2;1,1,1}_{0,1,1,0} =&\, \q^{-1/2}
\(-\tfrac{73}{995328}
+\tfrac{6793}{829440}\,\q
+\tfrac{1989169 }{829440}\,\q^2
+\tfrac{24402883}{497664}\,\q^3
+\tfrac{759624757 }{1658880}\,\q^4
\right.\\&\left.
+\tfrac{662585173 }{829440}\,\q^5
-\tfrac{17919291181}{497664}\,\q^6
-\tfrac{501969016519 }{829440} \,\q^7
-\tfrac{3398716646521}{552960} \,\q^8
+\cdots\)
\\
\gi{2;1,1,1}_{1,1,0,0} =&\, \q^{-1/6}\(
\tfrac{101}{995328}
+\tfrac{824453 }{4976640}\,\q
+\tfrac{476039}{207360}\,\q^2
+\tfrac{54842903 }{829440}\,\q^3
+\tfrac{2494364191}{1658880}\,\q^4
\right.\\&\left.
+\tfrac{26617621331 }{1244160}\,\q^5
+\tfrac{110997670277}{497664}\,\q^6
+\tfrac{1163782557059 }{622080}\,\q^7
+\tfrac{33182088835409}{2488320}\,\q^8
+\cdots\)
\\
\gi{2;1,1,1}_{1,1,1,1} =&\, \q^{-2/3}\(
-\tfrac{23}{9953280}
+\tfrac{16607}{1244160}\,\q
-\tfrac{284833}{995328}\,\q^2
-\tfrac{4913885 }{124416}\,\q^3
-\tfrac{1967897171}{2488320}\,\q^4
\right.\\&\left.
-\tfrac{1152163793 }{124416}\,\q^5
-\tfrac{378680060593}{4976640}\,\q^6
-\tfrac{97855125631 }{207360}\,\q^7
-\tfrac{5549175976751}{2488320}\,\q^8
+\cdots\).
\end{split}
\ee

\section{Index of notations}
\label{sec_index}

In most of the paper we use boldface letters to denote vectors consisting of $n$ or $n-1$ components,
the blackboard script for the extended lattice as well as for generic $d$-dimensional lattices and their vectors,
the sans-serif script for glue vectors, 
and the mathfrak script for vectors with $n_k$ components resulting from the decomposition $n=\sum_{k=1}^m n_k$
and for the vectors $\frt^{(\Nr)}_\alpha$ determining the dependence on the refinement parameters.

\begin{longtable}{lp{10cm}l}
	Symbol & Description & ${\mbox{Appears or} \atop \mbox{defined in }}$ \\[5pt]
	\hline
	\rule{0pt}{20pt}
$\Asf=\{\asf_0,\asf_1,\dots,\asf_n\}$ & glue vector indices for the lattice decomposition \eqref{lat-glue-our} & \eqref{gluegA}\\
$\cAr_0 $ & set of glue vectors of the lattice decomposition \eqref{lat-glue-our} corresponding to 
the maximal order zero modes  & \eqref{cond-glueA} \\
$\beta=-\Im(z)/\tau_2$ & real variable parametrizing $z$ & \eqref{Rirf-to-rmRrf} \\
$\bbbeta= - \Im(\zbbm)/\tau_2$ & real vector parametrizing $\zbbm$  & \eqref{gentheta} \\
$b_2=b_2(\CY)$ & second Betti number of $\CY$  & p.\pageref{pb2} \\
$\Bsf=\{\bsf_{ij},\tbsf_{ij}\}$ & self-glue vector indices induced by the basis \eqref{recover-L} & 
below \eqref{expkparal}\!\!\!\!\!\!\! \\
$c_{2}$ & the second Chern class of $\CY$ & \eqref{qmax} \\
$\cbfr$ & coefficient of the leading pole of the modular ambiguity& \eqref{constphi}\\
$\Csf=\{\csf_k\}$ & glue vector indices for the lattice decomposition \eqref{lat-decomp-r}  & \eqref{setC} \\
$\delta^{(n)}_x$ & mod-$n$ Kronecker delta symbol & \eqref{defdelta} \\
$d_r$ & dimension of the lattice extension factor $\Zbbm^{d_r}$ associated to charge $r$ &\eqref{defdr} \\
$\bfDr, \ID^{(\bfr)}_{||}$& discriminant groups of the lattices $\bfLami{\bfr}$, $\bbLami{\bfr}_{||}$ & \S\ref{ap-mainlat}, \S\ref{subsubsec-gengen}\!\!\!\!\!\!\! \\
$\cD_m^{(n)}$ & differential operator on Jacobi forms & \eqref{defcDmn}\\
$\cD\vthA{N}_{\asf}(\tau;\frt)$& modular derivative of the $A_{N-1}$ lattice theta series &\eqref{def-cDvth}\\
$\Dzm(\Asf,\Bsf)$ & Kronecker symbol ensuring the maximal order zero mode condition & \eqref{theta-mzm} \\
$e_m$& weights of degenerate charge configurations &\eqref{def-genthm}\\
$\ebbm_0, \ebbm_{i,\alpha}$ & basis vectors of $\bbLami{\bfr}_\perp$ & \eqref{def-basise} \\
$E_n(\tau)$ & Eisenstein series & \eqref{modtr-E2}, \S \ref{ap-consist} \\
$E_n(\cM;\vu)$ & generalized error function on $\IR^n$ & \eqref{generr-E} \\
$\Ev_n=\Ef_n+\Ep_n$ & function encoding the modular completion & \eqref{twocEs}  \\
$\chphi^{(\bfr)}_{\mu,\bfmu}(\tau,z,\bfz)$ & modular ambiguity of anomalous coefficients & \eqref{def-phin}\\
$\phi^{(\bfr)}_{\bbnu}(\tau,z)$ & Jacobi-like form encoding the modular ambiguity & \eqref{def-phin}\\
$\Phi_n^E(\{\vbbm_i\};\xbbm)$ & boosted error function & \eqref{generrPhiME} \\
$\EPhi_n(\bfx)$ & kernel corresponding to the function $\Ev_n$ & \eqref{rescEn}  \\
$\rPhi(\xbbm;\{\Fvi{\bfs}\})$ & kernel of the theta series representing $\chgirf{\bfr}_{\mu, \bfmu }$ & \eqref{ker-thm}\\
$\Phidi{\bfr}(\xbbm,\tau,z)$ & kernel representing the modular ambiguity $\phi^{(\bfr)}_\bbnu$ & \eqref{Phidn}\\
$\Fvi{\bfr}(\xbbm)$& theta series kernel encoding $\rPhi$ & \eqref{kern-manyz}\\
$\gamma= (p^0,r,q,q_0)$ &charge vector of a generic D6-D4-D2-D0 bound state & p.\pageref{pgamma} \\
$\hgam=(r,q)$ & reduced charge vector & \eqref{defRn} \\
$\bfhgam=(\hgam_1,\dots,\hgam_n)$ & vector of reduced charges & \eqref{defRn} \\
$\gamma_{ij}$ & Dirac-Schwinger-Zwanziger product & \eqref{def-gammaij} \\
$\gi{\bfr}_{\mu,\bfmu}(\tau)$ & anomalous coefficient & Thm \ref{thm-ancoef} \\
$\girf{\bfr}_{\mu,\bfmu}(\tau,z)$ & refined anomalous coefficient & \eqref{refexp-whgi} \\
$\chgirf{\bfr}_{\mu,\bfmu}(\tau,z,\bfz)$ & auxiliary anomalous coefficient associated to the extended lattice & 
\eqref{extmodan-manyz} \\
$\whgi{\bfr}_{\mu,\bfmu},\whgirf{\bfr}_{\mu,\bfmu},\whchgirf{\bfr}_{\mu,\bfmu}$ & 
modular completions for various versions of anomalous coefficients &\eqref{compl-gi} \\
$\vwgi{n,\mu}(\tau)$ & normalized generating function of $SU(n)$ VW invariants on $\IP^2$ & \eqref{def-VWnorm} \\
$\glueg_\Asf$ & glue vectors for the lattice decomposition \eqref{lat-glue-our}& \eqref{gluegA}\\
$\Gi{\kappa}_\mu$, $\chG^{(\kappa;r_1,r_2)}_\mu$ & vectorial versions of the anomalous coefficient for two charges & 
\eqref{sol-n=2}, \eqref{chG-n=2}\!\!\!\!\!\!\! \\
$\cGi{d}_\mu(\tau)$   & seed functions for the mock modular forms of optimal growth & \eqref{DMZ} \\
$h_{r,\mu}(\tau) $ & generating function of D4-D2-D0 BPS indices (rank 0 DT invariants) & \eqref{defhDT} \\
$\whh_{r,\mu}(\tau)$ & modular completion of $h_{r,\mu}(\tau)$ & \eqref{exp-whh} \\
$\tlh_{r,\mu}(\tau)$ & redefined version of  $h_{r,\mu}(\tau)$ & \eqref{def-tgi} \\
$\than_{r,\mu}(\tau)$ & anomalous part of $\tlh_{r,\mu}(\tau)$ & \eqref{han} \\
$\thh_{r,\mu}(\tau)$ & modular ambiguity in $\tlh_{r,\mu}(\tau)$ & \eqref{han} \\
$\cIr_{\bbnu}(z)$ &maximal order zero mode contribution to $\vthls{\bfr}_{\bbmu,\Asf}$ & \eqref{theta-mzm}\\
$j_k$ & labels appearing in anomaly equations & \eqref{split-rs}\\
$\kappa$ & intersection number of CY  & \eqref{lvolatt} \\
$\kappa_i,\kappa_{ij},\kappa_{ijk}$ & $\kappa$ rescaled by magnetic charges $r_i$ & \eqref{defkapij} \\
$\kbbm$ & summation variable for theta series & \eqref{gentheta} \\
$\bfLami{\bfr}$ & lattice of D2-brane charges of constituents & \eqref{def-bfLam} \\
$\tbfLami{\bfr}$ & sublattice of $\bbLami{\bfr}_{||}$ isomorphic to $\bfLami{\bfr}$  & \eqref{defLampar} \\
$\bbLami{\bfr}$ & extended lattice & \eqref{extlatNr} \\
$\bbLami{\bfr}_{||}$ & sublattice of $\bbLami{\bfr}$ containing null-vectors & \eqref{defLampar} \\
$\bbLami{\bfr}_\perp$ & sublattice of $\bbLami{\bfr}$ orthogonal to $\bbLami{\bfr}_{||}$ & \eqref{factor-perp} \\
$\lambda_i$ & convenient combination of residue classes & \eqref{def-lami} \\
$\mu$ & residue class of $q$ modulo spectral flow & \eqref{defmu} \\
$\mu_{12}$ & effective residue class for two charges & \eqref{defmu0} \\
$\bfmu=(\mu_1,\dots,\mu_n)$ & vector of residue classes & \eqref{exp-whh}\\
$\bfhmu=(\hmu_1,\dots,\hmu_n)$ & vector of residue classes belonging to $\bfDr$ & \eqref{thetadata}\\
$\bfhmu^{(\bfr)}\bigl(\bfmu,\lambda \bigr)$& a different parametrization of $\bfhmu$ &\eqref{defhmu-all}\\
$\bbmu$ & vector of residue classes of the extended lattice & \eqref{thetadata}\\
$\Delta\mu$ & difference of residue classes & \eqref{constr-epsi} \\
$m_\bfr$& index of $\girf{\bfr}_{\mu,\bfmu}$& \eqref{index-mr}\\
$\frm_k$ & subset of residue classes appearing in anomaly equations & \eqref{split-rs} \\
$\nu_0(\Asf), \nu_{i,\alpha}(\Asf)$ & residue classes determined by glue vectors & \eqref{glui-extzi}\\
$\bbnu=(\bfhmu,\bfhtmu)$ & residue class of $\bbLami{\bfr}_{||}$ & below \eqref{def-phin}\!\!\!\!\!\!\! \\ 
$n_k$ & number of charges in the $k$-th subset  & \eqref{compl-gi} \\
$N_{ij}$ & range of the glue vector index $\bsf_{ij}$ &\eqref{valNij}\\
$\pbbm$& characteristic vector of the extended lattice &\eqref{thetadata}\\
$\q=e^{2\pi\I \tau}$ & expansion parameter of generating series & \eqref{defhDT} \\
$q$ & D2-brane charge & \eqref{defmu} \\
$\hat q_0$ & invariant D0-brane charge & \eqref{defqhat} \\
$\bfq=\bigl(\frac{q_1}{\kappa r_1},\dots,\frac{q_n}{\kappa r_n}\bigr)$ & vector of D2-brane charges & \eqref{rel-scpr-th}\\
$Q_n(\bfgam_i)$ & quadratic form in anomaly equations & \eqref{defQlr} \\
$\bfrho=(\rho_1,\dots,\rho_n)$ & vector of Bezout integers for $\bfr$ & above \eqref{gluegA}\!\!\!\!\!\!\! \\
$\bfrho^{(0)},\bfrho^{(k)}$& vector of Bezout integers for $\bfs$ and $\frr_k$ &\eqref{bfrho}\\
$r$ & total D4-brane charge & p.\pageref{pgamma} \\
$r_0$ & gcd of $(r_1,\dots, r_n)$ & below \eqref{defkapij}\!\!\!\!\!\!\! \\
$r_i$ & D4-brane charge of the $i$-th constituent & \eqref{exp-whh} \\
$\rdcr_i=\Nr_i/\Nr_0$ & normalized D4-brane charge &\eqref{Lam2}\\
$r_{ij}$& gcd of $(r_i, r_j)$ & \eqref{def-bfcvij-hat}\\
$\bfr=(r_1,\dots,r_n)$ & vector of D4-brane charges & \eqref{exp-whh}\\
$\frr_k$ & $k$-th subset of D4-brane charges appearing in anomaly equations & \eqref{split-rs} \\
$\rmRi{\bfr}_{\mu,\bfmu}(\tau,\btau)$ & coefficients of the anomaly equation \eqref{exp-whh} & \eqref{defRn} \\
$\trmRi{\bfr}_{\mu,\bfmu}(\tau,\btau)$ & redefined coefficients of the anomaly equation & \eqref{redefRn} \\
$\rmRirf{\bfr}_{\mu,\bfmu}(\tau,\btau,z)$ & coefficients of the refined anomaly equation & \eqref{Rirf-to-rmRrf} \\
$\scR_n(\bfhgam;\tau_2)$ & contribution to $\rmRi{\bfr}_{\mu,\bfmu}$ of a given charge decomposition & \eqref{solRn} \\
$\scRrf_n(\bfhgam;\tau_2,\beta)$ & contribution to $\rmRirf{\bfr}_{\mu,\bfmu}$ of a given charge decomposition & \eqref{refsolRn} \\
$\bfs=(s_1,\dots,s_m)$ & vector of D4-brane charges of subsets of constituents & \eqref{split-rs}  \\ 
$S_\cT(\bfhgam)$& large $\tau_2$ limit of the generalized error function $\Phi_{n}^E({\{\bfv_e\}};\bfx)$ 
defined by the tree $\cT$ &\eqref{defST}\\
$\!\Sym$& symmetrization with respect to charges &\eqref{defRn}\\
$\frt^{(\Nr_i)}$ & vector determining the dependance on $z_i$ & \eqref{choicetr}\\
$\theta_1(\tau,z)$ & Jacobi theta series & \eqref{free-theta-1} \\
$\ths{\kappa}_\mu(\tau)$ & unary theta series with even quadratic form & \eqref{deftheta}  \\
$\vth^{(d)}_{\nu_0}(\tau)$ & unary theta series appearing in \eqref{factor-perptheta} &  \eqref{3thetaZ}\\
$\vth_{\bbmu}(\tau, \zbbm;\bbLambda, \Phi, \pbbm)$ & indefinite theta series associated with the 
lattice $\bbLambda$ and kernel $\Phi$ & \eqref{gentheta} \\
$\cvths{\bfr}_{\bftnu(\Asf,\Bsf)}(z)$ & contribution to $\cIr_{\bbnu}$ for a fixed glue vector of index $\Bsf$  &  \eqref{cvths}\\
$\vthpzi{\bfr}{k}_{\bbmu,\Asf}(\tau)$& unrefined limit of the contribution of order $k$ zero modes& \eqref{deftheta-unref}\\
$\vthls{\bfr}_{\bbmu,\Asf}(\tau, z)$& theta series associated with the lattice $\bbLami{\bfr}_{||}$ &\eqref{defprojtheta}\\
$\vthps{\bfr}_\Asf(\tau,\bfz)$&theta series associated with the lattice $\bbLami{\bfr}_\perp$ &\eqref{defprojtheta}\\
$\vthA{N}_\asf(\tau,z;\frt)$ & theta series associated with $A_{N-1}$ root lattice & \eqref{3thetas}\\
$\bftet$ & vector of $\bfLami{\bfr}$ determining the index & \eqref{thetadata}\\
$\IT_n^{\rm S}$ & set of Schr\"oder trees with $n$ leaves & p.\pageref{pSchr} \\
$\hubbm_{ij}$& normalized vectors in $\tbfLami{\bfr}$ & below \eqref{defLampar}\!\!\!\!\!\!\! \\
$\bfv_{ij}$, $\bfv_e$, $\bfv_\ell$ & vectors in $\bfLami{\bfr}$ & \eqref{def-bfvij}, \S\ref{ap-E}\!\!\!\!\!\!\! \\
$\vbbm_{ij}$, $\vbbm_\ell$ & vectors in $\bbLami{\bfr}$ extending $\bfv_{ij}$, $\bfv_\ell$
& \eqref{def-bfcvij}, \eqref{def-genthm}\!\!\!\!\!\!\! \\
$\hvbbm_{ij}$ & normalized vectors in $\bbLami{\bfr}$ &\eqref{def-bfcvij-hat} \\
$\om_{ij}$ & scalar product of the lattice vector with null vectors & \eqref{def-bij}\\
$\wbbm_{ij}$ & null vectors in $\bbLami{\bfr}$ & \eqref{def-bfcvij} \\
$\hwbbm_{ij}$ & normalized null vectors &\eqref{def-bfcvij-hat} \\
$\chi_r$ & arithmetic genus of the divisor $\cD_r$ &  \eqref{defL0} \\
$\xbbm = \sqrt{2\tau_2} \,( \kbbm + \bbbeta)$ & $d$-dimensional vector, argument of kernels of theta series & \eqref{gentheta} \\
$\xbbm_\bbbeta = \sqrt{2\tau_2} \,\kbbm$ & shifted vector $\xbbm$ & above \eqref{rel-scpr}\!\!\!\!\!\!\! \\
$y=e^{2\pi i z}$ & exponential of the refinement parameter & \S \ref{subsec-ref} \\
$z$ & refinement parameter & \S \ref{subsec-ref} \\
$\bfz=(z_1,\dots,z_n)$ & refinement parameters associated to the extended lattice & \eqref{extmodan-manyz} \\
$\zbbm$ & vector of all refinement parameters & \eqref{thetadata} \\
$\Zv_n$& set of integers from $1$ to $n$ & \eqref{kern-manyz}\\
$\cdot$ , $\ast$ & bilinear forms on $\bfLami{\bfr}$ and $\bbLami{\bfr}$ & \eqref{bf-r}, \eqref{bb-r}\!\!\!\!\!\!\!
\end{longtable}

\providecommand{\href}[2]{#2}\begingroup\raggedright\endgroup


\end{document}